\documentclass[a4paper,twocolumn,11pt,accepted=2024-10-22
]{quantumarticle}
\pdfoutput=1
\usepackage[numbers,sort&compress]{natbib}
\usepackage{amsmath,amsfonts,amssymb,graphics,graphicx,epsfig,color,times,bbm}
\usepackage{amsthm}
\usepackage{enumitem}
\usepackage{psfrag}
\usepackage{braket}
\AtBeginDocument{\usepackage{booktabs}}
\graphicspath{{figures/}}
\usepackage{qcircuit}
\usepackage[hidelinks]{hyperref}
\hypersetup{colorlinks=false}
\usepackage{array}
\newcolumntype{P}[1]{>{\centering\arraybackslash}p{#1}}
\newcolumntype{M}[1]{>{\centering\arraybackslash}m{#1}}
\usepackage{graphicx}
\usepackage{mathtools}
\usepackage{soul}
\usepackage[normalem]{ulem}
\usepackage[T1]{fontenc}
\usepackage{hhline}

\usepackage{tikz}
\usetikzlibrary{arrows.meta,
                positioning,
                quotes}
                \usepackage{tkz-graph}

\newtheorem{theorem}{Theorem}

\newcommand{\tr}[1]{\mathrm{tr}({#1})}

\theoremstyle{plain}
\newtheorem{thm}{Theorem}

\newtheorem{lem}[thm]{Lemma}

\newtheorem{cor}[thm]{Corollary}

\theoremstyle{definition}

\newtheorem{rmk}[thm]{Remark}
\newtheorem{ex}{Example}

\renewcommand{\v}[1]{\boldsymbol{#1}}

\newcommand{\ketbra}[2]{\ket{#1}\!\!\bra{#2}}

\usepackage[framemethod=tikz]{mdframed}
\usepackage{hyperref}

\definecolor{warning_bgcol}{RGB}{252,248,229}
\definecolor{warning_textcol}{RGB}{111,89,54}
\definecolor{warning_linecol}{RGB}{252,248,229}
\definecolor{danger_bgcol}{RGB}{239,223,222}
\definecolor{danger_textcol}{RGB}{128,60,57}
\definecolor{danger_linecol}{RGB}{239,223,222}
\definecolor{success_bgcol}{RGB}{224,237,216}
\definecolor{success_textcol}{RGB}{68,104,60}
\definecolor{success_linecol}{RGB}{224,237,216}
\definecolor{info_bgcol}{RGB}{220,237,246}
\definecolor{info_textcol}{RGB}{58,100,126}
\definecolor{info_linecol}{RGB}{220,237,246}

\mdfdefinestyle{box_style}{
  skipabove=.7\baselineskip,
  skipbelow=.7\baselineskip,
  innertopmargin=.65\baselineskip,
  innerbottommargin=.65\baselineskip,
  innerleftmargin=.5\baselineskip,
  innerrightmargin=.5\baselineskip,
  splittopskip=1.5\baselineskip,
  splitbottomskip=\baselineskip,
  roundcorner=.3\baselineskip
}
\mdfdefinestyle{warning_style}{
  style=box_style,
  backgroundcolor=warning_bgcol,
  linecolor=warning_linecol,
  fontcolor=warning_textcol,
}
\mdfdefinestyle{success_style}{
  style=box_style,
  backgroundcolor=success_bgcol,
  linecolor=success_linecol,
  fontcolor=success_textcol,
}
\mdfdefinestyle{danger_style}{
  style=box_style,
  backgroundcolor=danger_bgcol,
  linecolor=danger_linecol,
  fontcolor=danger_textcol,
}
\mdfdefinestyle{info_style}{
  style=box_style,
  backgroundcolor=info_bgcol,
  linecolor=info_linecol,
  fontcolor=info_textcol,
}

\begin{document}

\title{The cost of solving linear differential equations on a quantum computer: fast-forwarding to explicit resource counts}

	\author{David Jennings}
	\affiliation{PsiQuantum, 700 Hansen Way, Palo Alto, CA 94304, USA}.
	\author{Matteo Lostaglio}
		\email{Lead author, mlostaglio@psiquantum.com}
	\affiliation{PsiQuantum, 700 Hansen Way, Palo Alto, CA 94304, USA}

\author{Robert B. Lowrie}
\affiliation{Computational Physics and Methods Group (CCS-2),
Computer, Computational and Statistical Sciences Division,
Los Alamos National Laboratory, Los Alamos, NM 87545, USA}
	
	\author{Sam Pallister}
	\affiliation{PsiQuantum, 700 Hansen Way, Palo Alto, CA 94304, USA}.
	\author{Andrew T. Sornborger}
	\affiliation{Information Sciences (CCS-3), Los Alamos National Laboratory, Los Alamos, New Mexico 87545, USA}
 

\begin{abstract}
How well can quantum computers simulate classical dynamical systems?
There is increasing effort in developing quantum algorithms to efficiently simulate dynamics beyond Hamiltonian simulation, but so far exact resource estimates are not known. In this work, we provide two significant contributions. First, 
we give the first non-asymptotic computation of the cost of encoding the solution to general linear ordinary differential equations into  quantum states -- either the solution at a final time, or an encoding of the whole history within a time interval. 
Second, we show that the stability properties of a large class of classical dynamics allow their fast-forwarding, making their quantum simulation much more time-efficient. From this point of view, quantum Hamiltonian dynamics is a boundary case that does not allow this form of stability-induced fast-forwarding. In particular, we find that the history state can always be output with complexity $O(T^{1/2})$ for any stable linear system. We present a range of asymptotic improvements over state-of-the-art in various regimes. We illustrate our results with a family of dynamics including linearized collisional plasma problems, coupled, damped, forced harmonic oscillators and dissipative nonlinear problems. In this case the scaling is quadratically improved, and leads to significant reductions in the query counts after inclusion of all relevant constant prefactors. 
\end{abstract}

\maketitle

\section{Introduction}
\label{sec:introduction}
\subsection{Setting the stage}
One of the original motivations for quantum computing was that simulating quantum dynamics efficiently requires quantum systems. However, a substantial fraction of today’s high-performance computing resources are used to simulate classical physics, such as fluid dynamics and high-temperature plasmas, which may be represented as dynamical systems of coupled degrees of freedom.\footnote{For example, the DOE INCITE 2023 projects allocated over $55 \times 10^6$ core-hours over 56 projects, of which over 40\% focused on plasma and fluid-dynamics problems, see https://www.doeleadershipcomputing.org/awardees/.
} It is then crucial to investigate to what extent these simulations may be sped up by quantum computing.

Assessing the potential of quantum computing in this computational space requires moving from the current asymptotic counts analysis to detailed resource estimates of quantum algorithms for differential equations. It also requires looking for asymptotic improvements to make up for the overheads of transforming the ODE into a problem suitable for a quantum computer. This is what we set out to do in this work, see~Fig.~\ref{fig:detailedODEcounts} for an illustration. We leverage recent non-asymptotic resource estimates for quantum linear solvers~\cite{costa2022optimal, jennings2023efficient} to compute for the first time general, rigorous upper bounds on the non-asymptotic cost of a quantum algorithm for solving dynamical problems beyond Hamiltonian simulation. By `solving' here we mean outputting a quantum encoding of the solution to the dynamical equations. The problem of loading data about initial and boundary conditions, and the problem of extracting information from the output, need to be carefully assessed on a case-by-case basis. 

What is more, we introduce analysis showing that the stability properties of a large class of classical dynamical systems give rise to costs \emph{sublinear} (at best, square-root) in the target simulation time. Our results are based on the seminal work in Ref.~\cite{berry2017quantum}. Its scope is extended and its scaling improved by a combination of a new condition number analysis and a range of algorithmic modifications.  

For example, for damped harmonic oscillators with constant forcing  we find in some regimes where the original analysis in Ref.~\cite{berry2017quantum} does not apply, the algorithm in Ref.~\cite{krovi2022improved} gives a scaling $\sim T^{5/2}$ when the simulation time $T$ is sufficiently prior to equilibration, our improved analysis of Ref.~\cite{berry2017quantum} gives a scaling $\sim T^{3/2}$, and a modification of the algorithm lowers this further to $\sim T^{1/2}$. In terms of query counts, we illustrate our improved analysis in the context of a natural class of classical dynamics, which includes the collisional linearized Vlasov-Poisson equation describing plasma equilibration as a special case. The asymptotic scaling with $T$ in this case is quadratically improved compared to state-of-the-art, and we find orders of magnitude improvement in the query count over best available analysis due to the sublinear scaling, see Fig.~\ref{fig:detailedODEcounts}.\footnote{Note that these are counts needed for a rigorous analytical guarantee o convergence, but  numerical evidence suggests much better performance is possible from the underlying linear solver~\cite{costa2023discrete}.}  The cost is cut by  up to $7$ orders of magnitude for $T \leq 10^{15}$, depending on the target simulation time and the equilibration rate. Hamiltonian dynamics is a boundary case in our stability analysis, hence preventing this form of fast-forwarding, in agreement with the no-fast-forwarding theorem. 

  \begin{figure}[t]
\includegraphics[width=\columnwidth]{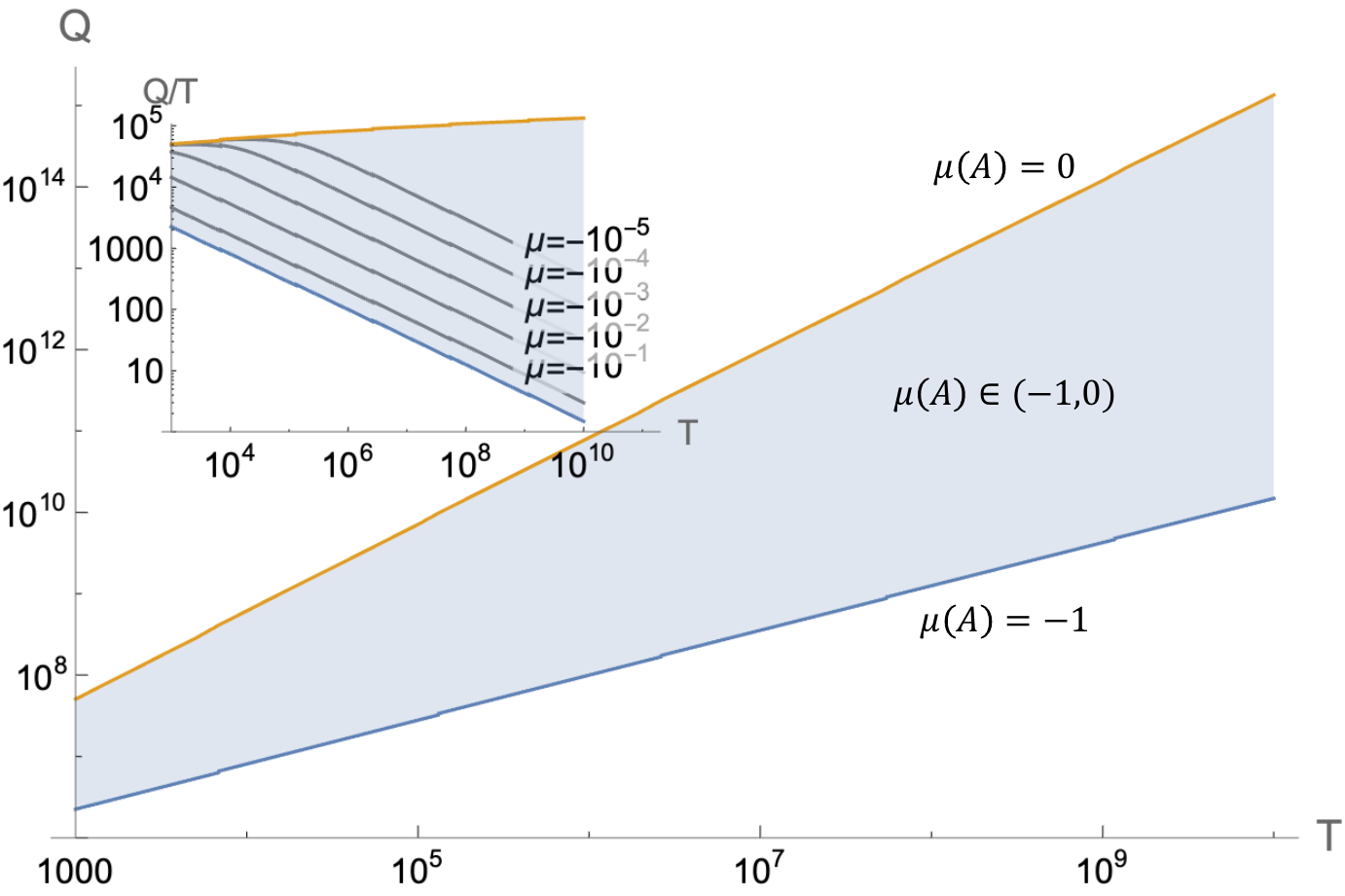}
 \caption{\textbf{Query counts for negative log-norm systems, plotted in log-log scale.} 
 We consider homogeneous ODE systems $\dot{\v{x}} = A \v{x}$ for which $A$ has a log-norm $\mu = \lambda_{\mathrm{max}}(A/2 + A^\dag/2)\leq 0$,  $\|A\|=h=1$. We fix the error tolerance in $1-$norm to $\epsilon=10^{-10}$ and set the block-encoding scaling in Eq.~\eqref{eq:UAblockencoding} to $\omega=1$. (Main figure): For all such systems, the number of times $\mathcal{Q}$ we need to query the block-encoding of $A$ to output the history state in Eq.~\eqref{eq:historystate} lies in the shaded region shown, along a curve that depends on $\mu \in [-1,0]$. The scaling in $T$ is sublinear for all $\mu<0$, in contrast to prior bounds, whose scaling is linear. Prior analysis would associate the upper curve bound to all such systems (in prior analysis~\cite{krovi2022improved}, $\mu \leq 0$ is only used to infer that $C_{\mathrm{max}}=1$ in Eq.~\eqref{eq:uniformbound}, which corresponds to the upper curve in the plot). Hence, the shaded area also corresponds to the maximum query counts saving from the present analysis. Note that at $T=10^{10}$ the cost can be cut by up to $90480\times$ by the sublinear behaviour. (Inset): Same figure for the cost per unit time $\mathcal{Q}/T$, and overlaid curves for different values $\mu$, showcasing the sublinear scaling. Note that while one can obtain $\mathcal{Q}/T = O(1)$ for the top curve ($\mu=0$), the slight increase is due to us using a QLSA with complexity $\kappa_L \log \kappa_L$~\cite{jennings2023efficient}. This is done so as to provide the smallest query counts once we include all constant prefactors.}
 \label{fig:detailedODEcounts}
	  \end{figure}

We study systems of linear ordinary differential equations (ODEs). Their main interest is as a component of a computational pipeline. Linear ODE systems are obtained after spatial discretization of linear partial differential equations (PDEs) -- e.g., the Schr\"odinger, heat, Fokker-Planck or linearized Vlasov  equations, among many others. Linear ODEs are also obtained after applying linear representation techniques such as Carleman linearization~\cite{liu2021efficient}, level-sets~\cite{jin2022quantum}, Koopman von Neumann/Liouville embeddings~\cite{koopman1931hamiltonian,kowalski1997nonlinear,joseph2020koopman,jin2023time, lin2022koopman,gonzalez2023mixed}, and perturbative homotopy methods~\cite{xue2021quantum} to nonlinear problems.

There are different approaches that could be employed to encode the solution of an ODE in a quantum computer. For example, one such method is to find a mapping to a Hamiltonian simulation problem~\cite{costa2019quantum, babbush2023exponential} or to a linear combination of Hamiltonian simulation problems~\cite{jin2022quantum2, an2023linear,apers2022quadratic, an2023quantum, an2023quantum2}.  Another approach is to encode the problem into a matrix equation of the form $L\v{y} = \v{c}$, which can be solved with a quantum linear solver algorithm~\cite{harrow2009quantum,berry2014high, berry2017quantum, costa2019quantum,childs2020quantum, linden2022quantum,berry2022quantum, krovi2022improved, bagherimehrab2023fast}. More recent developments include solving the system of equations via a time-marching algorithm~\cite{fang2023time} and Trotter methods~\cite{an2023quantum}.  Here we shall follow the most widely-studied approach of embedding the dynamical problem into a linear system of equations. This has been used to analyze the heat equation~\cite{linden2022quantum}, the wave equation~\cite{costa2019quantum}, the linearized Vlasov equation~\cite{ameri2023quantum} and nonlinear systems such as the Burgers' equation~\cite{liu2021efficient}, the reaction-diffusion equation~\cite{an2022efficient} and the lattice Boltzmann equation~\cite{li2023potential}, among others. Only for the wave equation numerical studies of the non-asymptotic complexity are available~\cite{suau2021practical}. Our work is significant in that it provides general, ready-to-use, analytical estimates that can be leveraged to evaluate the cost for all these use-cases, black-boxing the common components, with the extra advantage that our analysis obtains better scaling in the simulation time $T$ compared to state-of-the-art in certain regimes.

\subsection{Summary of the main results}

The main results that can be found in this manuscript are as follows:

\begin{enumerate}
    \item Ref.~\cite{krovi2022improved}
 presented an ODE algorithm improving over Ref.~\cite{berry2017quantum} by (a) extending its scope -- it drops the assumption that the generator of the dynamics is diagonalizable -- and (b) improving its scaling -- in some cases exponential speedups in the matrix dimension are achieved. We match and slightly outperform the same improvements within the original algorithm~\cite{berry2017quantum}, via a strengthened analysis of its performance. As an extra advantage, note that the original~\cite{berry2017quantum} algorithm is also simpler, as it avoids the nested matrix inversions of Ref.~\cite{krovi2022improved}. 
 \item State-of-the-art algorithms for ODEs scale linearly with time under some norm promises, barring  exceptional cases. Our strengthened condition number analysis introduces stability information in the estimation of the algorithm complexity, showing that, for all stable dynamics and under the same norm promises, the scaling is sublinear with time.  We present examples both in the context of linear and nonlinear systems of ODEs. We expect this stability analysis can be applied more widely to quantum algorithms for differential equations.
 \item We introduce an alternative version of Ref.~\cite{berry2017quantum}, whose complexity does not depend on the minimal solution norm if we output a quantum state encoding the full trajectory between time $0$ and some time $T$. This can lead to a polynomial speedup in $T$ for dissipative systems. We give an example of this phenomenon for damped harmonic oscillators sufficiently prior to equilibration. 
 \item We compute non-asymptotic query counts for the algorithm in Ref.~\cite{berry2017quantum} and its variation, under our improved condition number analysis, for general ODE dynamics. It is the first time that an algorithm for general ODEs is costed at this level of detail. We find that the previously mentioned asymptotic speedups lead to orders of magnitude improvements over state-of-the-art even after all constant prefactors are included.  Also, one can readily update our ODE cost analysis by plugging in further improvements to quantum linear solver algorithms worst-case estimates, as these are introduced.
\end{enumerate}

Let us now present the problem in more detail.

\subsection{The problem}

We are given a system of ODEs
\begin{equation}
\label{eq:ODE}
    \dot{\v{x}}(t) = A\v{x}(t) + \v{b}, \quad \mathrm{ with } \; \v{x}(0) = \v{x}^0,
\end{equation}
where $\v{x}(t) \in \mathbb{C}^N$ is the $N$--dimensional solution vector, \mbox{$A \in \mathbb{C}^{N\times N}$} is an $N \times N$ time-independent matrix which generates the dynamics, $\v{b} \in \mathbb{C}^N$ is an $N$-dimensional vector, also taken to be time-independent, which can be interpreted as a forcing term, and $t \in [0,T]$. We leave the extension to the time-dependent case to future work. When \mbox{$A = - iH$}, with $H$ Hermitian, and $\v{b} =\v{0}$, the problem in Eq.~\eqref{eq:ODE} reduces to quantum Hamiltonian dynamics. 

The time $t$ in the interval $[0,T]$ is discretized in equal steps of size \mbox{$h\leq1/\|A\|$}. More precisely, we consider the sequence of points $t=mh$ with $m\in \{ 0, \dots, M\}$, and assume that $T=Mh$ for some integer $M$.

In what follows we consider quantum states that encode dynamical vectors in a normalized form. For example the initial conditions $\v{x}^0 \in \mathbb{R}^N$ is associated with a normalized quantum state as follows
\begin{equation}
\ket{x^0} = \frac{1}{\|\v{x}^0\|} \sum_{n=0}^{N-1} x_n^0 \ket{n},
\end{equation}
and so the classical data is encoded in the computational basis up to an overall normalization constant. We adopt a similar notational convention for terms like $\ket{x(t)}$ (associated to the solution $\v{x}(t)$ to Eq.~\eqref{eq:ODE}), $\ket{b}$ (associated to $\v{b}$), and so on.

We now assume access to:
\begin{itemize}
  \item[O1] A \emph{block-encoding} of $A$, meaning a unitary $U_A$ with
\begin{equation}
\label{eq:UAblockencoding}
    U_{A} \ket{\psi} \ket{0^a} = \frac{A}{\omega} \ket{\psi} \ket{0^a} + \ket{\perp},
\end{equation}
where $\ket{0^a} := \ket{0}^{\otimes a}$, $\ket{\psi}$ is an arbitrary vector, $\omega$ is a rescaling factor required to encode $A$ into a unitary (we take $\omega \geq 1$), and \mbox{$(I \otimes \bra{0^a}) \ket{\perp} =0$}. The values $(\omega,a)$ determine the quality of the block-encoding and depend on the specific form of~$A$. Constructions for sparse and certain dense matrices achieving $\omega = \tilde{O}(\|A\|)$ can be found in Refs~\cite{lin2022lecture, camps2022explicit, sunderhauf2023block,nguyen2022block, li2023efficient}.
\item[O2] A unitary $U_0$ encoding the initial condition $\v{x}^0$ as a quantum state when applied to a reference state $\ket{0}$:
\begin{equation}
\label{eq:U0}
    U_0 \ket{0} = \ket{x^0}.
\end{equation} 

\item[O3] (If $\v{b} \neq \v{0}$): A unitary $U_b$ encoding the forcing $\v{b}$ up to normalization:
\begin{equation}
\label{eq:Ub}
    U_b \ket{0} = \ket{b}.
\end{equation}
\end{itemize}

Given the above, we distinguish two possible tasks. We want the algorithm to output a quantum state $\epsilon$-close in $1$-norm to either 

\begin{itemize}
    \item The \emph{solution state} \begin{equation}
\label{eq:finalstate}
   \ket{x_T} = \ket{x(T)},
\end{equation}
where $\v{x}(T)$ is the solution to Eq.~\eqref{eq:ODE} at time $T$. This is useful if, for instance, we want to study steady-state properties.
\item The \emph{history state}, which encodes a superposition of the solution along the entire trajectory at the discrete times $\{0, \dots, T\}$: 
\begin{equation}
\label{eq:historystate}
      \ket{x_H}:= \frac{1}{\|\v{x}_H\|} \sum_{m=0}^{M} \|\v{x}(mh)\| \ket{x(mh)} \otimes \ket{m},
\end{equation}
where $\ket{x(mh)}$ is a quantum state proportional to the solution of Eq.~\eqref{eq:ODE} at time $t=mh$, $\{\ket{m}\}_{m=0}^{M}$ are clock states that provide classical labels for the discrete times and $\v{x}_H = (\v{x}(0), \dots, \v{x}(T))$. This has not been studied in the same detail as the solution state problem, but history states are
potentially well-suited to time information analysis, e.g., time series analysis, spectral analysis, limit cycle analysis, etc. 

\end{itemize}

The cost of obtaining $\ket{x_T}$ or $\ket{x_H}$ is evaluated via the \emph{query count} parameter $\mathcal{Q}$, which gives the number of times we need to apply the above core unitaries. In particular we shall find that we require $\mathcal{Q}$ calls to the block-encoding of $U_A$ and $4\mathcal{Q}$ calls the state preparation unitaries. One of the main task of this work is to obtain stronger bounds on $\mathcal{Q}$ as a function of the properties of the dynamics.

\subsection{Incorporating classical Lyapunov theory}
\label{sec:classifying}

The problem under consideration is a wide generalization of quantum Hamiltonian simulation. Our aim is then to distill a few core parameters that allow us to compute the resource estimates for such quantum algorithms in this broad setting.
We find that stability analysis is particularly well-suited for this task (see also Appendix~\ref{sec:stabilitylyapunov}).

Every ODE system in Eq.~\eqref{eq:ODE} can be split according to a fundamental dichotomy, depending on the sign of the \emph{spectral abscissa}, 
which is the largest real part of the eigenvalues $\lambda_i$ of~$A$: \begin{equation}
    \alpha(A):= \max_i \mathrm{Re} \, \lambda_i(A).
\end{equation} 
Dynamical systems satisfying $\alpha(A)<0$ are called \emph{stable}, because small perturbations of $A$ cannot cause the norm of the solution of $\dot{\v{x}} = A \v{x}$ to shoot to infinity. Dynamical systems with $\alpha(A)>0$ are instead called \emph{unstable}, since they do exactly that for certain initial conditions, and are sometimes ruled out on physical grounds.

Stability theory tells us that for stable linear dynamics one can always find a positive matrix $P>0$ such that the \emph{Lyapunov inequality} is satisfied
 \begin{equation}
     P A + A^\dag P < 0. 
 \end{equation} 
Then the dynamical parameters of interest are the condition number $\kappa_P$ of $P$ and the $P$-log-norm of $A$:
 
 \begin{align}
 \label{eq:stabilityparameters}
	   \kappa_P& = \| P \| \|P^{-1}\| \nonumber \\
       \mu_P(A) &= \max_{\|\v{x}\| \ne 0} \mathrm{Re} \frac{  \langle A \v{x}, \v{x}\rangle_P }{\langle \v{x}, \v{x}\rangle_P},
	  \end{align}
   where $\langle \v{x}, \v{y} \rangle_P := \v{x}^\dag P \v{y}$ defines a valid inner product as $P>0$. The quantities $(\kappa_P, \mu_P(A))$ will be core parameters in determining the cost of the quantum algorithm. It can be proven (Theorem 3.35 in~\cite{plischke2005transient}) that the following inequality holds for all times $t >0$
	  \begin{equation}
	    \| e^{At} \| \leq \sqrt{\kappa_P} e^{\mu_P(A) t}.
	  \end{equation}
  It is worth emphasizing that this is not an additional assumption, but a direct consequence of the assumption of stability. Connecting with the vast literature on classical stability allows us to include physically relevant features of the dynamics in the computation of the resources required to run the quantum algorithm.
   
   For example, in the linearized collisional Vlasov-Poisson equation in plasma physics one finds that the Lyapunov inequality is satisfied with $P=I$, so $\kappa_P =1$. Also, $\mu_{P=I}(A)$ coincides with the so-called Euclidean log-norm of $A$,
   \begin{equation}
\mu_{P=I}(A) = \mu(A):= \lambda_{\mathrm{max}}(A+A^\dag)/2<0,   
\end{equation} 
where $\lambda_{\mathrm{max}}(X)$ is the largest eigenvalue of~$X$~\cite{soderlind2006logarithmic}. In fact, $\mu(A) <0$, with the numerical value depending on the strength of the collision term~\cite{ameri2023quantum}. So we have $(\kappa_P, \mu_P(A)) = (1, \mu(A)<0)$.

	  The case where \mbox{$\alpha(A) \geq 0$} includes systems for which the solution may grow arbitrarily large as we increase $T$, but it also contains cases where this does not occur. Notably, quantum Hamiltonian dynamics is a boundary case, for which $\alpha(A)=\alpha(-iH) = 0$. As in Ref.~\cite{krovi2022improved}, in the case $\alpha (A) \ge 0$ we shall bound the dynamics  with $C_{\mathrm{max}}(T)$ satisfying
 \begin{equation}
 \label{eq:uniformbound}
 \| e^{A t}\| \leq C_{\mathrm{max}}(T), \quad
 \end{equation}
 for all $t \in [0,T]$. Note that $C_{\mathrm{max}}(T)\equiv 1$ for quantum Hamiltonian dynamics.

In terms of the mathematical conditions, the $\alpha(A) \ge 0$ case can be understood as the formal limit $\mu_P(A) \rightarrow 0^-$ and $\kappa_P \rightarrow C_{\mathrm{max}}^2(T)$. This allows for a unified treatment, where we associate to every equation \eqref{eq:ODE} the two parameters $(\kappa_P, \mu_P(A))$, or a bound thereof, and obtain the classification sketched in Fig.~\ref{fig:classification}.

 \begin{figure}
     \centering
     \includegraphics[width=\linewidth]{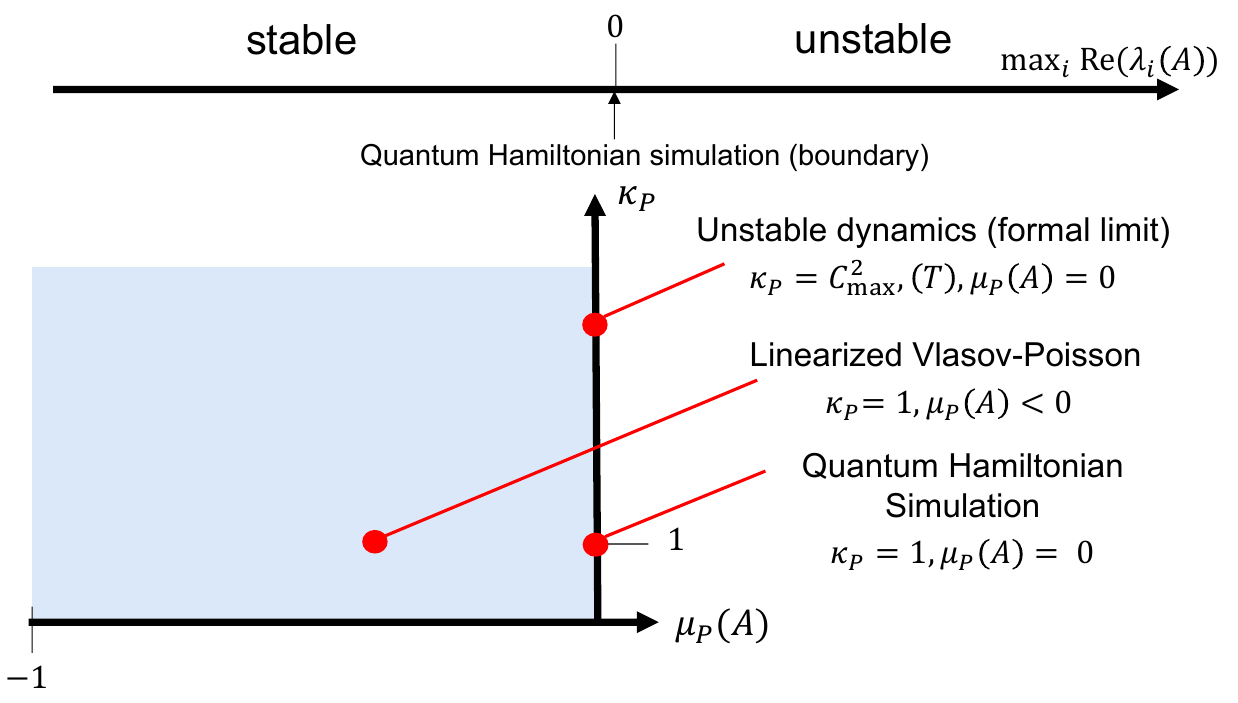}
     \caption{\textbf{Classifying dynamics.} (Top): A fundamental dichotomy is whether the generator of the dynamics $A$ is stable, when the largest real part of its eigenvalues is negative, or unstable, when it is positive. Quantum Hamiltonian dynamics is a boundary case between the two. (Bottom): To each dynamics we associate a point in the $2$-dimensional parameter space $(\kappa_P, \mu_P(A))$, given in Eq.~\eqref{eq:stabilityparameters}. These are fundamental parameters that will determine the cost of the algorithm.}
     \label{fig:classification}
 \end{figure}

 \subsection{Norm bounds and summary of relevant parameters}

 The classification picture we have just given needs to be augmented by further information due the fact that, differently from quantum dynamics, the norm of the solution can change from the initial to the final time. Depending on how we run the algorithm and the output we target, we need a certain set of norm bounds. The relevant quantities are bounds on the smallest solution norm, $ x_{\min} $, \emph{or} the largest solution norm, $ x_{\max} $ ($ x_{\min}  \leq \|\v{x}(t)\| \leq  x_{\max} $ for all $t \in [0,T]$), the root mean square of the solution norm, $ x_{\mathrm{rms}} $, and the ratio quantities $\bar{g}_{\times, +}(T)$. For their definition, and a summary of all potential input and user parameters, see Table~\ref{table:parameters}.

\begin{table*}[t]
\resizebox{\textwidth}{!}{
    \begin{tabular}{ | c | l | p{4.4cm} |}
    \hline
    \textbf{Parameter} & \textbf{Description}  \\ \hline $A$, $\v{b} $, $\v{x}^0$ & $A$ is the generator of the dynamics, $\v{b}$ a forcing term, $\v{x}^0$ the initial condition. \\ \hline
$T$ & Target simulation time, normalized in units of the time-step $h \leq 1/\|A\|$. 
    \\ \hline
    $\epsilon$ & Precision: output is a quantum state $\epsilon$-close in $1$-norm to the solution \eqref{eq:finalstate} or history~\eqref{eq:historystate} state.
    \\ \hline
    $\omega$ & Rescaling factor of the block-encoding of the generator $A$, Eq.~\eqref{eq:UAblockencoding}. 
    \\ \hline
    $\kappa_P$, $\mu_P$ & Stability / Lyapunov parameters: the generator $A$ satisfies $\|e^{At}\| \leq \sqrt{\kappa_P} e^{\mu_P(A)t}$ for $t \in [0,T]$. \\ \hline 
     $C_{\mathrm{max}}(T)$ & Upper bound on $\|e^{At}\|$, $C_{\mathrm{max}}(T) \geq \max_{t \in [0,T]} \| e^{At}\|$. \\ \hline 
 $ x_{\max} $, $ x_{\min} $ &  Upper and lower bound constants obeying:  $ x_{\min}  \le \|\v{x}(t)\| \le  x_{\max} $ for all $t \in [0,T]$. 
 \\ \hline 
 $ x_{\mathrm{rms}} $ &  Lower bound to the root mean square of the solution in $[0,T]$: $ x_{\mathrm{rms}}  \leq \sqrt{
\sum_{m=0}^M \| \v{x}(mh)\|^2/M}$. 
    \\ \hline  
    $\bar{g}_\times (T) $, $\bar{g}_+ (T) $ & Average norm ratios: \\ & $
        \bar{g}_\times (T)^2 = \frac{1}{M+1}\sum_{m=0}^M \frac{\| \v{x}(mh) \|^2}{  \|\v{x}(T)\|^2}
   $, \quad $
        \bar{g}_+(T)^2 = \frac{1}{M+1}\sum_{m=0}^M \left(\frac{1+\epsilon'}{1-\epsilon'}\right)^2 \frac{(\| \v{x}(mh) \|+\epsilon')^2}{  (\|\v{x}(T)\|-\epsilon')^2}
   $, \\
   &  with $\epsilon' = \epsilon x_{\mathrm{rms}} /8$ or $\epsilon' = \epsilon \|\v{x}(T)\|/8$ for history state or solution state respectively.\\ \hline
    \end{tabular}
   }
    \caption{\textbf{Summary of the algorithmic parameters and their description.} Note that Hamiltonian simulation is the special case $A= - iH$, $\v{b} = 0$, $\kappa_P =1$, $\mu_P = 0$, $C_{\mathrm{max}} =1$, $ x_{\min}  =  x_{\max}  =1$, $\bar{g}_\times (T) =1$, $\bar{g}_+ (T) =1$. Here the subscript $\times$ refers to the choice of multiplicative scheme, and the subscript $+$ refers to the choice of the additive scheme, as described in the main text.}
     \label{table:parameters}
\end{table*}

 \section{Results}

\subsection{Sketch of the algorithm}
\label{sec:sketch}

The solution of Eq.~\eqref{eq:ODE} can be formally written as 
\begin{equation}
\label{eq:ODEsolution}
    \v{x}(t) = e^{A t} \v{x}^0 + \int_{0}^t dt' e^{A t'} \v{b}.
\end{equation}
By discretizing time, we can recursively approximate the solution at time $t=mh$ advanced from that at time $t=(m-1)h$ via a Taylor series truncated at degree $k$:
\begin{equation}
\label{eq:recursive}
    \v{x}(mh) \approx \v{x}^m := T_k(Ah) \v{x}^{m-1} + S_k(Ah) \v{b}, 
\end{equation}
where 
\begin{align}
\label{eq:truncation}
    T_k(x) = \sum_{j=0}^k x^j/j!, \quad S_k(x) = \sum_{j=1}^k x^{j-1}/j!.
\end{align}
Within the algorithm we also extend the dynamics beyond time $T$ with $p$ steps of trivial `idling' dynamics(the trivial evolution $A=0$, $\v{b}=\v{0}$). This is done so as to increase the norm
  of the solution at the final time in the output state, as discussed in Ref.~\cite{berry2017quantum}. This will be useful if we are trying to output the solution state~$\ket{x(T)}$.

When we perform the time-discretization within the algorithm we have two choices of internal error schemes: a multiplicative (relative) error scheme and an additive error scheme. It turns out that this choice of internal error scheme can affect the overall complexity of the algorithm depending on the particular kind of dynamics. Firstly, we choose a multiplicative error condition as in \cite{berry2017quantum, krovi2022improved}, which corresponds to finding a truncation scale $k$ such that 
\begin{equation}
\label{eq:multiplicativeerror}
    \| \v{x}^m - \v{x}(mh) \| \leq \epsilon_{\mbox{\tiny TD}} \| \v{x}(mh)\|, \quad m=0, \dots, M.
\end{equation}
We shall label any algorithm parameter that depends on this multiplicative scheme, rather than the below additive scheme, by a subscript or superscript $\times$ for the multiplicative error scheme. For example, $k_\times$ denotes a Taylor truncation order $k$ that leads to Eq.~(\ref{eq:multiplicativeerror}). For this case, using Theorem 3 in Ref.~\cite{krovi2022improved}, we find an exact value for a truncation degree $k_\times$ that suffices, see Appendix~\ref{sec:timediscretizationerror}. In particular, $k_\times = O(\log(T/( x_{\min}  \epsilon_{\mbox{\tiny TD}}))$. 
 
We also consider an alternative additive error scheme such that
 \begin{equation}
\label{eq:additiveerror}
    \| \v{x}^m - \v{x}(mh) \| \leq \epsilon_{\mbox{\tiny TD}}, \quad m=0, \dots, M.
\end{equation}
We shall label any algorithm parameter that depends on the choice of the additive scheme by a subscript or superscript $+$. For example, the Taylor truncation that is used in Eq.~\eqref{eq:truncation} to realize Eq.~\eqref{eq:additiveerror} is denoted by $k_+$, and we show that a sufficient truncation is \mbox{$k_+ = O(\log(T  x_{\max}  /\epsilon_{\mbox{\tiny TD}}))$}. 

We would like to handle both the multiplicative and additive schemes throughout. To avoid major repetition, we shall label a parameter under either of the error schemes with the subscript or superscript notation ``$+,\times $''. For example, $k_{+,\times}$ denotes the Taylor truncation order under either the multiplicative scheme or the additive scheme.

We will set up a quantum algorithm solving a linear system whose output is a quantum state $\epsilon$-close to either the solution state~\eqref{eq:finalstate} or the history state~\eqref{eq:historystate}. Hence, setting a multiplicative error (default in previous quantum algorithms~\cite{berry2017quantum, krovi2022improved}) or an additive error truncation is an algorithmic choice granting us some extra flexibility to further lower the  final complexity. It is not a change of the problem statement.
    
 The recursive equations \eqref{eq:recursive} are embedded, following Ref.~\cite{berry2017quantum}, into a linear system $L \v{y} = \v{c}$, where $\v{y}= L^{-1} \v{c}$ encodes the solution to the dynamical system at all times $[0,T]$. The blocks of $L$ only contain matrices proportional to $A$ or the identity. This is convenient, since the block-encoding of $L$ and a unitary preparing the vector of constants $\v{c}$ (the base unitaries required for QLSA~\cite{costa2022optimal,jennings2023efficient}) can be constructed from the unitaries O1, O2, O3 without large overheads, see the schematics in Fig.~\ref{fig:algorithm_schematics}. Note that we may rescale $L \rightarrow \tilde{L} = \Lambda L$ and $\v{c} \rightarrow \tilde{\v{c}} = \Lambda \v{c}$, for any constant $\Lambda \ne 0$, without changing the solution vector~$\v{y}$. We do this rescaling so that $\|\tilde{L} \| \le 1$, which puts the linear system into a canonical form for quantum linear solver algorithms.

 \begin{figure}[h]
	 \includegraphics[width=\columnwidth]{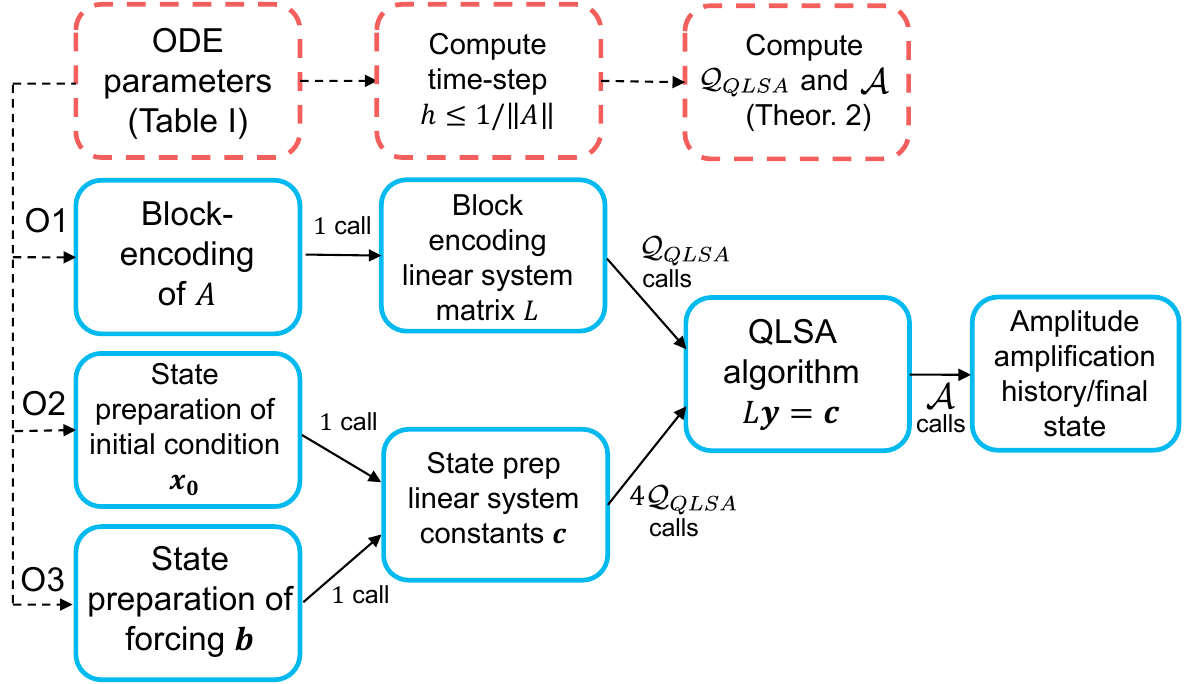}
 \caption{\textbf{Algorithm schematics.}  ODE algorithm schematics. (Red, dashed): classical precomputations. (Blue, solid): quantum operations. The algorithm reduces the ODE to a linear system of equations.  O1, O2, O3 are the unitaries in Eqs.~\eqref{eq:UAblockencoding}-\eqref{eq:Ub}. Our core analysis consists in determining the properties of the linear system (in particular, its condition number) as a function of the dynamical properties in Table~\ref{table:parameters}. Then the query count $\mathcal{Q}$ of the quantum ODE solver can be computed from the query cost $\mathcal{Q}_{QLSA}$ of the linear solver (given in Ref.~\cite{jennings2023efficient}) and the number $\mathcal{A}$ of amplitude amplification rounds, $\mathcal{Q} = \mathcal{A} \mathcal{Q}_{QLSA}$.}
 \label{fig:algorithm_schematics}
	  \end{figure}
 
 With these ingredients the linear system is solved via QLSA. The number of times the base unitaries need to be called in the QLSA algorithm is classically precomputed using the best non-asymptotic counts upper bounds in Ref.~\cite{jennings2023efficient}. Ref.~\cite{jennings2023efficient} gives an explicit formula upper bounding the query complexity of the QLSA, given 
 \begin{itemize}
     \item the rescaling prefactor of the linear system block-encoding, denoted by $\omega_{\tilde{L}}$, which is determined by the construction of the block-encoding of $L$
     and an upper bound on the norm $\|L\|$.
      \item the required QLSA precision $\epsilon_L$, which in our case needs to be chosen so that the overall error, including the time discretization errors $\epsilon_{\mbox{\tiny TD}}$, is below a target $\epsilon$.
      \item the condition number of $\kappa_L$ of $L$.
 \end{itemize} 

 The above counts guarantee convergence to the solution, but numerics suggests the quantum linear solver may well converge much faster than we can analytically prove~\cite{costa2023discrete}.
 
 Finally, the output solution is post-processed via amplitude amplification as needed, to get either a quantum state encoding of (an $\epsilon$-approximation of) the solution state~\eqref{eq:finalstate} or of the history state~\eqref{eq:historystate}. We compute a success probability which determines the number of amplitude amplification rounds required, strengthening previous bounds~\cite{berry2017quantum}.

The main technical contribution of this work is a novel analysis of the linear system solved by the algorithm of Ref.~\cite{berry2017quantum}. The improved analysis shows that this algorithm can be applied to a wider class of ODEs than previously known. In fact, similarly to the more recent algorithm in Ref.~\cite{krovi2022improved}, in the analysis of $L$ we will not assume that $A$ is diagonalizable. As described in the next section, we match and go beyond the improved asymptotics found in~\cite{krovi2022improved}, without introducing the more complex encoding of the new algorithm (involving a matrix inversion). In other words, we prove that the original algorithm of Ref.~\cite{berry2017quantum} has wider scope and beyond state-of-the-art performance. These improvements, discussed in the next section and summarized in Table~\ref{table:comparisonssimplified}, follow from our stability-based analysis and  the alternative error scheme. We expect both ideas to be of broader interest, since they can be applied in the context of other quantum algorithms for differential equations, such as that of Ref.~\cite{krovi2022improved}. This is left to future work.

\begin{table*}
      \begin{tabular}{ | p{0.7cm} | p{6.3cm}  | p{5.3cm} | p{2.39cm} |}
    \hline 
    \textbf{Ref.} & \textbf{Solution state} & \textbf{History state} & \textbf{Assumptions} \\ \hline
    \cite{berry2017quantum} & $\tilde{O}( \kappa_V g T \log \frac{T}{\epsilon  x_{\min} } \log\frac{g T }{\epsilon})$ & $\tilde{O}( \kappa_V T \log \frac{T}{\epsilon  x_{\min} } \log\frac{1}{\epsilon})$ &  $\alpha(A) \le 0$  \\ 
     \cite{krovi2022improved} & $\tilde{O}( \kappa_V g  T \log^{3/2} \frac{T}{\epsilon  x_{\min} } \log\frac{g T}{\epsilon})$ & $\tilde{O}( \kappa_V T \log^{3/2} \frac{T}{\epsilon  x_{\min} } \log\frac{1}{\epsilon})$ & $A$ diag. via~$V$ \\
    Here &  $O( \kappa_V \v{\bar{g}_\times} T \log \frac{T}{\epsilon  x_{\min} }\log\frac{\bar{g}_\times T}{\epsilon})$ &  $O( \kappa_V  T \log \frac{T}{\epsilon  x_{\min} } \log\frac{1}{\epsilon})$ &  \small $\kappa_V \! \! =\! \|V\| \|V^{-1}\|$ \normalsize  \\ 
    \hline \hline 
     
     \cite{berry2017quantum} & beyond scope & beyond scope &  $A$ stable, \\
     
     \cite{krovi2022improved} & $\tilde{O}( C_{\mathrm{max}} g  T \log^{3/2} \frac{T}{\epsilon  x_{\min} } \log\frac{g T}{\epsilon})$ & $\tilde{O}( C_{\mathrm{max}} T \log^{3/2} \frac{T}{\epsilon  x_{\min} }\log\frac{1}{\epsilon})$ &  i.e., $\alpha(A) <0$
     \\
    Here & $O(\v{\sqrt{\kappa_P} \bar{g}_\times   T^{3/4} \log \frac{T}{\epsilon  x_{\min} } \log\frac{\bar{g}_\times T} {\epsilon}})$ & $O(\v{\sqrt{\kappa_P}  \sqrt{T} \log \frac{T}{\epsilon  x_{\min} } \log\frac{1 }{\epsilon}})$ & 
    \\ 
       & $O(\v{\sqrt{\kappa_P} \bar{g}_+   T^{3/4} \log \frac{ T    }{\epsilon \| \v{x}(T)\|} \log \frac{\bar{g}_+ T}{\epsilon}})$ & $O(\v{\sqrt{\kappa_P}  \sqrt{T} \log \frac{ T   }{\epsilon} \log \frac{1}{\epsilon}})$ &
        \\ \hline 
    \cite{an2022theory} & $\tilde{O}( |\alpha(A)^{-1}| g \sqrt{T}  \log^2 \frac{1}{\epsilon})$ & - &
      $A <0$ \\ \hline \hline 
      \cite{berry2017quantum} & beyond scope & beyond scope &  $A$ general \\
\cite{krovi2022improved} & $\tilde{O}( C_{\mathrm{max}} g  T \log^{3/2} \frac{T}{\epsilon  x_{\min} } \log\frac{g T}{\epsilon})$ & $\tilde{O}( C_{\mathrm{max}} T \log^{3/2} \frac{T}{\epsilon  x_{\min} }\log\frac{1}{\epsilon})$ & 
     \\ 
Here &  $O  (C_{\mathrm{max}} \v{\bar{g}_{\times}} T\v{\log}\frac{T}{\epsilon  x_{\min} }\log \frac{\bar{g}_{\times}T}{\epsilon} )$ & $O (C_{\mathrm{max}}T \log \frac{1}{\epsilon} \frac{T}{ x_{\min} } \log\frac{1}{\epsilon} )$ & \\
&  $\v{O (C_{\mathrm{max}} \bar{g}_{+}  T\log \frac{ x_{\max}  T}{\| \v{x}(T)\|\epsilon }\log \frac{\bar{g}_{+} T}{\epsilon} )}$ & $\v{O (C_{\mathrm{max}}T \log \frac{ x_{\max}  T}{\epsilon}   \log\frac{1}{\epsilon} )}$   & \\
    \hline 
    \end{tabular}  
   \caption{\textbf{Asymptotic query complexities for linear ODEs with optimal quantum linear solvers when $\| \v{b} \| \neq 0$}. Comparison of asymptotic complexities of the algorithm in Ref.~\cite{berry2017quantum}, its new analysis (this work) and the recent improved algorithm in Ref.~\cite{krovi2022improved}.  The parameter $\omega$ enters linearly in all the above. Parameters are in Table~\ref{table:parameters} and $g(T) = \max_{t \in [0,T]} \| \v{x}(t)\|/\|\v{x}(T)\|$, so
    $\bar{g}_{\times} (T)\leq g(T)$. 
   Highlighted in bold are the improvements in complexity we obtained beyond state-of-the-art.The top block shows that we recover the results of~\cite{berry2017quantum}. The bottom block proves that we can apply the algorithm of~\cite{berry2017quantum} to the broader setting introduced in \cite{krovi2022improved}, and match the corresponding exponential improvements. The middle block is where our analysis introduces stability-induced fast-forwarding leading to sublinear scaling in $T$, and our additive error analysis giving further speedups
   (see applications  in Sec.~\ref{sec:applicationtoODEs}).  Also note that for the restricted scenario where $A$ is Hermitian and negative-definite, and we wish to output the solution state, then depending on the scaling of $g$ the algorithm presented in \cite{an2022theory} may be better suited. 
   A tighter, more general complexity table can be found in Table~\ref{table:comparisons} of Sec.~\ref{sec:fastforwarding}, from which we obtain the above by using $x_{{\rm rms}}\geq \|\v{x}(0)\|/\sqrt{T}$. \label{table:comparisonssimplified}} 
\end{table*}

 \subsection{Asymptotic complexity and fast-forwarding} \label{sec:asymptoticcomplexity} 
 
 In terms of scope, Ref.~\cite{berry2017quantum} assumes that $A$ has $\alpha (A) \le 0$ and is diagonalizable via a matrix $V$ whose condition number $\kappa_V$ is known. In this special case our framework returns $(\kappa_P, \mu_P) = (\kappa_V^2, 0)$ and hence recover these results (first block of Table~\ref{table:comparisons}). Our new analysis however can be deployed to generic settings (second and third block of Table~\ref{table:comparisons}). 

 In terms of asymptotic improvements, these are highlighted in bold in Table~\ref{table:comparisons}. If we drop the diagonalizability assumption and we only have a uniform bound as in Eq.~\eqref{eq:uniformbound}, we can take $(\kappa_P, \mu_P) = (C_{\mathrm{max}}^2(T),0)$ (third block of Table~\ref{table:comparisons}). This recovers the results in \cite{krovi2022improved} with slightly improved asymptotics and within the simpler encoding of Ref.~\cite{berry2017quantum}, whose linear system construction does not require matrix inversion~\cite{krovi2022improved} or exponentiation~\cite{berry2022quantum}.

However, it is for the broad class of stable systems, $\alpha(A)<0$, that we obtain significant improvements to the scaling with the target simulation time from our new analysis (second block of Table~\ref{table:comparisons}). In this case one has $(\kappa_P, \mu_P <0)$, which leads to a $\tilde{O}(T^{3/4})$ scaling for the solution state and a $\tilde{O}(T^{1/2})$ scaling for the history state, improving over $\tilde{O}(T)$. One might expect that simulating ODE systems could in general be harder than the Schr\"odinger equation, since one incurs various overheads in transforming the problem into a form suitable for a quantum computer. The above sharpens this relation by showing that quantum dynamics are at the boundary of stable dynamics, and that for the latter the condition number of the linear system associated to the ODE scales \emph{sublinearly} in the target simulation time~$T$, leading to sublinear in $T$ complexity:

 \begin{theorem}[Fast-forwarding of stable systems]
\label{thm:fast-forwarding}
    Let $\dot{\v{x}}(t) = A \v{x}(t) + \v{b}$ with $t\in[0,T]$, vectors $\v{x}(t),\v{b} \in \mathbb{C}^N$, and $A \in \mathbb{C}^{N\times N}$ an $N\times N$ matrix. Let the parameters $ \kappa_P$, $ x_{\max} $, $ x_{\min} $, and $\bar{g}_{+,\times} (T)$ be defined as in Table~\ref{table:parameters}. Assume that $A$ is a stable matrix, namely, the largest real part of its eigenvalues is negative. Then outputting a state $\epsilon$-close to $\ket{x_T}$ in Eq.~\eqref{eq:finalstate} can be realised with bounded probability with
    \begin{equation}
        \tilde{O}\left ( \omega \kappa_P^{1/2} \bar{g}_\times (T) T^{3/4} \log \frac{T}{  x_{\min}  \epsilon}\log \frac{1}{\epsilon} \right),
    \end{equation}
    or with
    \small
      \begin{equation}
       \tilde{O}\left ( \omega \kappa_P^{1/2} \bar{g}_+ (T)  T^{3/4} \log \frac{ T  x_{\max}  }{\epsilon \| \v{x}(T)\|} \log \frac{1}{\epsilon \| \v{x}(T)\|} \right),
    \end{equation}
    \normalsize
    queries to a unitary block-encoding of $A$ with scale-factor $\omega$, a unitary $U_0$ that prepares the normalized state~$\ket{x^0}$ and, if $\v{b} \neq \v{0}$, a unitary $U_b$ that prepares the normalized state~$\ket{b}$.

    In addition, outputting a state $\epsilon$-close to the history state $\ket{x_H}$ in Eq.~\eqref{eq:historystate} can be realised with bounded probability with
    \begin{equation}
        \tilde{O} \left ( \omega \kappa_P^{1/2}  T^{1/2} \log \frac{T}{ x_{\min}  \epsilon} \log\frac{1}{\epsilon} \right ),
    \end{equation}
    or with
    \small
     \begin{equation}
        \tilde{O} \left( \omega \kappa_P^{1/2}  T^{1/2} \log \frac{ T  x_{\max}  }{\epsilon} \log \frac{1}{\epsilon } \right)
    \end{equation}
    \normalsize
    queries to these unitaries.

\end{theorem}

See Appendix~\ref{sec:fastforwarding} for a precise statement and the proof. Speedups have been previously obtained for the case when $A=-H^2$ for some Hermitian $H$~\cite{an2022theory}, or when we access a unitary that diagonalizes $A$; but beyond these special cases, the complexity of prior ODE solvers is at least linear in $T$. Our speedup applies instead to all stable dynamics.

We note that stable dynamics tend asymptotically in time to a fixed point $A^{-1}\v{b}$. Therefore, for any constant error threshold $\epsilon$ there is a time $T_\epsilon$ beyond which one does not need to simulate dynamics, but can simply implement a linear solver algorithm to output $A^{-1}\ket{b}$ that is guaranteed to be within the fixed error threshold. This would be $O(1)$ in time~\cite{dong2024private}. The fast forwarding presented here therefore applies to the regime where $T<T_\epsilon$. We also highlight that the choice of $\epsilon$ is also crucial to the scaling in time. For example, we can take Hamiltonian dynamics and redefine it with a constant exponential damping envelope that gives a stable system. However this clearly does not give a sublinear complexity in time, violating the No-Fast Forwarding Theorem, since the required accuracy must scale exponentially with time. See Appendix~\ref{sec:No-FF} for further discussion.
\subsection{Applications: improved scaling for linear and nonlinear systems of ODEs}
\label{sec:applicationtoODEs}

The fast-forwarding regime refers to when we obtain a sublinear scaling in $T$. The general conditions can be easily worked out from Table~\ref{table:comparisons}. Let us present several examples where this occurs.

\begin{ex}[\textbf{Quadratic speedup for plasma system}] \label{ex:vlasovasymptotic} Consider the linearized collisional Vlasov-Poisson equation, which in Hermite representation has $\v{b}=\v{0}$, $\alpha(A)<0$, specifically it has negative log-norm~\cite{ameri2023quantum}. As discussed in Sec.~\ref{sec:classifying}, this implies $\kappa_P =1$. Also, the generator of the dynamics is symmetric, so $\kappa_V =1$. Then we obtain the results in Table~\ref{table:VP}.

\begin{table}[h!]
    \centering
    \begin{tabular}{|c|c|}
    \hline 
        Ref.~\cite{berry2017quantum}   & $\tilde{O}(T)$  \\ \hline
       Ref.~\cite{krovi2022improved}  & $\tilde{O}(T)$ \\ \hline
      This work    & $\tilde{O}(T^{1/2})$ \\ \hline
    \end{tabular}
\caption{Asymptotic query complexities for the linearized, collisional Vlasov-Poisson equation, for the history state output.    \label{table:VP} }

\end{table}
\end{ex}

Improvements can also be obtained when forcing is present:
\begin{ex}[\textbf{`Quintic speedup' for damped oscillator system before equilibration}]
Take a set of coupled classical harmonic oscillators with damping and constant forcing. Absorbing masses into the variables for simplicity, we have that,
\begin{align}
    \ddot{\v{y}} = - W \v{y} + D \dot{\v{y}} + \v{F},  
\end{align}
where $\v{F}$ is a constant force term, such as a uniform gravitation field. For a purely oscillatory system we have that $W \geq 0$ and $D$ is a damping matrix that is taken to have negative log-norm, which describes the situation of energy being dissipated by friction. Note that when $D$ is diagonal we obtain the familiar classical damping, but here we also allow off-diagonal terms, also known as non-classical damping, which is required in some structural mechanics applications~\cite{gupta2017significance, mcfarland1990sources, roehner2009influence}. Let $R$ be any matrix such that $W = R R^\dag$. Let us define $\v{x}(t):= [\dot{\v{y}}(t), i R^\dag \v{y}(t)]^T$. We then obtain the system of equations
\begin{align*}
    \dot{\v{x}}(t) = A \v{x}(t) + \v{b}, \quad A = 
    \begin{bmatrix}
        D & - i R \\
        i R^\dag & 0
    \end{bmatrix}, \v{b} =
    \begin{bmatrix}
        \v{F} \\ \v{0}
    \end{bmatrix}.
\end{align*}
Note that $A = - iH + \mathrm{diag}(D,0)$, with $H$ Hermitian. So for $D=0$, $\v{F}=0$, we recover the setup in Ref.~\cite{babbush2023exponential}. Strategies to block-encode $H$ when $W$ encodes sparse couplings among the oscillators are discussed in~\cite{babbush2023exponential}. If $D$ can be efficiently block-encoded, e.g., it is sparse and locations and values of the nonzero entries are efficiently computable, then $A$ can be efficiently block-encoded via LCU of the block-encoding of $D$ and $H$. We see that $\mu(A) = \mu(D):= \mu <0$, so $A$ is a stable matrix with negative log-norm. 
 
For finite evolution times, we can use Gronwall's lemma to upper bound the solution norm as
\begin{align}
   \| \v{x}(t)\| \leq e^{-|\mu| t} \|\v{x}(0)\| + (1- e^{-|\mu| t} )\frac{\| \v{b}\|}{|\mu|}. 
\end{align}
Suppose that the initial condition satisfies \mbox{$\| \v{x}(0)\| \gg \|\v{b}\|/|\mu| = \| \v{F}\|/|\mu|$} and that we are interested in times $T$ before `equilibration time', \mbox{$T \leq 1/|\mu|$}. Then, for all \mbox{$t \in [0,T]$} and all \mbox{$T \leq 1/|\mu|$}, the decay of the solution norm is approximately exponential, \mbox{$\| \v{x}(t)\| \lessapprox e^{-|\mu| t} \|\v{x}(0)\|$}. We assume that we wish to output the history state for the system. In this regime, \mbox{$ x_{\min} ^{-1} = O(e^{T})$}. Also, \mbox{$ x_{\max}  = O(1)$}, $ x_{\mathrm{rms}}  = O(1/T)$. From Table~\ref{table:comparisons} we see that prior algorithms scale as $O(T^{5/2})$ in this regime. With our stability-based improved condition number and using the additive error form of the algorithm discussed in Sec.~\ref{sec:sketch}, we obtain a scaling $O(T^{1/2})$: 
\begin{table}[h!]
    \centering
    \begin{tabular}{|c|c|}
    \hline 
        Ref.~\cite{berry2017quantum}   &  \textrm{beyond scope}  \\ \hline
       Ref.~\cite{krovi2022improved}  & $\sim T^{5/2}$ \\ \hline
      This work ($\times$)   & $\sim T^{3/2}$ \\ \hline
      This work ($+$)   & $\sim T^{1/2}$ \\ \hline
    \end{tabular}
    \caption{Query complexities for classical damped harmonic oscillators with constant forcing, for $T$ smaller than the equilibration time, history state output. The labels ($\times$) and ($+$) denote two different versions of the algorithm.}
    \label{table:HO}
\end{table}
\end{ex}
So far we discussed improvements to algorithms for linear differential equations. However, as we discussed in the introduction, a technique often used in quantum algorithms consists of embedding nonlinear systems into a larger linear problem. This suggests that our results can lead to improvements also in quantum algorithms for nonlinear differential equations. The following example shows that this is indeed the case.

\begin{ex}[\textbf{Sublinear time scaling for nonlinear ODEs}]

Consider a system of quadratic ODEs:
\begin{align}
    \dot{\v{u}} = F_2 \v{u}^{\otimes 2} + F_1 \v{u} + \v{F}_0,
\end{align}
where $\v{u}$ and $\v{F}_0$ are $N$-dimensional vectors, $F_1$ a $N \times N$ matrix and $F_2$ a $N \times N^2$ matrix. The problem is rescaled so that $\| \v{u}(0)\| <1$. We only assume $\mu(F_1)<0$, i.e., $F_1$ has negative log-norm. This is the setup of Ref.~\cite{krovi2022improved}.

One defines the Reynold's-like quantity
\begin{align}
    R = \frac{1}{|\mu(F_1)|} \left( \frac{\|F_2\|}{\|\v{u}(0)\|} + \|\v{F}_0\| \|\v{u}_0\| \right). 
\end{align}
    Under the condition that $R<1$, the above problem can be well-approximated by~\cite{liu2021efficient, krovi2022improved}
    \begin{align}
        \dot{\v{x}} = A \v{x} + \v{b},
    \end{align}
   where $\v{x} = [ \v{u}, \v{u}^{\otimes 2}, \v{u}^{\otimes 3}, \dots, \v{u}^{\otimes N_{\mathrm{tr}}}]$,  for a truncation $N_{\mathrm{tr}}$ that scales logarithmically with the simulation time and the inverse target error. The matrix $A$ can be efficiently block-encoded and $\v{b}$ can be efficiently prepared, assuming the same properties hold for $F_2$, $F_1$, $\v{F}_0$ and $\v{u}(0)$. Furthermore, $\mu(A)<0$ (proof of Lemma~16 in Ref.~\cite{krovi2022improved}), so $\kappa_P =1$. The state-of-the-art algorithm to obtain the history state from this problem scales asymptotically as $O(T)$, while the effective scaling before equilibration may be higher, as we have seen in the previous example. One can prove that $\| \v{u}(t)\| < 1$ for all $t>0$, so $ x_{\max}  = O(1)$ and $ x_{\mathrm{rms}}  = \tilde{O}(1/T)$.  Using our results in additive form (Table~\ref{table:comparisons}, second block, last row), we find a scaling to extract the history state of $\tilde{O}(T^{1/2})$. We can similarly improve the complexity scaling in the extraction the solution state, which is linear in prior state-of-the-art (Theorem~8 of Ref.~\cite{krovi2022improved}), while it is $\tilde{O}(T^{3/4})$ with using our results in the multiplicative form (Table~\ref{table:comparisons}, second block, second to last row). 

\begin{table}[h!]
    \centering
    \begin{tabular}{|c|c|c|}
    \hline 
    Output & History state & Solution state \\
    \hline \hline 
        Ref.~\cite{berry2017quantum, an2022theory}   &  beyond scope & beyond scope \\ \hline
       Ref.~\cite{krovi2022improved}  & $\tilde{O}(T)$ & $\tilde{O}(T)$ \\ \hline
      This work & $\tilde{O}(T^{1/2})$  & $\tilde{O}(T^{3/4})$ \\ \hline
    \end{tabular}
    \caption{Asymptotic time complexities for nonlinear systems of quadratic ODEs with negative log-norm and $R<1$.}
\end{table}
\end{ex}

 \subsection{Detailed query counts of the ODE-solver} 
 \label{sec:detailedcounts}
 So far we have discussed asymptotics, but in fact we can compute analytic query complexity upper bounds. Since the final expression is cumbersome, we describe it via a sequence of straightforward replacements. Recall that $+, \times$ refer to two versions of the algorithm, and one can take the best counts among the two. 
\begin{theorem}[Explicit query counts for ODE-solver]
\label{thm:detailedODEcounts}
Consider any $N$--dimensional ODE system of the form
\begin{equation}
\dot{\v{x}} = A \v{x} + \v{b},\quad \v{x}(0) = \v{x}^0,
\end{equation}
 specified by a complex matrix $A\in \mathbb{C}^{N\times N}$, vectors $\v{b}, \v{x}^0 \in \mathbb{C}^N$ and $t\in [0,T]$. Choose a timestep $h$ such that $\|A\|h \le 1$ and $Mh=T$ for integer $M$. Let $(\kappa_P, \mu_P(A),  x_{\max} ,  x_{\min} , x_{\mathrm{rms}} )$ be the associated ODE parameters, as defined in Table~\ref{table:parameters}. 

We assume that we have oracle access to a unitary $U_A$, which is an $(\omega, a, 0)$ block-encoding of $A$, a unitary $U_0$ that prepares the normalized state $\ket{x^0}$ and, if $\v{b} \neq \v{0}$, a unitary $U_b$ that prepares the normalized state $\ket{b}$. Then a quantum state $\epsilon$-close in $1$-norm to the solution state $\ket{x_T}=\ket{x(T)}$, or to the history state $\ket{x_H}$ in Eq.~\eqref{eq:historystate}, 
can be outputted with bounded probability of success, using an expected number
$\mathcal{Q}_{H,T}$ queries to $U_A$ and $4\mathcal{Q}_{H,T}$ queries to $U_0$. In the case where we have $\v{b} \neq \v{0}$ then $4\mathcal{Q}_{H,T}$ queries to $U_b$ are also made. The query count parameter  is $\mathcal{Q}_{H,T}=\min_{+,\times} \mathcal{Q}_{H,T}^{+,\times}$ with a minimization over the choice of either the $+$ scheme or the $\times$ scheme. The values of $\mathcal{Q}_{H,T}^{+,\times}$ are  given as follows:  
\begin{enumerate}[label={[\arabic*]}]
 \item Set the time discretization error $\epsilon_{\mbox{\tiny TD}}=\epsilon/(8f^{+,\times}_{H,T})$, where $f^{\times}_{H,T} = 1$, $f^+_H = \frac{1}{ x_{\mathrm{rms}} }$, $f^+_T =\frac{1}{\|\v{x}(T)\|}$. 
 \item Compute the Taylor truncation order: 
        \begin{equation}
        \label{eq:k}
            k_{+,\times} = \left \lceil \frac{3\log(s_{+,\times})/2+1}{\log[1+ \log(s_{+,\times})/2]}-1 \right \rceil,
        \end{equation} 
        with
        \small
        \begin{align}
    s_\times \geq & \frac{M e^3}{\epsilon_{\mbox{\tiny TD}}} \left(1+ T e^2 \frac{\| \v{b}\|}{ x_{\min} }\right) \nonumber \\
    s_+ \geq & \frac{Me^3  x_{\max} }{\epsilon_{\mbox{\tiny TD}}} \left( 1 +  T e^2 \frac{\|\v{b}\|}{ x_{\max} } \right),
    \end{align}
    \normalsize
\item In the case of the history state output, set $p =0$. In the case of the solution state output and $A$ being stable, set 
\small
\begin{align*}
     p_{+,\times} =\lceil \sqrt{M}/(k_{+,\times}+1)\rceil (k_{+,\times}+1)
\end{align*}
\normalsize and otherwise
\begin{align}
    p_{+,\times} =\lceil M/(k_{+,\times}+1)\rceil (k_{+,\times}+1),
\end{align}
 in the case when $A$ is not stable.
 \item Compute the scaling parameter $\omega_{\tilde{L}} =  \frac{1+\sqrt{k_{+,\times}+1}+\omega h }{\sqrt{k_{+,\times}+1} +2}$.
 \item Compute the upper bound on the condition number of~$L$:
            \footnotesize
        \begin{align}
         \nonumber 
            \! \! \! \!   \! \! \! \!   \! \! \!
 \kappa_L \! = \! & \left[   (1+\frac{\epsilon_{\mbox{\tiny TD}}}{ c_{+,\times}})^2 \left( 1  + g(k) \right ) \kappa_P  \left(p \frac{1- e^{2 T\mu_P(A) +2 h\mu_P(A) }}{ 1- e^{2 h\mu_P(A) }} \right. \right. \\ \nonumber & \left. \left. + I_0(2) \frac{e^{2\mu_P(A)(T+2h)} + M +1 - e^{2h\mu_P(A) } (2+M)}{(1-e^{2  h\mu_P(A)})^2 } \right) \right. \\ & \left. + \frac{p(p+1)}{2} \! + (p + M k) (I_0(2)-1) \right]^{\frac{1}{2}} \! \!\! (\sqrt{k+1}+2)   \label{eq:kappaLbound}
        \end{align}  
        \normalsize
        with $k = k_{+,\times}$, $p= p_{+,\times}$ and where  $g(k) = \sum_{s=1}^{k} \left(s! \sum_{j=s}^{k} 1/j!\right)^2 \leq ek$, $I_0(2) \approx 2.2796$ ($I_0$ is the order-zero modified Bessel function of the first kind), $c_{\times}=1$, $c_+ = x_{\mathrm{max}}$.
          \item Fix $K=(3-e)^2$ for the inhomogeneous case ($\v{b} \neq \v{0})$ and $K=1$ for the homogeneous case $(\v{b} = \v{0})$. For outputting the history state compute the success probability
\begin{equation}
\label{eq:prH}
  \! \! \!  \!\mathrm{Pr}^{+,\times}_H \ge \frac{K}{K-1+I_0(2)} \geq \begin{cases}
    29/500 \quad \v{b}\neq \v{0} \\
     219/500 \quad \v{b} = \v{0},
    \end{cases}
\end{equation}
and for outputting the solution state compute
\footnotesize
\begin{align}
  \! \! \!  \! \mathrm{Pr}^{+,\times}_T 
	 \geq \frac{1 }{ (1-\frac{I_0(2)-1}{(p+1)K})  + (M+1) \frac{I_0(2)-1}{(p+1)K} \left(\frac{1+\epsilon_{\mbox{\tiny TD}}}{1-\epsilon_{\mbox{\tiny TD}}}\right)^2 \bar{g}_{+,\times}^2},
  \end{align}
  \normalsize
  with $p= p_{+,\times}$. 
 \item Compute the error parameter $\epsilon_L = \epsilon\frac{\mathrm{Pr}^{+,\times}_{H,T} }{4+\epsilon}$.
  \item Compute
  \begin{align}
      \mathcal{Q}_{QLSA} = \mathcal{Q}_{QLSA}(\omega_{\tilde{L}}, \kappa_L, \epsilon_L),
  \end{align} using the best available query counts for QLSA given the parameters $\omega_{\tilde{L}}$, $\kappa_L$, $\epsilon_L$. In this work we use the closed formula for $\mathcal{Q}_{QLSA}$ from Ref.~\cite{jennings2023efficient},
 $\mathcal{Q}_{QLSA} = \mathcal{Q}^*/{(0.39-0.204 \epsilon_L)}$, where $\mathcal{Q}^*(\omega_{\tilde{L}}, \kappa_L, \epsilon_L) $ equals 
		\scriptsize
 \begin{align*}
  \frac{ 581\omega_{\tilde{L}} e}{250}  \sqrt{\kappa_L^2 +1} \left(  \left[\frac{133}{125} + \frac{4}{25\kappa_L^{1/3}}\right]  \pi \log (2\kappa_L +3) +1 \right)   \\ +  \frac{ 117}{50} \log ^2(2 \kappa_L +3) 
 \left(\log \left(\frac{ 451 \log ^2(2 \kappa_L +3)}{\epsilon_L }\right)+1\right) \\ + \omega_{\tilde{L}} \kappa_L \log \frac{32}{\epsilon_L}.
	\end{align*}
  \normalsize
  \item Compute $\mathcal{A}(\mathrm{Pr}^{+,\times}_{H,T})$, the query counts for amplitude amplification from success probability $\mathrm{Pr}_{H,T}$ (see, e.g., Ref.~\cite{yoder2014fixed}), or $\mathcal{A}(\mathrm{Pr}^{+,\times}_{H,T}))= 1/\mathrm{Pr}^{+,\times}_{H,T})$ if we sample till success.
  \item Finally, the query count parameter for any choice of $+$ versus $\times$, or $H$ versus $T$, is given by $\mathcal{Q}^{+,\times}_{H,T} = \mathcal{A}\mathcal{Q}_{QLSA} $.
\end{enumerate}
 The total logical qubit count for the algorithm is $a + 1 3 + \lceil \log_2 [((M+1)(k_{+,\times}+1) +p )N] \rceil$.
\end{theorem}
This result gives the first detailed counts for solving ODEs on a quantum computer, a crucial step to determining what problems of interest could be run on early generations of fault-tolerant quantum computers. These results also provide explicit targets for future work to improve on.

The flow of dependencies for computing $\mathcal{Q}$ is shown explicitly in Fig.~\ref{fig:cost_diagram}.

\begin{figure}[h!] 
\centering
	 \includegraphics[width=0.85\columnwidth]{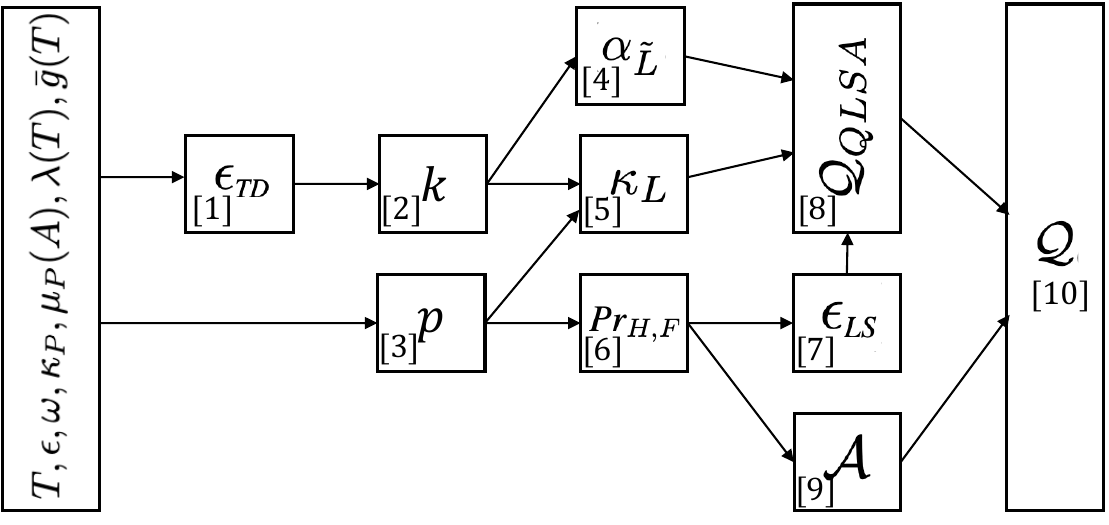}
 \caption{Flow-chart for computing the query-count for a given set of ODE parameters.}\label{fig:cost_diagram}
	  \end{figure}
The intermediary parameters that arise in this flow-chart play the following roles within the algorithm:
\begin{itemize}
    \item The parameter $\epsilon_{\mbox{\tiny TD}}$ is the target error associated to the time discretization of the dynamics via a finite Taylor truncation order. This parameter can be chosen to ensure additive errors (labeled by $+$) or multiplicative errors (labeled by $\times$). As we discussed in Sec.~\ref{sec:sketch}, this does not affect the output of the quantum algorithm. We make the choice that minimizes the overall query complexity of the quantum algorithm.
    \item The parameter $k$ is the Taylor truncation order.
    \item The parameter $p$ is an `idling' parameter required to magnify the success probability of the algorithm we when want to output the solution state.
    \item The parameter $\omega_{\tilde{L}}$ is a block-encoding scale factor for the linear-solver component, where the dynamics is encoded into a matrix $L$, rescaled to $\tilde{L}$ such that $\|\tilde{L}\|\leq 1$. 
    \item The parameter $\kappa_L$ is the condition number of the linear-solver matrix $\tilde{L}$.
    \item The parameter $\epsilon_{L}$ is the target error for the linear-solver step.
    \item The parameter $\mathcal{A}$ is the overhead in  repeating-till-success, or using amplitude amplification, that arises due to the finite success probability.
\end{itemize}
 
\begin{ex}[Query counts for history state of homogeneous ODEs]
   To illustrate the result, we apply Theorem~\ref{thm:detailedODEcounts} to analytically evaluate the non-asymptotic cost of outputting the history state for a homogeneous system of ODEs with negative (Euclidean) log-norm, an example being the  collisional linearized Vlasov-Poisson equation discussed in Sec.~\ref{sec:applicationtoODEs}.

   Choosing the multiplicative error scheme, we have 
\footnotesize
\begin{align}
	\nonumber
    \v{b}=0, \; \kappa_P =1, \; \mu_P(A) = \mu(A)<0, \;  \mathrm{Pr}^\times_H \geq I_0(2)^{-1}.
\end{align}
\normalsize
We shall set the block-encoding scale factor $\omega=1$ (the total cost will be increased by about a factor of $\omega$ if $\omega>1$).  As discussed in Example~\ref{ex:vlasovasymptotic}, the scaling with $T$ is $\tilde{O}(T^{1/2})$, but we now compute a formula for the non-asymptotic query count. Recall that $T=Mh$ is measured in units of $h$, where $h\leq 1/\|A\|$. 
Set
\begin{equation}
   k= k_\times =\frac{\frac{3}{2} \log \left(\frac{8 e^3 M}{\epsilon }\right)+1}{\log \left(\frac{1}{2} \log \left(\frac{8 e^3 M}{\epsilon }\right)+1\right)},
\end{equation}
\begin{equation}
    g(k) = \sum_{s=1}^k \left(s! \sum_{j=s}^k 1/j!\right)^2,
\end{equation}
\begin{equation}
  \xi_\mu(M) := \frac{\left(e^{2 \mu h(M+2)}-e^{2 \mu h} (M+2)+M+1\right)}{(1-e^{2\mu h})^2}.
\end{equation}
Theorem~\ref{thm:detailedODEcounts} then returns the analytical query count upper bound    

\begin{widetext}
 \tiny
\begin{align}
\! \! \! \! \! \! \! \! \! \! \! \!  \mathcal{Q} =\frac{117}{50} \frac{I_0(2)}{0.39\, -\frac{0.08949 \epsilon }{\epsilon +4}}  \left\{\log \! \! \left[ 451 \frac{(\epsilon +4)}{\epsilon} I_0(2) \log ^2\left[2 \left(\sqrt{k+1}+2\right) \sqrt{\left(\frac{\epsilon }{8}+1\right)^2 I_0(2) (g(k)+1) \xi_\mu(M)+k T/h (I_0(2)-1)}+3\right]\right] \! \! +1 \right\} \nonumber \\
\log ^2\left[2 \left(\sqrt{k+1}+2\right) \sqrt{\left(\frac{\epsilon }{8}+1\right)^2 I_0(2) (g(k)+1) \xi_\mu(M)+k M (I_0(2)-1)}+3\right] \nonumber \\ +\frac{581}{250} e \sqrt{\left(\sqrt{k+1}+2\right)^2 \left(\left(\frac{\epsilon }{8}+1\right)^2 I_0(2) (g(k)+1) \xi_\mu(M)+k M (I_0(2)-1)\right)+1} \nonumber \\
   \left\{\! \!\frac{4\pi/25}{\! \sqrt[3]{\! \left(\!\sqrt{k\! +\!1}+\!2\right) \! \!  \sqrt{\! \! \left(\! \frac{\epsilon }{8}\!+\!1\!\right)^2 \! I_0(2) (g(k)\!+\!1) \xi_\mu(M)+k M (I_0(2)\! -\!1)}}}\! +\! \frac{ 133\pi}{125} \log \! \! \left[\! 2 \! \left(\! \! \sqrt{k\!+\!1}+2\! \right) \! \! \sqrt{\! \!\left(\frac{\epsilon }{8}\!+\!1\right)^2 \! \! \! I_0(2) (g(k)+1) \xi_\mu(M)+k M (I_0(2)\! -\! 1)}+3\! \right]\! \! +1\! \! \right\} \nonumber \\ +\left(\sqrt{k+1}+2\right) \log \left[\frac{32 (\epsilon +4) I_0(2)}{\epsilon }\right]\sqrt{\left(\frac{\epsilon }{8}+1\right)^2 I_0(2) (g(k)+1) \xi_\mu(M)+k M (I_0(2)-1)},
\label{eq:explicitquerylognegative}
\end{align}
\normalsize
\end{widetext}
as a function of the target error $\epsilon$, simulation time $T$, for each choice of the log-norm parameter $\mu h \in [-1,0)$. While clearly unwieldy, the main point here is to illustrate that Theorem~\ref{thm:detailedODEcounts} can be used to analytically compute closed, non-asymptotic query cost upper bounds of ODE solvers by simple sequential replacements.

In Fig.~\ref{fig:detailedODEcounts} we set $\epsilon=10^{-10}$, $\|A\|=h = 1$ and plot $\mathcal{Q}$ as a function of $T$ for all possible $\mu(A)$, i.e. $\mu(A) \in [-1,0]$, with $\mu(A) \rightarrow 0^-$ corresponding to the $\alpha (A) \ge 0$ case with $C_{\mathrm{max}}=1$.  We computed $\mathcal{Q}$ under the assumption that we do not perform amplitude amplification, but simply repeat until success. We obtain the shaded region in Fig.~\ref{fig:detailedODEcounts}, which interpolates from the linear in $T$ scaling to the best possible $\sqrt{T}$ scaling. The upper curve in Fig.~\ref{fig:detailedODEcounts} can be upper bounded for $T \in [10^6, 10^{15}]$ and \mbox{$\epsilon = 10^{-10}$} as $\mathcal{Q} \leq  6133 T \log(T)$. 
For the lower curve, \mbox{$Q \leq 7260 \sqrt{T} \log(T)$} with the same specifications. 
We see that incorporating knowledge about stability can lead to orders of magnitude improvements in the query counts (up to $5$ orders of magnitude at $T= 10^{10}$). The best prior algorithm \cite{krovi2022improved} simply inferred $C_{\mathrm{max}}=1$ from $\mu<0$ and hence obtained linear in $T$ scaling corresponding to the upper curve in Fig.~\ref{fig:detailedODEcounts}. However, as suggested before, we expect that a similar analysis to the one we performed here for the algorithm in Ref.~\cite{berry2017quantum} can also be performed for the algorithm in Ref.~\cite{krovi2022improved} to extract an improved scaling.
\end{ex}

\section{Outlook}

Quantum computers are expected to be a disruptive technology for the simulation of quantum mechanical problems. We presently do not know if quantum computing will have a similar impact on important classes of classical dynamical problems. What we do know, however, is that there is no scarcity of computational problems of broad relevance, such as in fluid dynamics and plasma physics, which are often under-resolved on today's largest computers and will remain so in the foreseeable future. If we could identify such use-cases and provide a fully compiled quantum algorithm with reasonable resource estimates, then the scientific and technological value of quantum computers would be substantially increased.
If this development follows a similar line to that for quantum chemistry, increasingly better algorithms and resource estimates will be required, likely starting from very high resource counts and then pushing numbers down for specific problem instances. 

We see the present work as a step in this direction. In giving the first detailed query count analysis for the general problem of solving linear differential equations on a quantum computer, we hope to set a signpost for improved resource estimates in the future, and also set the stage for a detailed analysis of the cost of solving \emph{nonlinear} differential equations. The fast-forwarding results of Theorem~\ref{thm:fast-forwarding}, and the results of the additive error treatment, furthermore, suggest these algorithms may be significantly less costly than previously thought in certain regimes. We illustrated these results presenting polynomial improvements in the $T$ scaling for the collisional Vlasov-Poisson equation, systems of damped, forced harmonic oscillators and dissipative nonlinear systems.

In conclusion, these results will make end-to-end analysis of specific use-cases considerably simpler, since given an equation $\dot{\v{x}} = A \v{x} + \v{b}$ one can focus on the construction of the unitary $U_A$ encoding the ODE matrix $A$, the unitary $U_b$ encoding the forcing $\v{b}$ and the unitary $U_0$ encoding the initial condition $\v{x}(0)$, as well as the circuits required for information extraction from the solution~\eqref{eq:finalstate} or history \eqref{eq:historystate} state. These are questions that can only be analyzed in the context of specific use-cases. The exact number of times these unitaries need to be applied, given the ODE parameters in Table~\ref{table:parameters}, has instead been analytically computed for general ODE systems in this work via Theorem~\ref{thm:detailedODEcounts}.

	\bigskip

  	{\textbf{Authors contributions:} ML conceived the core algorithm analysis and derived an early version of the results. ML and DJ derived the new condition number upper bound, obtained the fast-forwarding results via the stability analysis, computed the detailed cost upper bound for the ODE solver,  derived the success probabilities lower bounds and the error propagation analysis. RBL and ATS contributed to the stability and discretization methods. ML and DJ wrote the paper and RBL, SP, ATS contributed to structuring and reviewing the article. SP and ATS coordinated the collaboration.
	
	\bigskip
	
	{\textbf{Acknowledgements:} RBL and ATS are supported by the U.S. Department of Energy through the Los Alamos National Laboratory, under the ASC Beyond Moore’s Law project. Los Alamos National Laboratory is operated by Triad National Security, LLC, for the National Nuclear Security Administration of U.S. Department of Energy (Contract No. 89233218CNA000001). Many thanks to Yi\u{g}it Suba\c{s}\i \, and Dong An for useful discussions on these topics and Micha\l{} St\k{e}ch\l{}y, Mark Steudtner and Kamil Korzekwa for useful comments on a draft of this work. }

		 \bibliography{Bibliography.bib}

	\onecolumngrid
		\newpage

  	\section*{Appendices}
   
		\appendix
  
 \section{Notation and definitions}
 For reference, we first state the definitions of key quantities used throughout our analysis.
\begin{itemize}
    \item The time-interval $[0,T]$ is uniformly divided up into $M$ subintervals of size $h$, and we index the endpoints via $t=mh$ with $m=0,1,\dots, M$.
    \item We denote by $\v{x}(mh)$ the exact solution under continuous dynamics, at any time $t=mh$, and denote by $\v{x}^m$ the exact solution under the discretized dynamics, obtained by truncating to Taylor order $k$.
    \item We define constant upper and lower bounds on the solution norm as
    \begin{equation}
         x_{\min}  \le \|\v{x}(t)\| \le  x_{\max} , 
    \end{equation}
    for all $t \in [0,T]$. We also define the quantity $ x_{\mathrm{rms}} $ as
    \begin{equation}
         x_{\mathrm{rms}}  := \sqrt{\frac{1}{M} \sum_{m=0}^M \|\v{x}(mh)\|^2},
    \end{equation}
    which is effectively an rms-estimate of the norm of the solution over the time-interval.
    \item (Discretization error-conditions) For the choice of errors in the time-discretization section we consider both a \emph{multiplicative} error condition
    \begin{equation}
        \|\v{x}(mh) - \v{x}^m\| \le \epsilon_{\mbox{\tiny TD}} \|\v{x}(mh)\| \mbox{ for all } m=0,1,\dots M,
    \end{equation}
    as well as an \emph{additive} error condition
    \begin{equation}
        \|\v{x}(mh) - \v{x}^m\| \le \epsilon_{\mbox{\tiny TD}}  \mbox{ for all } m=0,1,\dots M.
    \end{equation}
    In what follows, we will use subscripts of $\times$ and $+$ to refer to the multiplicative and additive choice, respectively.
    \item We define the following functions for the sufficient Taylor truncation order analysis:
    \begin{align}
        \lambda_\times(T) &:= \max \left\{ \| \v{b}\| / x_{\min} , T\right\}\nonumber \\
          \lambda_+(T) &:= \max \left\{ \| \v{b}\| , T  x_{\max}  \right\},
    \end{align}
    with $\times$ referring to the multiplicative error condition and $+$ the additive error condition.
    \item We deal with two choices of quantum outputs: a coherent encoding of the solution state $\v{x}(T)$ at the final time $t=T$ (Eq.~\eqref{eq:finalstate}), and a coherent encoding of the history state data $\v{x}_H$ over the $M$ timesteps (Eq.~\eqref{eq:historystate}). Quantities that depend on this choice of output are labelled with subscripts $T$ and $H$ respectively.
    \item We define
    \begin{align}
        \bar{g}_\times (T)^2 &:= \frac{1}{M+1}\sum_{m=0}^M \frac{\| \v{x}(mh) \|^2}{  \|\v{x}(Mh)\|^2} \nonumber \\
         \bar{g}_+^2 &:= \frac{1}{M+1}\sum_{m=0}^M \left(\frac{1-\epsilon'}{1+\epsilon'}\right)^2 \frac{(\| \v{x}(mh) \|+\epsilon')^2}{  (\|\v{x}(Mh)\|-\epsilon')^2},
    \end{align}
    where $\epsilon' = \epsilon  x_{\mathrm{rms}} /8$ in the case of the history state output and $\epsilon' = \epsilon \|\v{x}(T)\|/8$ in the case of the solution state output. 
    \item To quantify errors in normalized quantum state expressions, and their dependence on the choice of additive/multiplicative time-discretization error and choice of history/solution state we define the quantitites
    \begin{align}
        f^\times_{H} &:= f^\times_{T} := 1 \nonumber \\
        f^+_H &:= \frac{1}{ x_{\mathrm{rms}} }, \quad 
        f^+_T :=\frac{1}{\|\v{x}(Mh)\|}.
    \end{align}
    
    \item Subscript and superscript convention: Note that in a number of places we have options of assuming either the additive error scheme $(+)$ or the multiplicative error scheme $(\times)$. Or to output the history state $(H)$ or solution state $(F)$. To handle all of these possibilities in tandem we shall use a comma notation to denote a choice of one or another option. For example,
    \begin{equation}
        f^{+,\times}_{H,T} = \mbox{$f^+_H$ or $f^+_T$ or $f^\times_H$ or $f^\times_T$. }
    \end{equation}
    The purpose of the notation is to streamline the analysis for the different cases that arise.
\end{itemize}

\section{Roadmap of the analysis}
\label{sec:algorithm}
Before we get into proving technical results, we provide a roadmap of what the analysis will entail. 
 
As discussed above, the recursive equations~\eqref{eq:recursive} are encoded into a linear system following Ref.~\cite{berry2017quantum}. Specifically, we introduce $3$ registers $\ket{m,j,n}$. Register $m$ is a clock degree of freedom, labeling the solution at each time-step $t=mh$ within $[0,T]$, plus potentially further idling steps to increase the amplitude of the solution at time $T$, if necessary. It takes values $m=0, \dots,  M+p/(k+1)$, where $p$ is a multiple of $k+1$ to be fixed. The $j$ register is used to compute Taylor sums, $j=0,\dots,k$. Finally, $n$ labels the components of the vector $\v{x}$, with $n=1,\dots,N$. We encode data about the solution to the ODE problem into the solution $\v{y} := L^{-1} \v{c}$ of a linear system (Appendix~\ref{sec:linearsystemembedding}): 
\begin{equation}
\label{eq:linearsystem}
    L \v{y} = \v{c},
\end{equation}
where $L = L_{1} + L_{2} + L_{3}$, and
\small
\begin{align}
L_1 := & I \otimes I \otimes I, \nonumber\\
L_{2} := & - \sum_{m=0}^{M-1} \sum_{j=1}^{k}  \ketbra{m,j}{m,j-1} \otimes \frac{Ah}{j}  - \theta(p) \sum_{m=M}^{ M+ \frac{p}{k+1}-1} \sum_{j=1}^{k} \ketbra{m,j}{m,j-1} \otimes I \nonumber\\
L_{3} := & - \sum_{m=0}^{M-1} \sum_{j=0}^k   \ketbra{m+1,0}{m,j} \otimes I,
\end{align}
\normalsize
with the Heaviside function $\theta(p) = 0$ for $p=0$ and $\theta(p) =1$ for $p>0$.
The term $L_{2}$ generates Taylor components, which are summed at each $m$ by $L_{3}$ to generate a truncated Taylor series that approximates the dynamics up to $m=M$. If \mbox{$p\geq k+1$}, from $m=M$ to $m=M+p/(k+1)-1$ and for all $j=0,\dots, k$, so $p$ steps overall, the term $L_2$ generates `idling', which allows for an increase of the amplitude of the solution state component. In Eq.~\eqref{eq:linearsystem} we set

 \begin{equation}
    \v{c} = [ \overbrace{\v{x}^0, \underbrace{\v{b}, \v{0}, \dots, \v{0}}_{k}}^{m=0}, \dots, \overbrace{\v{0},\underbrace{\v{b}, \v{0}, \dots, \v{0}}_{k}}^{m=M-1},\underbrace{\v{0}, \dots, \v{0}}_{1+p}].
\end{equation}

It can be shown that the solution to this system can be split up as $\v{y} =  \v{x}_H^{\epsilon_{\mbox{\tiny TD}}} + \v{x}_{\mathrm{junk}}$, where 
\begin{equation}
    \v{x}_H^{\epsilon_{\mbox{\tiny TD}}} = (\v{x}^0, \underbrace{\v{0}, \dots, \v{0}}_{k}, \v{x}^1,\underbrace{\v{0}, \dots, \v{0}}_k,\underbrace{\v{x}^M, \dots, \v{x}^M}_{1+p})
\end{equation} 
and $\v{x}_{\mathrm{junk}}$ is an orthogonal `junk' component, that is unwanted, and can be removed by post-selection or amplitude amplification (specifically, $\v{x}_{\mathrm{junk}}$ is nonzero only where $  \v{x}_H^{\epsilon_{\mbox{\tiny TD}}}$ is zero).

As we discussed, there are two ways of choosing the truncation scale $k$. We can either follow Ref.~\cite{berry2017quantum} and choose $k$ such that multiplicative error is achieved 
\begin{equation}
    \| \v{x}^m - \v{x}(mh) \| \leq \epsilon_{\mbox{\tiny TD}} \| \v{x}(mh)\|, \quad m=0, \dots, M.
\end{equation}
or we can choose it such that additive error is achieved
 \begin{equation}
    \| \v{x}^m - \v{x}(mh) \| \leq \epsilon_{\mbox{\tiny TD}}, \quad m=0, \dots, M.
\end{equation}
The problem statement remains the same, i.e., preparing a state $\epsilon$ close to the solution~\eqref{eq:finalstate} or the history state~\eqref{eq:historystate}. The multiplicative (labeled by $\times$) or additive (label by $+$) choice gives two different algorithms with distinct query complexities. 

When we output the history state, we set $p=0$ and  $\v{x}_H^{\epsilon_{\mbox{\tiny TD}}}$, as a normalized quantum state $\ket{x^{\epsilon_{\mbox{\tiny TD}}}_{H}}$,
provides a discrete approximation to the exact history state encoded in Eq.~\eqref{eq:historystate}. Instead, when we output the solution state, we set $p>0$ and target the components $\v{x}^T$ with amplitude amplification. The time-discretization error analysis can be found in Appendix~~\ref{sec:timediscretizationerror}. 

Eq.~\eqref{eq:linearsystem} is tackled using both the best available asymptotic~\cite{costa2022optimal} and non-asymptotic~\cite{jennings2023efficient} query complexity upper bounds for QLSA. The resource estimate depends on three parameters:
\begin{enumerate}
    \item[(a)] The cost of constructing a block-encoding of $L$ from that of $A$.
    \item[(b)] An upper bound on the condition number of $L$, i.e., on $\kappa_{L}:= \| L\| \| L^{-1}\|$. 
    \item[(c)] The QLSA does not output the ideal solution $\ket{y}$, but a quantum state $\rho$ which is $\epsilon_L$-close to it in $1-$norm. The cost depends on $\epsilon_L$.
\end{enumerate}
All three will be determined from the ODE parameters:

\emph{(a) Constructing the linear system block-encoding.}
We prove that \mbox{$\| L\| \leq \omega_k := \sqrt{k+1} +2$}. Hence we define $\tilde{L} = L/\omega_k$, $\tilde{\v{c}} = \v{c}/\omega_k$ and solve the problem $\tilde{L} \v{y} = \tilde{\v{c}}$. This has $\| \tilde{L} \| \leq 1$, as normally required by QLSA. The detailed resource estimate of a QLSA is in terms of calls to a unitary block-encoding $U_{\tilde{L}}$ of $\tilde{L}$ and a state preparation unitary $U_{\tilde{c}}$ of $\tilde{\v{c}}$. These need to be constructed from the available unitaries $U_A$, $U_b$, $U_0$ (details in Appendix~\ref{sec:block-encoding}). Concerning $U_{\tilde{c}}$, it can be constructed with a single call to both $U_0$ and~$U_b$. As for $U_{\tilde{L}}$, 
 we construct an $( \omega_{\tilde{L}} ,a+ 6,0)$-block-encoding of $\tilde{L}$ with a single call to $U_A$, where $\omega_{\tilde{L}} = (1+\sqrt{k+1}+\omega h)/\omega_k$, as given in Theorem~\ref{thm:detailedODEcounts}. The QLSA requires $6$ further ancilla qubits, as well as $\lceil \log_2[ ((M+1)(k+1) + p) N] \rceil$ qubits to encode~$\tilde{L}$.

\emph{(b) Condition number bound.} Recall that the Lyapunov analysis leads to stability parameters $(\kappa_P, \mu_P(A))$. In Appendix~\ref{sec:conditionnumberanalysis} we find the upper bound on $\kappa_L$ in terms of the stability properties as given in Theorem~\ref{thm:detailedODEcounts}, inequality~\eqref{eq:kappaLbound}. The core asymptotic and non-asymptotic savings in the query count of our analysis are due to this tightened bound, as well as to the flexibility in choosing additive or multiplicative error truncation.

	    We set $p=0$ for the case where we output the history state.  For the case where we output the solution state, we take \mbox{$p=\lceil M^{1/2}/(k+1)\rceil (k+1)$} in the case that $A$ is stable, and \mbox{$p=\lceil M/(k+1)\rceil (k+1)$} 
 otherwise. For both final and history state, we then have $\kappa_L = O(T^{1/2} k)$ in the stable case and $\kappa_L = O(T k)$ otherwise (recall that $T=Mh$). The better scaling of this bound leads to the improvements in Table~\ref{table:comparisons}.
    
    Also note that in Appendices~\ref{sec:linearsystemembedding} and \ref{sec:conditionnumberanalysis} we provide a broader analysis of idling, where we allow varying amplifications at each discrete time-step, specified by integers $\{p_m : m =0, 1, \dots, M\}$. This allows us to increase the amplitude of selected terms $\ket{x(t_m)}$ in the history state \eqref{eq:historystate}, which can be understood as applying a filter to the dynamical data. 

\emph{(c) Required precision.} To determine the required $\epsilon_L$, we first need to analyze the output processing, as we shall now discuss.

We first sketch the analysis assuming we access the idealized linear system solution $\ket{y}$, and then consider the effect of finite precision.

\emph{Idealized analysis.} The ideal solution of the linear system encodes ODE data in the sense that we recover a state close to $\ket{x_{H,T}}$ by performing a projective measurement \mbox{$\{ \Pi_{H,T}, I - \Pi_{H,T}\}$} on the auxiliary/clock registers. Specifically, if 
\begin{equation*}
    \Pi_H = I \otimes \ketbra{0}{0} \otimes I, \quad \Pi_T = \sum_{m=M}^{ M+p/(k+1)} \ketbra{m}{m} \otimes I \otimes I,
\end{equation*}
we have
\begin{equation}
  \Pi_{H} \ketbra{y}{y} \Pi_{H}=   \mathrm{Pr}_{H} \ketbra{x^{\epsilon_{\mbox{\tiny TD}}}_{H}}{x_{H}^{\epsilon_{\mbox{\tiny TD}}}},
\end{equation}
where $\ket{x^{\epsilon_{\mbox{\tiny TD}}}_{H}}$  approximates $\ket{x_H}$ and
\begin{equation}
 \mathrm{Tr}_{CA} [\Pi_{T} \ketbra{y}{y} \Pi_{T}]=   \mathrm{Pr}_{T} \ketbra{x^M}{x^M},
\end{equation}
where the trace discards clock and ancilla qubits and $\ket{x^M}$ approximates $\ket{x_M}$ (Appendix~\ref{sec:timediscretizationerror}).

The values $\mathrm{Pr}_{H,T}$ are success probabilities of the corresponding measurements. 
For outputting the history state we obtain the success probability in \eqref{eq:prH},
so the overhead to repeat the measurement until success or via amplitude amplification is $O(1)$.
To output the solution state we have the expressions given in Theorem~\ref{thm:detailedODEcounts}, which imply
\begin{equation}
\label{eq:prF}
    \mathrm{Pr}_T = \begin{cases}
     \Omega(\bar{g}^{-2}_{+,\times} T^{-1/2}) \quad \textrm{for $A$ stable,} \\
    \Omega(\bar{g}^{-2}_{+,\times}) \quad \mathrm{otherwise}.
    \end{cases}
\end{equation}
 Using amplitude amplification we have an overhead of  $O(1/\sqrt{\mathrm{Pr}_T})$ repetitions.

 \emph{Error analysis.} We then consider the effect of the finite precision $\epsilon_L$ of the QLSA and go back to point~(c) of the previous section, i.e. answer the question of what linear solver error $ \epsilon_L $ suffices, as well as what discretization error $\epsilon_{\mbox{\tiny TD}}$ we need to target, given a total error budget of $\epsilon$. This analysis is performed in Appendix~\ref{sec:errorpropagation}. The errors $ \epsilon_L $ and $\epsilon_{\mbox{\tiny TD}}$ are then chosen as to ensure total error $\epsilon$ in $1$-norm. More specifically, the algorithm outputs a mixed state $\epsilon$-close in $1$-norm to the solution or the history state.

\section{Stability and Lyapunov relations}
\label{sec:stabilitylyapunov} 
For this discussion we refer to Ref.~\cite{plischke2005transient, van2006study} for further details, but many more sources can be found on the topic. In what follows we assume all matrices lie in $\mathbb{C}^{N\times N}$ and all vectors lie in $\mathbb{C}^N$, unless otherwise stated.

Consider the dynamics generated by the matrix $A$, i.e., the homogeneous linear system $\dot{\v{x}}(t) = A \v{x}(t)$, whose formal solution is $\v{x}(t) = e^{At} \v{x}(0)$. A sufficient condition for the norm of $\v{x}(t)$ not to blow-up as $t \rightarrow +\infty$ is that $\|e^{At}\|$ is bounded, since $\|\v{x}(t)\| \leq \|e^{At} \| \| \v{x}(0)\|$. This motivates the definition $\alpha(A) := \max_i \{ \mathrm{Re}(\lambda_i(A))\}$, where $\lambda_i(A)$ are the eigenvalues of $A$. If $\alpha(A) < 0 $, then $A$ is said to be \emph{stable}. We have that 
\begin{equation}
\label{eq:stabledecay}
    \lim_{t \rightarrow \infty} \| e^{At} \| = 0,
\end{equation} 
if and only if $A$ is stable. A necessary and sufficient condition for stability is that there is a matrix $P>0$ (strictly positive) such that 
\begin{equation}\label{eqn:Lyapunov}
    P A + A^\dag P = - Q,
\end{equation}
for some $Q>0$. This is known as the \emph{Lyapunov equation}. This condition is related to the existence of quadratic Lyapunov functions for the  dynamics generated by $A$. Indeed, it can be shown that for the homogeneous linear system $\dot{\v{x}}(t)= A \v{x}(t)$,  Eq.~(\ref{eqn:Lyapunov}) holds if and only if the quadratic form $V(\v{x}) := \v{x}^\dagger P \v{x}$ obeys
 \begin{equation}
    \frac{d}{dt} V(\v{x}(t)) = -\v{x}^\dagger (t) Q \v{x}(t) < 0 ,
\end{equation}
along all trajectories $\v{x}(t)$. The non-negative quadratic form $V(\v{x})$ can be viewed as a form of energy function, and from the fact that  $\dot{V}(\v{x}) =-\v{x}^\dagger (t) Q \v{x}(t) \le -c V(\v{x})$, for some constant $c>0$. This condition ensures stability for the dynamics. 
 Note that the above discussion still holds in the $\alpha(A) = 0$ case, but the strict inequality $c>0$ is replaced with the non-strict one $c\geq 0$, allowing $c=0$ and $C$ not constant.

While, asymptotically, for stable systems the norm of $e^{At}$ decays to zero (Eq.~\eqref{eq:stabledecay}), it can still display transient increases at finite times. One way to capture this is via the \emph{logarithmic norm} (or log-norm for short) of $A$, denoted $\mu(A)$ and defined as the initial growth rate
\begin{equation}
    \mu(A) := \frac{d}{dt^+} \|e^{At} \|_{t=0} = \lim_{h \rightarrow 0^+} \frac{1}{h} ( \| I +Ah \| - 1).
\end{equation}
A significant property of the log-norm is that it is the smallest $\beta \in \mathbb{R}$ such that $\|e^{At}\| \le e^{\beta t}$ for all $t \ge 0$,
\begin{equation}
\label{eq:euclideanlognormbound}
    \| e^{At}\| \leq e^{\mu(A) t}.
\end{equation}
In the case of the operator norm, its value can also be calculated from
\begin{equation}
    \mu(A) = \frac{1}{2} \lambda_{\mathrm{max}} (A+A^\dagger)=\max_{\|\v{x}\| \ne 0} \mathrm{Re} \,\frac{  \langle  A \v{x}, \v{x}\rangle }{\langle  \v{x}, \v{x}\rangle},
\end{equation}
where $\lambda_{\mathrm{max}} (X)$ denotes the largest eigenvalue of any Hermitian matrix $X$. If  $\mu(A) < 0$ then the system is  stable, but even if $\mu(A) >0$ stability can still occur. However, in this case the bound \eqref{eq:euclideanlognormbound} is no longer useful. To account for this regime one can introduce a family of log-norms that depend on $P$:
\begin{equation}
	      \mu_P(A) = \max_{\|\v{x}\| \ne 0} \mathrm{Re} \,\frac{  \langle P A \v{x}, \v{x}\rangle }{\langle P \v{x}, \v{x}\rangle},
	  \end{equation}
	  the \emph{log-norm with respect to the elliptical norm induced by} $P$. This can be seen as an initial growth rate with respect to the norm induced by $P$ [Lemma 3.31, Ref.~\cite{plischke2005transient}], and it recovers the log-norm bound when $P=I$.

   Note that stability implies $\mu_P(A)<0$ for some $P>0$. In particular, denote the condition number of $P$ by $\kappa_P$. One then has [Theorem 3.35, Ref.~\cite{plischke2005transient}]
\begin{equation}
\label{eq:stablebound}
    \| e^{At}\| \leq \sqrt{\kappa_P} e^{\mu_{P}(A) t},
\end{equation}
 which accounts for potential increases of  $ \| e^{At}\|$ in a transient regime, and recovers Eq.~\eqref{eq:euclideanlognormbound} for $P=I$.

Finally, we note that in the case of Schr\"{o}dinger evolution under $A=-iH$ for some Hamiltonian $H$,  we have a constant norm for the dynamics, and so the system is not stable in the above sense; in the literature this case is called ``marginally stable''. Quantum dynamics, or any other linear dynamics with $A$ having eigenvalues with vanishing real parts, are not stable under arbitrarily small perturbations in $A$, while stable systems do allow small perturbations while still remaining stable.

\section{Time discretization error analysis}
\label{sec:timediscretizationerror}
We now describe the formal solution to the ODE, and the discretized solution that our algorithm will return.
The solution of Eq.~\eqref{eq:ODE} can be formally written as 
\begin{equation}
    \v{x}(t) = e^{At} \v{x}(0) + \int_{0}^t dt' e^{A t'} \v{b},
\end{equation}
for any $t \in [0,T]$. 

We proceed by splitting the time interval up into $M$ equal-sized sub-intervals of width $h$, so that we can approximate the dynamics at each of these time points. The exact solution data at these points are $\v{x}(0), \v{x}(h), \dots, \v{x}(mh), \dots, \v{x}(Mh) = \v{x}(T)$. We shall approximate $\v{x}(mh)$ with $\v{x}^m$, which is a vector obtained from a discrete sequence of evolutions obtained from a truncated Taylor series for the exponential $e^{At}$. In particular we use
\begin{align}
    T_k( Ah) & := \sum_{j=0}^k \frac{1}{j!} (Ah)^j \nonumber \\
    S_k(Ah) &:= \sum_{j=1}^k \frac{1}{j!}(Ah)^{j-1},
\end{align}
to approximate the exponentials needed to order $k$.
The discretized dynamics is now defined recursively by $\v{x}^0 := \v{x}(0)$, and
\begin{equation}
\label{eq:truncatedtaylorserieapprox}
     \v{x}^m := T_k(Ah) \v{x}^{m-1} + S_k(Ah) \v{b}, 
\end{equation}
for all $m=1, \dots, M$. The following result gives the condition on the truncation order $k$ so as to obtain a target relative error~$\epsilon_{\mbox{\tiny TD}}$.

\begin{lem}[Time-discretization multiplicative error, Theorem 3 of Ref.~\cite{krovi2022improved}]
\label{lem:timediscretization}
Let \mbox{$ x_{\min}  \leq \min_{t \in [0,T]} \| \v{x}(t)\|$} and take $\|A\| h \leq 1$ for a time-step $h>0$.  Let $\v{x}(t)$ be the solution to $\dot{\v{x}}(t) = A\v{x}(t) + \v{b}$ and $\v{x}^m$ its approximation at time $t= mh$ defined in Eq.~\eqref{eq:truncatedtaylorserieapprox}. If we choose the integer $k$ such that
		\begin{equation}
		\label{eq:krequirement}
		    (k+1)! \geq \frac{M e^3}{\epsilon_{\mbox{\tiny TD}}} \left (1+ T e^2 \frac{\| \v{b}\|}{ x_{\min} } \right ),
		\end{equation}
		then it follows that the relative time-discretization error obeys
		\begin{equation}
		\label{eq:timediscretizationerror}
			\| \v{x}(m h) -  \v{x}^{m} \| \leq \epsilon_{\mbox{\tiny TD}} \| \v{x}(m h)\| , \end{equation}
			for all $m =0,1, \dots M$.
\end{lem}
However there are situations where we might want an additive error of the form
\begin{equation}
    \| \v{x}(m h) -  \v{x}^{m} \| \leq \epsilon_{\mbox{\tiny TD}}.
\end{equation}
To establish a sufficient Taylor truncation order in this setting we introduce some notation that relate the continuous time dynamics and the discrete time dynamics. 

 From Eq.~\eqref{eq:ODEsolution}, a time-increment of $h$ gives the following recursion relation
 
    \begin{align}  
    \v{x}(mh) &= e^{Ah} \v{x}((m-1)h) +  \int_0^{h} ds e^{A s} \v{b}  \nonumber \\ &  
        = e^{Ah} \v{x}((m-1)h) + h \sum_{j=0}^\infty \frac{(Ah)^j}{(j+1)!} \v{b} \nonumber \\ & = T_\infty(Ah) \v{x}((m-1)h) + h S_\infty(Ah) \v{b},
        \label{eq:exactrecursion}
    \end{align}
    for any $m=1,2,\dots M$, where
    \begin{align}
    \label{eq:definitionsTandS}
T_k(Ah) = \sum_{j=0}^{k} \frac{(Ah)^j}{j!}, \quad S_k(Ah) = \sum_{j=0}^{k} \frac{(Ah)^j}{j+1!}.
    \end{align}
    We shall bound errors obtained by approximating this recursion relation via
    \begin{equation}
        \v{x}^{m} = T_k(Ah) \v{x}^{m-1} + hS_k(Ah)\v{b},
    \end{equation}
    and provide estimates on $k$ such that $T_k(Ah)$ sufficiently approximates $T_\infty(Ah)$ and $S_k(Ah)$ sufficiently approximates $S_\infty(Ah)$. 

Now let us look at the integrated relation, giving the solution at time $mh$ from that at time $0$. For the exact solution up to time $t=mh$
    \begin{align}
        \v{x}(mh) &= e^{Amh} \v{x}(0) + \int_0^{mh}\!\! ds \,\, e^{As} \v{b} \nonumber \\
         &= \left (e^{Ah}\right )^m \v{x}(0) + \sum_{j=0}^{m-1} \int_{0}^h ds e^{A (jh+s)}  \v{b}\nonumber \\
        &= \left (e^{Ah}\right )^m \v{x}(0) + h\sum_{j=0}^{m-1}\left (e^{Ah} \right)^j S_\infty(Ah) \v{b}\nonumber \\
                &=e^{Amh} \v{x}(0) + L_\infty \v{b},
                \label{eq:integratedexact}
    \end{align}
    where
    \begin{align}
    \label{eq:definitionL}
        L_k := h\sum_{j=0}^{m-1} T_k(Ah)^j S_k(Ah),
    \end{align}
    for any $k=1,2,\dots, \infty$. We shall then approximate                 $\v{x}(mh)$ with  $\v{x}^m$, where 
    \begin{align}
    \label{eq:integratedapprox}
        \v{x}^m &= T_k(Ah) \v{x}^0 + h\sum_{j=0}^{m-1} T_k(Ah)^j S_k(Ah) \v{b}\nonumber \\
        &= T_k(Ah)\v{x}^0 + L_k \v{b}.
    \end{align}

 Finally, we can derive the following useful relation:
    \begin{align}
        L_k A = AL_k &= Ah\sum_{j=0}^{m-1} T_k(Ah)^j S_k(Ah) \nonumber \\
        &= \sum_{j=0}^{m-1} T_k(Ah)^j \sum_{j=0}^{k-1} \frac{(Ah)^{j+1}}{(j+1)!} \nonumber \\
            &= \sum_{j=0}^{m-1} T_k(Ah)^j (T_k(Ah)-I)\nonumber \\
            &= \sum_{j=0}^{m-1} (T_k(Ah)^{j+1}-T_k(Ah)^j)\nonumber \\
            &=T_k(Ah)^m -I,
            \label{eq:LAusefulrelation}
    \end{align}
    for any choice of truncation order $k=1,2,\dots, \infty$.  
    
The following lemma from~\cite{krovi2022improved} gives the truncation errors for the relevant operators of the dynamics.

\begin{lem}[From \cite{krovi2022improved}] \label{lem:taylor-operator-bounds}
 Assume $\|Ah\| \le 1$, and choose any $m=1,2,\dots, M$. With the definitions in Eq.~\eqref{eq:definitionsTandS} and Eq.~\eqref{eq:definitionL}, if $k$ obeys \mbox{$me^2/(k+1)! \le 1$} it follows that
    \begin{equation}
        \| (e^{Amh} - T_k(Ah)^m)e^{-Amh} \| \le \frac{(e-1)me^2}{(k+1)!},
    \end{equation}
    and also
    \begin{equation}
        \|(L_\infty - L_k)e^{-Amh}\| \le \frac{mTe^5}{(k+1)!}.
    \end{equation}
    
\end{lem}
Therefore, the error grows linearly or quadratically in time and decreases exponentially with the truncation order. We now prove the following discretization error result, which parallels the bounds in~\cite{krovi2022improved} for multiplicative errors.

\begin{lem}[Time-discretization error with additive errors]
\label{lem:timediscretizationadditive}
Let $h>0$ be chosen so that $\|A\| h \leq 1$. Let $\v{x}(t)$ be the solution to $\dot{\v{x}}(t) = A\v{x}(t) + \v{b}$ and $\v{x}^m$ its approximation at time $t= mh$ defined in Eq.~\eqref{eq:truncatedtaylorserieapprox}.  Choose $\epsilon \in (0,1]$ and let $ x_{\max}  \geq \sup_{t \in [0, T]} \| \v{x}(t)\|$. If we choose an integer $k$ such that
	
  \begin{align}
  \label{eq:krequirementadditive}
    (k+1)! \geq  \frac{Me^3  x_{\max} }{\epsilon} \left( 1 +  T e^2 \frac{\|\v{b}\|}{ x_{\max} } \right),
\end{align}
then it follows that the time-discretization error obeys
		\begin{equation}
		\label{eq:timediscretizationerror}
			\| \v{x}(m h) -  \v{x}^{m} \| \leq \epsilon, 
   \end{equation}
			for all $m =0,1, \dots M$.
\end{lem}
\begin{proof}
 Noting that $\v{x}(0) = \v{x}^0$, from Eq.~\eqref{eq:integratedexact}, \eqref{eq:integratedapprox}, as well as the relation in Eq.~\eqref{eq:LAusefulrelation},
    \begin{align}
        \|\v{x}(m h) -  \v{x}^{m} \| &= \| (e^{Amh} -T_k(Ah)^m) \v{x}(0) + (L_\infty - L_k)\v{b} \| \nonumber \\
        &= \| (L_\infty - L_k) A\v{x}(0) + (L_\infty - L_k)\v{b} \| \nonumber \\
                &= \| (L_\infty - L_k) (A\v{x}(0) + \v{b}) \|.
    \end{align}
    Using the recursion relation \eqref{eq:exactrecursion} and the useful relation in \eqref{eq:LAusefulrelation}, we see that the vector $\tilde{\v{x}}(mh):= A\v{x}(mh)+\v{b}$ evolves homogeneously:
    \begin{align}
     \tilde{\v{x}}(mh) & =  A\v{x}(mh)+\v{b} = A e^{A mh} \v{x}(0) + A L_\infty \v{b} +  \v{b} \nonumber \\    & =  e^{A mh} A\v{x}(0) + (e^{A mh} - I) \v{b} +  \v{b} = e^{A mh} (A\v{x}(0) + \v{b}) \nonumber \\ & = e^{A mh} \tilde{\v{x}}(0).
    \end{align}
    
 Therefore, we have
     \begin{align}
        \|\v{x}(m h) -  \v{x}^{m} \| &= \| (L_\infty - L_k) e^{-Amh}(A\v{x}(mh) + \v{b}) \| \nonumber \\
        & \le \| (L_\infty - L_k) A e^{-Amh}\v{x}(mh)\| + \|(L_\infty - L_k) e^{-Amh}\v{b} \| \nonumber \\
                & =\left \| \left [e^{Amh} - T_k(Ah)^m \right ]e^{-Amh}\v{x}(mh) \right\| + \left\|(L_\infty - L_k) e^{-Amh}\v{b} \right \| \nonumber \\
                & \le \left \| \left [e^{Amh} - T_k(Ah)^m \right ]e^{-Amh}\right \| \|\v{x}(mh) \| + \left\|(L_\infty - L_k) e^{-Amh}\right \| \|\v{b}\|.
    \end{align}
    We now assume that integer $k$ obeys $Me^2/(k+1)! \le 1$, and we use the operator bounds in Lemma~\ref{lem:taylor-operator-bounds} to deduce
      \begin{align}
        \|\v{x}(m h) -  \v{x}^{m} \| 
                & \le \frac{(e-1)me^2}{(k+1)!}\|\v{x}(mh) \| + \frac{mTe^5}{(k+1)!}\|\v{b}\|  \nonumber\\ & \le \frac{(e-1)Me^2}{(k+1)!} x_{\max}  + \frac{MTe^5}{(k+1)!}\|\v{b}\|
               \nonumber \\
                & \le \frac{Me^3}{(k+1)!} x_{\max}  + \frac{MTe^5}{(k+1)!}\|\v{b}\|
    \end{align}
    The result follows.
\end{proof}

We can provide an explicit upper bound on the $k$ required from the following lemma.
\begin{lem}[Sufficient truncation scale]\label{lem:sufficientk}
Under the conditions of  Lemma~\ref{lem:timediscretization} $(\times)$ and Lemma~\ref{lem:timediscretizationadditive} $(+)$, we respectively have that the choice
\begin{equation}
\label{eq:kconditiongeneral}
   k_{+,\times} = \left \lceil \frac{3\log(s_{\times,+})/2+1}{1+ \log( \log(s_{\times,+})/2e+1/e)}-1 \right \rceil
\end{equation}
suffices to ensure multiplicative ($\times$) error, $\|\v{x}(mh) - \v{x}^m\| \le \epsilon_{\mbox{\tiny TD}} \|\v{x}(mh)\|$, or additive ($+$) error, $\|\v{x}(mh) - \v{x}^m\| \le \epsilon_{\mbox{\tiny TD}}$, over all $M=T/h$ time-steps of the dynamics. 
Here
\begin{align}
    s_\times = \frac{M e^3}{\epsilon_{\mbox{\tiny TD}}} \left(1+ T e^2 \frac{\| \v{b}\|}{ x_{\min} }\right), \quad s_+ = \frac{Me^3  x_{\max} }{\epsilon_{\mbox{\tiny TD}}} \left( 1 +  T e^2 \frac{\|\v{b}\|}{ x_{\max} } \right)
\end{align}
Therefore, to obtain either an additive error or multiplicative error $\epsilon_{\mbox{\tiny TD}}$ it suffices to choose
\begin{equation}
    k_{+,\times} = \Theta \left ( \log \lambda_{+,\times} (T) /\epsilon_{\mbox{\tiny TD}} \right ),
\end{equation}
where $\lambda_\times(T) = \max \left\{ T^2 \| \v{b}\| / x_{\min} , T\right\}$, $\lambda_+(T) = \max \left\{ T^2 \| \v{b}\| ,  T  x_{\max}  \right\}$.
\end{lem}

\begin{proof}
From Lemma~\ref{lem:timediscretizationadditive} we obtain the relevant $s_+$ for additive errors, while for purely multiplicative errors we use Lemma~\ref{lem:timediscretization} for $s_{\times}$. In either case,  it suffices to take $
    (k+1)! \geq s_{\times,+}$.
From the Stirling approximation we have that 
\begin{equation}
 n! \geq  \left ( \frac{n}{e} \right )^n,
\end{equation}
so it suffices to take $(\frac{k+1}{e})^{k+1} \geq s^*$, i.e., $k =  \lceil f^{-1}(s_*) \rceil -1$ where we define $f(s) = \left( \frac{s}{e} \right)^s$. However we have that for any $s >0$
\begin{equation}
    f^{-1} (s) = \exp \left [ W \left (\frac{\log s}{e} \right ) +1 \right ],
\end{equation}
where $W(s)$ is the principal branch of the Lambert W-function. We now use the upper bound
\begin{equation}
    e^{W(s)} \le \frac{s+y}{1 + \log y},
\end{equation}
which holds for any $s > -1/e$ and any $y> 1/e$~\cite{hoorfar2008inequalities}.

In our case this reads
\begin{equation}
     e^{W(\log(s)/e)} \leq \frac{\log(s)/e + y}{1+ \log y}
\end{equation}

We now take $y = 1/e+\log(s)/2e$, and therefore have that
  \begin{equation}
    f^{-1}(s)-1 \leq e^{W(\log(s)/e)+1}-1 \leq \frac{3\log(s)/2+1}{1+ \log( \log(s)/2e+1/e)}-1 = O(\log (s)).
\end{equation}
We now set $s=s_{\times,+}$ and round up to the nearest integer to obtain the result claimed.
\end{proof}

This expression shows that the truncation order in either the multiplicative or additive case grows at worst logarithmically in time, $k_{+,\times}= O (\log(T))$, for the homogeneous ODE case ($\v{b} =\v{0}$), while for the inhomogenous case ($\v{b} \ne \v{0}$) one can generally obtain $k_{+,\times}=O(T)$ if $\lambda_{+,\times}(T)$ decreases or increases exponentially in $T$, as we discuss later. 

For homogeneous systems  of interest the value of $k$ is quite mild. For example, taking $\lambda_\times(T) = 0$, $M=10^6$ and $\epsilon_{\mbox{\tiny TD}} = 10^{-9}$ the bound returns $k=19$ (rounding up $18.2$). Directly solving numerically the original bound in Lemma~\ref{lem:timediscretization} for  $k$ gives $k=18$, so the analytical bound is relatively tight.

The recursively defined dynamics generates terms such as $T_k(Ah)^m$, which approximate $(e^{Ah})^m = e^{Amh}$ for all $m=0,1,\dots ,M$. For the condition number analysis we need an estimate on how well the truncated expression $T_k(Ah)^m$ approximates the exact exponential $e^{Amh}$ for the continuous dynamics. The next result determines the relative error  for these matrices in the discrete dynamics. 

\begin{lem}[Truncation error for dynamics] \label{lem:rel-Tm-error}
Let $k=k_{+,\times}$ obey the time-discretization condition given by Eq.~\eqref{eq:kconditiongeneral} with error $\epsilon_{\mbox{\tiny TD}}$, and let $T_k(Ah)$ be defined as in Eq.~\eqref{eq:definitionsTandS}. Then for all $m=0,1,\dots M$ we have the following relative error for the norm $\|T_k(Ah)^m\|$ relative to $\|e^{Ahm}\|$.
\begin{equation}
\|T_k(Ah)^m\| \le \|e^{Ahm}\| (1+\epsilon_{\mbox{\tiny TD}}),
\end{equation}
for multiplicative error, and
 \begin{align}
	\|T_k(Ah)^m\| \le \|e^{Ahm}\| \left (1+\frac{\epsilon_{\mbox{\tiny TD}}}{  x_{\max} } \right ).
\end{align}
for additive error.
\end{lem}
\begin{proof}
    We have that
    \begin{align}
        \|T_k(Ah)^m\| &\le \|e^{Amh}\| + \|T_k(Ah)^m - e^{Amh}\| \nonumber\\
        &=\|e^{Amh}\| + \|(T_k(Ah)^m - e^{Amh})e^{-Amh} e^{Amh}\| \nonumber\\
        &=\left (1 + \|(T_k(Ah)^m - e^{Amh})e^{-Amh}\| \right) \|e^{Amh}\| \nonumber\\
                &=\left (1 + \frac{(e-1)me^2}{(k+1)!} \right) \| e^{Amh}\| \nonumber\\
                                &=\left (1 + \frac{Me^3}{(k+1)!} \right) \|e^{Amh}\|.
    \end{align}
For multiplicative error, $k = k_{\times}$,  Eq.~\eqref{eq:krequirement} implies
    	\begin{align}
  \|T_k(Ah)^m\| \le \|e^{Ahm}\| (1+\epsilon_{\mbox{\tiny TD}}).
    	\end{align}
    	 For additive error, $k=k_+$, from Eq.~\eqref{eq:krequirementadditive} we have 
    	 \begin{align}
    	   \|T_k(Ah)^m\| \le \|e^{Ahm}\| \left (1+\frac{\epsilon_{\mbox{\tiny TD}}}{  x_{\max} } \right ).
    	   \end{align}
\end{proof}
This extends a result given for multiplicative error in Ref.~\cite{krovi2022improved}.

Note that so far we discussed time-discretization errors at the level of unnormalized solution vectors, but these can be readily translated into errors on the normalized quantum states. 
\begin{lem}
\label{lem:relativetounitvectorerror}
  For any two vectors
    \begin{equation}
        \ket{z} = \frac{1}{\|\v{z}\|} (\v{z}^1,\dots,\v{z}^J), \quad \quad   \ket{\tilde{z}} = \frac{1}{\|\tilde{\v{z}}\|} (\tilde{\v{z}}^1,\dots,\tilde{\v{z}}^J),
    \end{equation}
    if the relative error of each component is upper bounded as $\| \v{z}^j - \tilde{\v{z}}^j\| \leq a\| \tilde{\v{z}}^j\| + b$ for all $j=1,\dots, J$, then it follows that for the corresponding normalized vectors we have
    \begin{equation}
 \|\ket{z} - \ket{\tilde{z}}\| \le 2 \sqrt{a^2 +2a b \sum_j \frac{\|\tilde{\v{z}}^j\|}{\|\tilde{\v{z}}\|^2} + J\frac{b^2}{\|\tilde{\v{z}}\|^2}}.
\end{equation}
\end{lem}
\begin{proof}
\begin{align}
\|\ket{z} - \ket{\tilde{z}}\| &= \left\| \frac{\v{z}}{\| \v{z}\|} - \frac{\tilde{\v{z}}}{\| \tilde{\v{z}}\|}   \right\| \leq \left\| \frac{\v{z}}{\| \v{z}\|} - \frac{\v{z}}{\| \tilde{\v{z}}\|}  \right\| + \left\| \frac{\v{z}}{\| \tilde{\v{z}}\|} - \frac{\tilde{\v{z}}}{\| \tilde{\v{z}}\|}\right\|  \nonumber\\ & =    \frac{\left|\|\v{z}\| - \| \tilde{\v{z}}\|\right|}{\| \tilde{\v{z}}\|} +  \frac{\|\v{z} - \tilde{\v{z}}\|}{\|\tilde{\v{z}}\|}\leq   2 \frac{\|\v{z} - \tilde{\v{z}}\|}{\|\tilde{\v{z}}\|}.
\end{align}
\begin{equation}
    \|\v{z} - \tilde{\v{z}}\| = \left(\sum_{j} \| \v{z}^j - \tilde{\v{z}}^j\|^2\right)^{1/2} \leq   \left(\sum_j (a\| \tilde{\v{z}}^j\|+b)^2\right)^{1/2} = \sqrt{a^2 \|\tilde{\v{z}}\|^2 +2ab \sum_j \|\tilde{\v{z}}^j\| + J b^2},
\end{equation}
and therefore,
\begin{equation}
\|\ket{z} - \ket{\tilde{z}}\| \le 2 \sqrt{a^2 +2a b \sum_j \frac{\|\tilde{\v{z}}^j\|}{\|\tilde{\v{z}}\|^2} + J\frac{b^2}{\|\tilde{\v{z}}\|^2}}.
\end{equation}
\end{proof}
Recall the definitions (we set $p=0$ for the history state)
\begin{equation}
\label{eq:definitionxHapprox}
\v{x}_H^{\epsilon_{\mbox{\tiny TD}}} = (\v{x}^0, 0, \dots, 0, \v{x}^1,0, \dots, 0,\v{x}^M),
\end{equation}
\begin{equation}
\label{eq:definitionxH}
\v{x}_H = (\v{x}(0), 0, \dots, 0, \v{x}(h),0, \dots, 0,\v{x}(Mh)),
\end{equation}
 and recall that we denote by kets the corresponding normalized states. Then, defining the corresponding normalized coherent encodings of these vectors we have the following corollary.
\begin{cor}
\label{cor:unitvectorerrors}
Let $\ket{x_H}$ and $\ket{x^{\epsilon_{\mbox{\tiny TD}}}_H}$ be the normalized states associated to the vectors in Eq.~\eqref{eq:definitionxH} and Eq.~\eqref{eq:definitionxHapprox}, respectively. Assume that $\|\v{x}(mh) - \v{x}^m\| \le \epsilon_{\mbox{\tiny TD}} \|\v{x}(mh)\|$ for all $m=1,\dots, M$ (multiplicative error). 
Then we have that
\begin{align}
    \| \ket{x^{\epsilon_{\mbox{\tiny TD}}}_{H}} - \ket{x_{H}} \| \leq 2 \epsilon_{\mbox{\tiny TD}} \quad  \quad
    \| \ket{x^M} - \ket{x(Mh)} \| \leq 2 \epsilon_{\mbox{\tiny TD}},
\end{align}

Assume that $\|\v{x}(mh) - \v{x}^m\| \le \epsilon_{\mbox{\tiny TD}}$ for all $m=1,\dots, M$ (additive error). 
Then we have that
\begin{align}
    \| \ket{x^{\epsilon_{\mbox{\tiny TD}}}_{H}} - \ket{x_{H}} \| &\leq 2\epsilon_{\mbox{\tiny TD}} \frac{\sqrt{M}}{\|\v{x}_H\|} = 2\epsilon_{\mbox{\tiny TD}} \frac{1}{\sqrt{\sum_m \|\v{x}(mh)\|^2/M}} \nonumber \\
    \| \ket{x^M} - \ket{x(Mh)} \| &\leq 2\epsilon_{\mbox{\tiny TD}} \frac{1}{\|\v{x}(Mh)\|} .
\end{align}
Therefore, 
\begin{align}
    \| \ket{x^{\epsilon_{\mbox{\tiny TD}}}_{H}} - \ket{x_{H}} \| \leq 2 \epsilon_{\mbox{\tiny TD}}f^{+,\times}_H \quad \quad
    \| \ket{x^M} - \ket{x(Mh)} \| \leq 2 \epsilon_{\mbox{\tiny TD}}f^{+,\times}_T,
\end{align}
where 
\begin{align}
           f^\times_{H} := f^\times_{T} := 1, \quad 
        f^+_H := \frac{1}{\sqrt{\sum_m \|\v{x}(mh)\|^2/M}}, \quad 
        f^+_T :=\frac{1}{\|\v{x}(Mh)\|}.
\end{align}
\end{cor}
\begin{proof}
    The result follows by applying Lemma~\ref{lem:relativetounitvectorerror} to the time-discretization error Lemmas~\ref{lem:timediscretization}, \ref{lem:timediscretizationadditive}.
\end{proof}

We next turn to embedding the ODE system into a linear system of equations.

\section{Linear system embedding of the dynamics}
\label{sec:linearsystemembedding}

In what follows we assume tunable idling parameters at each step $m=0,1,\dots, M$ so as to cover both the case where we want to output the history state,  the case where we want to output the solution state, as well as generalized settings, as anticipated in the main text. These idling parameters are specified by integers $p_0, p_1, \dots , p_M$. The case $p_m=0$ except for $p_M=p$ corresponds to the solution state output case, while $p_m=p$ for all $m$ corresponds to the history state output case.

The linear system reduction for the ODE protocol in Ref.~\cite{berry2017quantum}, here in a generalized version, has the form $L= L_1 + L_2 +L_3$ where
\begin{align}
L_1&:=  \sum_{m=0}^{M-1} \sum_{j=0}^{p_m+k} \ketbra{m,j}{m,j}\otimes I + \theta(p_M)\sum_{j=1}^{p_M} \ketbra{M,j}{M,j} \otimes I, \nonumber \\
L_2 &:=- \sum_{m=0}^{M} \theta(p_m) \sum_{j=1}^{p_m}  \ketbra{m,j}{m,j-1}\otimes I -\sum_{m=0}^{M-1} \sum_{j=1}^{k}  \ketbra{m,p_m+j}{m,p_m+j-1} \otimes \frac{Ah}{j},  \nonumber \\
L_3 &:= - \sum_{m=0}^{M-1}  \sum_{j=0}^k   \ketbra{m+1,0}{m,p_m+j}  \otimes I. 
\label{eq:L}
\end{align}
Here we use the Heaviside function $\theta(p) =1$ for $p>0$ and zero otherwise.
The integers $\{p_m\}$ govern the filtering mentioned, as they amplify the solution vector at specified times.  Note by setting $p_m=0$ for all $m=0,\dots,M-1$ and $p_M=p$ we recover the simplified form of the matrix $L$ given in Eq.~\eqref{eq:linearsystem}, but with a different indexing. The difference here is that the idling phase is entirely encoded in the $j$ index, whereas in  in Eq.~\eqref{eq:linearsystem} it was more efficiently encoded in all $m\geq M$ and $j=0,\dots,k$ indices. This amounts to a simple relabeling. To compute the condition number bound it is convenient to use the form above, whereas to construct the block-encodings it is convenient to adopt the form used in the main text.

It can be useful to visualize the matrix $L$ in the computational basis:
\begin{equation}
	\label{eq:berrymatrix}
	L =
	\begin{bmatrix} 
	I\\
				 -I & I\\
						   &  \ddots & \ddots \\
							&	  & -I & I \\
	&&&	- A h & I &							&							&  	&							&							&  	&							&							\\ 	&&&	
		& -\frac{Ah}{2} & I &							&							&  	&							&							&  	&							\\ 	&&&	
		& 							& \frac{- A h}{3} & I &							&							&  	&							&							&  	\\	&&&	
		&							&							&   \ddots & \ddots& &&&&   \\	&&&	
		&							&							&   			&-\frac{Ah}{k}& I   \\ 	&&&	
		-I			&			-I				&		-I				&		\dots & - I 				& - I & I &\\	 	&&&	
					&							&						&		 & 				&  & -I & I\\ 	&&&	
							&							&						&		 & 				&  &  &   \ddots & \ddots \\ 	&&&	
								&							&						&		 & 				& &  &  & -I & I\\ 	&&&	
		&	&				&					&		 &  				& & && - A h & I &&	\\ 	&&&		&	&					&					&		 &  &				& & && \ddots & \ddots &&	
	\end{bmatrix}. 
\end{equation}
$L$ has $M$ repeated blocks. The  $m^{th}$ block contains $p_m$ repetitions of the $-I, I$ `idling' and $k$ sub-blocks $-Ah/j$, $j=1,\dots,k$, generating the Taylor series components. Note that $A$ only appears to the first power, and hence \eqref{eq:berrymatrix} is a relatively simple encoding compared to alternative ones involving matrix exponentiation~\cite{berry2022quantum} or matrix inversion subroutines~\cite{krovi2022improved}.

A simple extension of Eq.~(16)-(24) in Ref.~\cite{berry2017quantum} shows that if we consider $L\v{y} = \v{c}$ with the vector $\v{c}$ given in terms of the following subvectors
\begin{align}
\label{eq:cdefinition}
  \v{c}^{(0,0)} &= \v{x}^0 \nonumber \\  
      \v{c}^{(m,p_m+1)}&=h\v{b} \mbox{ for all } m=0, \dots , M-1\nonumber \\
      \v{c}^{(m,r)} &=\v{0} \mbox{ otherwise},
\end{align}
then we get
\small
\begin{equation}
\label{eq:vectorberry2}
	\v{y} = L^{-1} \v{c} = [ \underbrace{\v{y}^{(0,0)}, \dots,  \v{y}^{(0,p_0)}}_{p_0+1}, \v{y}^{(0,p_0+1)}, \dots, \v{y}^{(0,p_0+k)},\underbrace{\v{y}^{(1,0)},  \dots, \v{y}^{(1,p_1)}}_{p_1+1}, \dots, \v{y}^{(M-1,p_{M+1}+k)}, \underbrace{\v{y}^{(M,0)}, \dots , \v{y}^{(M,p_M)}}_{p_M+1}],
\end{equation}
\normalsize
where for each $m$ we have $\v{y}^{(m,j)} = \v{x}^m$ for all $j=0,\dots, p_m$. Here $\v{x}^m$ are the components of the history state in the truncated Taylor series approach (Eq.~\eqref{eq:truncatedtaylorserieapprox}). Instead, $\v{y}^{(m,j)}$ for $j = p_m+1,\dots, p_m +k$ is, with the above labeling, junk data that can be removed by a post-selection measurement on the qubits labeling the indices $(m,j)$, or via amplitude amplification. We show that this form of $\v{y}$ is the result of the matrix inversion in Remark~\ref{rmk:solutionlinearsystem}. Note that for the choice $p_m = 0$ for all $m$ we recover the history state output for the $(m,0)$ labels. For the choice $p_m =0$ for $m<M$ and $p_m = p$ for $m=M$ we obtain a vector that can be post-processed to output the solution state with bounded probability by amplifying the $m= M$ component of the above state, if we take $p$ large enough. The exact value of $p$ required will be analyzed in more detail in Sec.~\ref{sec:successprobabilityanalysis}.

\section{Condition number upper bound of the linear embedding matrix}
\label{sec:conditionnumberanalysis}

We next turn to the condition number of the matrix $L$, which determines the time-complexity of the algorithm. Previously Ref.~\cite{berry2017quantum} presented upper bounds on the condition number of $L$ for a special choice of $\{p_m\}$ and under several assumptions about $A$, in particular: that $A$ has $\alpha(A) \le 0$; that $A$ is diagonalizable; and that one has a bound on the condition number $\kappa_{V}$ of the matrix $V$ that diagonalizes $A$. Here we extend the scope of that analysis and strengthen the bounds. This will give exponential savings in $\log N$ for certain classes of matrices (similarly to Ref.~\cite{krovi2022improved}, but in the context of a simpler linear embedding), as well as a potential sublinear scaling with time for the large class of stable dynamics.

\begin{theorem}[Condition number bound -- General case]
\label{thm:conditionnumbergeneral}
Given $A\in \mathbb{C}^{N\times N}$, let $C(t)$ be an upper bound to $\|e^{A t} \|$.  Assume a Taylor truncation scale $k=k_{+,\times}$ which gives rise to either a multiplicative or additive error $\epsilon_{\mbox{\tiny TD}}$ for the norm distance of the discretized dynamics $\v{x}^m$ to the exact dynamics $\v{x}(mh)$. Choose idling parameters $p_0, p_1, \dots, p_M$ and for convention set $p_{-1}=0$. Then the condition number of $L$ in Eq.~\eqref{eq:L} is upper bounded as 
\small
\begin{align}
\! \! \!   \kappa_L &\leq \left[\sum_{m=0}^{M} \left( \left (1+ \frac{\epsilon_{\mbox{\tiny TD}}}{c_{+,\times}} \right)^2 (p_m+I_0(2)) \sum_{l=0}^m C(l h)^2 \left( 1 +  p_{l-1} \left(1+ \frac{\epsilon_{\mbox{\tiny TD}}}{c_{+,\times}}\right )^2 C(h)^2  + g(k) \right ) + \frac{p_m(p_m+1)}{2} \right) + \right .\nonumber \\
&\left . +( \sum_{m=0}^{M}p_m + M k) (I_0(2)-1)\right]^{1/2}   (\sqrt{k+1}+2).
\end{align}
\normalsize
Here $I_0(2) \approx 2.2796$, with $I_0(x)$ the modified Bessel function of order zero, and	where $
	    g(k) :=  \sum_{s=1}^k \left(s! \sum_{j=s}^k 1/j!\right)^2 \le ek$,  $c_+ = x_{\mathrm{max}}$ and $c_\times = 1$.
\end{theorem}
\begin{proof}
A general vector will be divided up into sectors as follows
\begin{equation}
\v{\gamma} = [ \v{\gamma}^{(0,0)} \, \v{\gamma}^{(0,1)}\, \cdots \v{\gamma}^{(0,p_0)} \, \v{\gamma}^{(0,p_0+1)}\, \cdots \v{\gamma}^{(0,p_0+k)} \, \v{\gamma}^{(1,0)}\, \v{\gamma}^{(1,1)} \, \cdots \v{\gamma}^{(1,p_1)}\, \cdots \v{\gamma}^{(1,p_1+k)} \dots \v{\gamma}^{(M,p_M)} ],
\end{equation}
where we use the idling parameters $(p_0, p_1, \dots , p_M)$.
We can bound $\|L\|$ via $\|L \| \le \|L_1\| + \|L_2\| + \|L_3\|$. Firstly, it is clear that $\|L_1\| =1$. Secondly,
\begin{equation}
    L_2L_2^\dagger = \sum_{m=0}^{M} \theta(p_m) \sum_{r=1}^{p_m}  \ketbra{m,r}{m,r} \otimes I +\sum_{m=0}^{M-1} \sum_{j=1}^{k}  \ketbra{m,p_m+j}{m,p_m+j} \otimes \frac{(Ah)^2}{j^2},
\end{equation}
and so $\|L_2\| = \max\{ \|Ah\| , 1\} =1$. Similarly (see \cite{berry2017quantum}, proof of Lemma~4) we have $\|L_3\| = \sqrt{k+1}$. Therefore, we have $\|L\| \le \sqrt{k+1} + 2$.

We next estimate the norm $\|L^{-1}\|$ which is obtained from
\begin{equation}
\|L^{-1}\| = \max_{\|\v{\beta}\|=1} \| L^{-1} \v{\beta}\|.
\end{equation}
By writing $L \v{\gamma} = \v{\beta}$ we have that $\|L^{-1}\| = \max_{\|\v{\beta}\|=1} \| \v{\gamma}\|$, where $\v{\gamma}$ is implicitly a function of $\v{\beta}$. We now provide an upper bound on $\|L^{-1}\|$ via bounding $\|\v{\gamma}\|$. For $m$ obeying $0 \le m <M$ we have the following system of equations:
\begin{align}
\v{\gamma}^{(0,0)}&= \v{\beta}^{(0,0)} \nonumber \\
-\v{\gamma}^{(m,0)} + \v{\gamma}^{(m,1)}&= \v{\beta}^{(m,1)} \nonumber\\ 
-\v{\gamma}^{(m,1)} + \v{\gamma}^{(m,2)}&= \v{\beta}^{(m,2)} \nonumber\\ 
\vdots &  \nonumber \\
-\v{\gamma}^{(m,p_m-1)} + \v{\gamma}^{(m,p_m)}&= \v{\beta}^{(m,p_m)} \nonumber\\
-(Ah)\v{\gamma}^{(m,p_m)} + \v{\gamma}^{(m,p_m+1)} &= \v{\beta}^{(m,p_m+1)} \nonumber\\ 
-\frac{(Ah)}{2}\v{\gamma}^{(m,p_m+1)} + \v{\gamma}^{(m,p_m+2)} &= \v{\beta}^{(m,p_m+2)} \nonumber\\ 
-\frac{(Ah)}{3}\v{\gamma}^{(m,p_m+2)} + \v{\gamma}^{(m,p_m+3)} &= \v{\beta}^{(m,p_m+3)} \nonumber\\ 
\vdots &  \nonumber \\
-\frac{(Ah)}{k}\v{\gamma}^{(m,p_m+k-1)} + \v{\gamma}^{(m,p_m+k)} &= \v{\beta}^{(m,p_m+k)} \nonumber \\
-\v{\gamma}^{(m,p_m)} -\v{\gamma}^{(m,p_m+1)}\cdots -\v{\gamma}^{(m,p_m+k)} + \v{\gamma}^{(m+1,0)} &= \v{\beta}^{(m+1,0)}.
\end{align}
While for $m=M$ we only have the $p_M$ equations for idling parameters and no Taylor series terms.

We can formally solve for $\v{\gamma}$ in the above system of equations to get
\begin{align}
 \v{\gamma}^{(m,1)}&= \v{\beta}^{(m,1)} + \v{\gamma}^{(m,0)} \nonumber\\ 
 \v{\gamma}^{(m,2)}&= \v{\beta}^{(m,2)} + \v{\beta}^{(m,1)} + \v{\gamma}^{(m,0)} \nonumber\\ 
\vdots &  \nonumber \\
 \v{\gamma}^{(m,p_m)}&= \v{\beta}^{(m,p_m)}  + \v{\beta}^{(m,p_m-1)} + \cdots +\v{\beta}^{(m,2)} + \v{\beta}^{(m,1)} + \v{\gamma}^{(m,0)},
\end{align}
for the idling parts. Thus, for any $m=0,\dots , M$ and any $r=1,\dots , p_m$ we have
\begin{align}\label{eqn:idle-copying}
    \v{\gamma}^{(m,r)}&= \v{\gamma}^{(m,0)}+\sum_{s=1}^r\v{\beta}^{(m,s)}   .
\end{align}

For the Taylor series parts we have
\begin{align}
 \v{\gamma}^{(m,p_m+1)} &= \v{\beta}^{(m,p_m+1)}+(Ah)\v{\gamma}^{(m,p_m)} \nonumber\\ 
 \v{\gamma}^{(m,p_m+2)} &= \v{\beta}^{(m,p_m+2)} + \frac{(Ah)}{2}\v{\gamma}^{(m,p_m+1)} \nonumber\\ 
 \v{\gamma}^{(m,p_m+3)} &= \v{\beta}^{(m,p_m+3)} +\frac{(Ah)}{3}\v{\gamma}^{(m,p_m+2)} \nonumber\\ 
\vdots &  \nonumber \\
 \v{\gamma}^{(m,p_m+k)} &= \v{\beta}^{(m,p_m+k)} +\frac{(Ah)}{k}\v{\gamma}^{(m,p_m+k-1)} \nonumber \\
 \v{\gamma}^{(m+1,0)} &= \v{\beta}^{(m+1,0)} +\v{\gamma}^{(m,p_m)} +\v{\gamma}^{(m,p_m+1)}+ \cdots +\v{\gamma}^{(m,p_m+k)} .
 \label{eq:taylorsumproof}
\end{align}
Solving the previous expression in terms of $\v{\gamma}^{(m,p_m)}$ gives
\begin{align}
 \v{\gamma}^{(m,p_m+1)} &= \v{\beta}^{(m,p_m+1)}+(Ah)\v{\gamma}^{(m,p_m)} \nonumber\\ 
 \v{\gamma}^{(m,p_m+2)} &= \v{\beta}^{(m,p_m+2)} + \frac{(Ah)}{2} \v{\beta}^{(m,p_m+1)}+\frac{(Ah)^2}{2} \v{\gamma}^{(m,p_m)} \nonumber\\ 
 \v{\gamma}^{(m,p_m+3)} &= \v{\beta}^{(m,p_m+3)} +\frac{(Ah)}{3}\v{\beta}^{(m,p_m+2)} + \frac{(Ah)^2}{3\cdot2} \v{\beta}^{(m,p_m+1)}+\frac{(Ah)^3}{3\cdot2\cdot1} \v{\gamma}^{(m,p_m)} \nonumber\\ 
\vdots &  \nonumber \\
 \v{\gamma}^{(m,p_m+k)} &= \v{\beta}^{(m,p_m+k)} + \frac{(Ah)}{k}\v{\beta}^{(m,p_m+k-1)}\nonumber \\ & +\cdots +\frac{(Ah)^{k-2}}{k\cdot (k-1)\cdots 3}\v{\beta}^{(m,p_m+2)} +\frac{(Ah)^{k-1}}{k\cdot(k-1)\cdots 3\cdot 2}\v{\beta}^{(m,p_m+1)} +\frac{(Ah)^k}{k!}\v{\gamma}^{(m,p_m)}.
\end{align}
This implies that for any $m=0,1,\dots, M-1$ and any $j=1,\dots, k$ we have
\begin{equation}
\label{eq:jsumproof}
\v{\gamma}^{(m,p_m+j)} = \frac{1}{j!} (Ah)^j \v{\gamma}^{(m,p_m)} + \sum_{s=1}^j \frac{s!}{j!} (Ah)^{j-s} \v{\beta}^{(m,p_m+s)}.
\end{equation}

We now iteratively solve over the $m$ index as follows. Firstly,
\begin{align}
\v{\gamma}^{(m,0)} & \stackrel{\eqref{eq:taylorsumproof}}{=}  \sum_{j=0}^k \v{\gamma}^{(m-1,p_{m-1}+j)} + \v{\beta}^{(m,0)}\nonumber \\
&\stackrel{\eqref{eq:jsumproof}}{=}  \sum_{j=0}^k \frac{1}{j!} (Ah)^j \v{\gamma}^{(m-1,p_{m-1})} +  \sum_{j=1}^k\sum_{s=1}^j \frac{s!}{j!} (Ah)^{j-s} \v{\beta}^{(m,p_{m-1}+s)} + \v{\beta}^{(m,0)} \nonumber \\
&\stackrel{\eqref{eqn:idle-copying}}{=} T_k(Ah) \v{\gamma}^{(m-1,0)} + \sum_{r=1}^{p_{m-1}}T_k(Ah)  \v{\beta}^{(m-1,r)} +  \sum_{j=1}^k\sum_{s=1}^j \frac{s!}{j!} (Ah)^{j-s} \v{\beta}^{(m-1,p_{m-1}+s)} + \v{\beta}^{(m,0)} .
\end{align}
We define 
\begin{equation}
\label{eq:Bm-1proof}
    \v{B}^{(m-1)} :=\sum_{r=1}^{p_{m-1}}T_k(Ah)  \v{\beta}^{(m-1,r)} +\sum_{j=1}^k\sum_{s=1}^j \frac{s!}{j!} (Ah)^{j-s} \v{\beta}^{(m-1,p_{m-1}+s)},
\end{equation}
so that for any $m=1,2, \dots, M$ we have
\begin{equation}\label{eqn:m-iterate}
  \v{\gamma}^{(m,0)} =  T_k(Ah) \v{\gamma}^{(m-1,0)} +  \v{B}^{(m-1)}  + \v{\beta}^{(m,0)}.
\end{equation}
Since we have the condition $\v{\gamma}^{(0,0)} = \v{\beta}^{(0,0)}$ it follows that
\begin{align}
  \v{\gamma}^{(1,0)} &=  T_k(Ah) \v{\beta}^{(0,0)} +  \v{B}^{(0)}  + \v{\beta}^{(1,0)}\nonumber \\
    \v{\gamma}^{(2,0)} &=  T_k(Ah) (T_k(Ah) \v{\beta}^{(0,0)} +  \v{B}^{(0)}  + \v{\beta}^{(1,0)}) + \v{B}^{(1)} + \v{\beta}^{(2,0)} \nonumber \\
     \v{\gamma}^{(3,0)} &=  T_k(Ah) [T_k(Ah) (T_k(Ah) \v{\beta}^{(0,0)} +  \v{B}^{(0)}  + \v{\beta}^{(1,0)}) + \v{B}^{(1)} + \v{\beta}^{(2,0)} ] + \v{B}^{(2)} + \v{\beta}^{(3,0)},
\end{align}
and so iterating we find that
\begin{equation}
\v{\gamma}^{(m,0)}= \sum_{r=0}^m T_k(Ah)^{m-r} \left ( \v{\beta}^{(r,0)} + \v{B}^{(r-1)} \right ),
\end{equation}
where we set $\v{B}^{(-1)} = \v{0}$ by setting $\v{\beta}^{(-1,r)} \equiv 0$ for all $r$.

We now make use of Lemma~\ref{lem:rel-Tm-error} to estimate the following vector norm.
\begin{align}
\|\v{\gamma}^{(m,0)}\|&\le \sum_{r=0}^m \|T_k(Ah)^{m-r} \|\left ( \|\v{\beta}^{(r,0)}\| + \|\v{B}^{(r-1)}\| \right ) \nonumber \\
&\le(1+ \epsilon_{\mbox{\tiny TD}}/ c_{+,\times}) \sum_{r=0}^m \| e^{Ah(m-r)}\| \left ( \|\v{\beta}^{(r,0)}\| + \|\v{B}^{(r-1)}\| \right ) \nonumber \\
&\leq (1+ \epsilon_{\mbox{\tiny TD}}/ c_{+,\times}) \sum_{r=0}^m C(mh-rh)\left ( \|\v{\beta}^{(r,0)}\| + \|\v{B}^{(r-1)}\| \right ),
\end{align}
 where $c_{\times}=1$, $c_+ = x_{\mathrm{max}}$. Since $\|A h \| \leq 1$ we have that
	\small \begin{align}\label{eqn:m0-bound}
		\| \v{\gamma}^{(m,0)} \| &\leq  (1+ \epsilon_{\mbox{\tiny TD}}/ c_{+,\times})  \sum_{r=0}^m C(mh-rh) \left [ \|\v{\beta}^{(r,0)}\| +   \sum_{j=1}^k \sum_{s=1}^j\frac{s!}{j!} \| \v{\beta}^{(r-1,p_{r-1}+s)}\| +\sum_{w=1}^{p_{r-1}}\|T_k(Ah)\|  \|\v{\beta}^{(r-1,w)}\| \right ]  \nonumber\\
		&\leq  (1+ \epsilon_{\mbox{\tiny TD}}/ c_{+,\times})  \sum_{r=0}^m C(mh-rh) \left [ \|\v{\beta}^{(r,0)}\| +   \sum_{j=1}^k \sum_{s=1}^j\frac{s!}{j!} \| \v{\beta}^{(r-1,p_{r-1}+s)}\| +(1+ \epsilon_{\mbox{\tiny TD}}/c_{+,\times}) C(h) \sum_{w=1}^{p_{r-1}} \|\v{\beta}^{(r-1,w)}\| \right ] 
	\end{align}
	\normalsize
	where we used $\| T_k(Ah)\| \leq C(h) (1+ \epsilon_{\mbox{\tiny TD}}/ c_{+,\times}) $, again from Lemma~\ref{lem:rel-Tm-error}, which holds for either the multiplicative error or additive error case. Now,
 \begin{align}
 \left [   \sum_{j=1}^k \sum_{s=1}^j\frac{s!}{j!} \| \v{\beta}^{(r,p_r+s)}\|\right ] 	&	=  \left[ \sum_{s=1}^1 \frac{s!}{1!} \| \v{\beta}^{(r,p_r+s)} \| + \sum_{s=1}^2 \frac{s!}{2!} \| \v{\beta}^{(r,p_r+s)} \| + \dots + \sum_{s=1}^k \frac{s!}{k!} \| \v{\beta}^{(r,p_r+s)} \| \right] \nonumber \\
	& =\left[ 1! \left( \frac{1}{1!} + \frac{1}{2!} + \dots + \frac{1}{k!}\right)\|\v{\beta}^{(r,p_r+1)}  \| \right. \nonumber \\  & \left. + 2! \left( \frac{1}{2!} + \dots + \frac{1}{k!}\right) \|\v{\beta}^{(r,p_r+2)} \| + \dots + k! \frac{1}{k!} \|\v{\beta}^{(r,p_r+k)}  \| \right] \nonumber \\ 
	& =  \sum_{s=1}^k s! \sum_{j=s}^k \frac{1}{j!} \|\v{\beta}^{(r,p_r+s)} \| \nonumber \\
		& =  \sum_{s=1}^k f(s,k) \|\v{\beta}^{(r,p_r+s)} \|,
	\end{align}
where we define	
\begin{equation}
    f(s,k):= s! \sum_{j=s}^k \frac{1}{j!} = \sum_{j=s}^k \frac{1}{j(j-1) \cdots (s+1)}. 
\end{equation}
Thus, we have
\small
	\begin{align}
\frac{\|\v{\gamma}^{(m,0)}\|}{(1+ \epsilon_{\mbox{\tiny TD}}/ c_{+,\times}) } \leq  \sum_{r=0}^m C(mh-rh) \left [ \|\v{\beta}^{(r,0)}\| +   \sum_{s=1}^k f(s,k) \| \v{\beta}^{(r-1,p_{r-1}+s)}\| +(1+ \epsilon_{\mbox{\tiny TD}}/ c_{+,\times})  C(h) \sum_{w=1}^{p_{r-1}} \|\v{\beta}^{(r-1,w)}\| \right ]. 
	\end{align}
\normalsize
The RHS of this is in the form of an inner product $\v{V}^{(m)} \cdot \v{W}^{(m)}$, between vectors with components,
\begin{align}
\v{V}^{(m)}&:=\left (\oplus_{r=0}^m \|\v{\beta}^{(r,0)}\| \oplus_{r=0}^{m} \oplus_{w=1}^{p_{r-1}} \|\v{\beta}^{(r-1,w)}\| \oplus_{r=0}^m  \oplus_{s=1}^k\|\v{\beta}^{(r-1,p_{r-1}+s)} \| \right ) \nonumber\\
    \v{W}^{(m)}&:=\left (\oplus_{r=0}^m C(mh-rh) \oplus_{r=0}^{m} \oplus_{w=1}^{p_{r-1}}  (1+ \frac{\epsilon_{\mbox{\tiny TD}}}{ c_{+,\times}}) C(mh-rh)C(h) \oplus_{r=0}^m  \oplus_{s=1}^k \left[ C(mh-rh)f(s,k) \right ] \right ),
\end{align}
where we set $p_{-1} =0$, and use direct sum notation so that, for example, $(\oplus_{r=0}^3 z_r) = (z_0,z_1,z_2,z_3)$. Since $\|\v{\beta}\| =1$ it follows that $\|\v{V}^{(m)}\|\le 1$ for any $m=1,\dots , M$. We now  use Cauchy-Schwarz on $\v{V}^{(m)} \cdot \v{W}^{(m)}$ and get that for any $m=1, 2, \dots ,M$
	\begin{align}\label{eqn:m-bound}
		\| \v{\gamma}^{(m,0)}\|^2  &\leq \left (1+ \frac{\epsilon_{\mbox{\tiny TD}}}{c_{+,\times}} \right)^2 \sum_{r=0}^m C(r h)^2 \left( 1 +  p_{r-1} \left (1+ \frac{\epsilon_{\mbox{\tiny TD}}}{c_{+,\times}} \right)^2 C(h)^2  + g(k) \right )
	\end{align}
	where 
	\begin{equation}
	    g(k) = \sum_{s=1}^k f(s,k)^2 = \sum_{s=1}^k \left(s! \sum_{j=s}^k 1/j!\right)^2.
	\end{equation}
For $m=0$ we have that $\v{\gamma}^{(0,0)} = \v{\beta}^{(0,0)}$ and so $\|\v{\gamma}^{(0,0)}\| \le 1$. If we define $p_{-1}=0$ then we can alternatively use the bound given in~(\ref{eqn:m-bound}) with $m=0$, which is slightly weaker. Note that since $f(s,k)\leq e$ one has $g(k) \leq e k$,  but this upper bound overestimates $g(k)$ by almost a factor of $2$, so we shall not use it. 

We next bound $\|\v{\gamma}^{(m,r)}\|$ for $r=1,\dots , p_m$ and also  $\|\v{\gamma}^{(m,p_m+j)}\|$ for $j=1,\dots , k$.

Firstly, from Eq.~\eqref{eqn:idle-copying},
\begin{align}
    \|\v{\gamma}^{(m,r)}\| &\le \|\v{\gamma}^{(m,0)}\| + \sum_{s=1}^r \|\v{\beta}^{(m,s)}\|,
\end{align}
for $r=1,2, \dots , p_m$. We follow a similar line to before, and have from Eq.~\eqref{eqn:m0-bound}:
\small
	\begin{align}
\|\v{\gamma}^{(m,r)}\| & \leq  (1+ \frac{\epsilon_{\mbox{\tiny TD}}}{c_{+,\times}})\sum_{l=0}^m C(mh-lh) \left [ \|\v{\beta}^{(l,0)}\| +   \sum_{s=1}^k f(s,k) \| \v{\beta}^{(l-1,p_{l-1}+s)}\| +(1+ \frac{\epsilon_{\mbox{\tiny TD}}}{ c_{+,\times}}) C(h) \sum_{w=1}^{p_{l-1}} \|\v{\beta}^{(l-1,w)}\| \right ] \nonumber\\ 
&  + \sum_{s=1}^r \|\v{\beta}^{(m,s)}\|,
	\end{align}
	\normalsize
and using Cauchy-Schwarz in an analogous way to Eq.~\eqref{eqn:m-bound} we have
	\begin{align}\label{eqn:mr-norm}
		\| \v{\gamma}^{(m,r)}\|^2  &\leq (1+ \frac{\epsilon_{\mbox{\tiny TD}}}{c_{+,\times}})^2 \sum_{l=0}^m C(l h)^2 \left( 1 +  p_{l-1} (1+\epsilon_{\mbox{\tiny TD}})^2 C(h)^2  + g(k) \right ) + r,
	\end{align}
	for any $r=0,1,2, \dots, p_m$.
Finally, using Eq.~\eqref{eq:jsumproof} and Eq.~\eqref{eqn:idle-copying},
\begin{align}
\|\v{\gamma}^{(m,p_m+j)}\| &\le \frac{1}{j!}  \|\v{\gamma}^{(m,p_m)}\| + \sum_{s=1}^j \frac{s!}{j!}  \| \v{\beta}^{(m,p_m+s)}\| \nonumber \\
&\le \frac{1}{j!} ( \|\v{\gamma}^{(m,0)}\| + \sum_{r=1}^{p_m} \|\v{\beta}^{(m,r)}\|) + \sum_{s=1}^j \frac{s!}{j!}  \| \v{\beta}^{(m,p_m+s)}\|.
\end{align}
Again using Cauchy-Schwarz we have
	\begin{align}\label{eqn:mp-norm}
		\| \v{\gamma}^{(m,p_m+j)}\|^2  &\leq \frac{(1+\epsilon_{\mbox{\tiny TD}}/ c_{+,\times})^2}{j!^2} \sum_{l=0}^m C(l h)^2 \left( 1 +  p_{l-1} (1+\epsilon_{\mbox{\tiny TD}}/c_{+,\times})^2 C(h)^2  + g(k) \right ) + \frac{p_m}{j!^2} + \sum_{s=1}^j \left ( \frac{s!}{j!}\right)^2.
	\end{align}

	We now have that
	\begin{align}
	   \|\v{\gamma}\|^2 = \sum_{m=0}^{M-1}\left (\sum_{r=0}^{p_m} \|\v{\gamma}^{(m,r)}\|^2 + \sum_{j=1}^{k} \|\v{\gamma}^{(m,p_m+j)}\|^2\right) + \sum_{r=0}^{p_M} \|\v{\gamma}^{(M,r)}\|^2.
	\end{align}
	Using the above bounds in Eq.~\eqref{eqn:mr-norm} and Eq.~\eqref{eqn:mp-norm} we have that
\small
\begin{align}
	   \|\v{\gamma}\|^2 &\le \sum_{m=0}^{M} \left( \left (1+ \frac{\epsilon_{\mbox{\tiny TD}}}{c_{+,\times}} \right)^2 \sum_{r=0}^{p_m} \sum_{l=0}^m C(l h)^2 \left( 1 +  p_{l-1} \left (1+ \frac{\epsilon_{\mbox{\tiny TD}}}{c_{+,\times}} \right)^2 C(h)^2  + g(k) \right ) + \sum_{r=0}^{p_m}  r \right)  \nonumber\\
	   & + \sum_{m=0}^{M}\sum_{j=1}^k \left(  \frac{\left (1+ \frac{\epsilon_{\mbox{\tiny TD}}}{c_{+,\times}} \right)^2}{j!^2} \sum_{l=0}^m C(l h)^2 \left( 1 +  p_{l-1} \left (1+ \frac{\epsilon_{\mbox{\tiny TD}}}{c_{+,\times}} \right)^2 C(h)^2  + g(k) \right ) + \frac{p_m}{j!^2} + \sum_{s=1}^j \left ( \frac{s!}{j!}\right)^2 \right) \nonumber \\
	   & \leq  \sum_{m=0}^{M} \left( \left (1+ \frac{\epsilon_{\mbox{\tiny TD}}}{c_{+,\times}} \right)^2 (p_m+1) \sum_{l=0}^m C(l h)^2 \left( 1 +  p_{l-1} \left (1+ \frac{\epsilon_{\mbox{\tiny TD}}}{c_{+,\times}} \right)^2 C(h)^2  + g(k) \right ) + \frac{p_m(p_m+1)}{2} \right) \nonumber \\
	   & + \sum_{m=0}^{M} \left(  \left (1+ \frac{\epsilon_{\mbox{\tiny TD}}}{c_{+,\times}} \right)^2 (I_0(2)-1) \sum_{l=0}^m C(l h)^2 \left( 1 +  p_{l-1} \left (1+ \frac{\epsilon_{\mbox{\tiny TD}}}{c_{+,\times}} \right)^2 C(h)^2  + g(k) \right ) \right)  \!+\!\left( \sum_{m=0}^{M}p_m + M k\right) (I_0(2)\!-\!1) \nonumber\\
	   & \le \sum_{m=0}^{M} \left( \left (1+ \frac{\epsilon_{\mbox{\tiny TD}}}{c_{+,\times}} \right)^2 (p_m+I_0(2)) \sum_{l=0}^m C(l h)^2 \left( 1 +  p_{l-1} \left (1+ \frac{\epsilon_{\mbox{\tiny TD}}}{c_{+,\times}} \right)^2 C(h)^2  + g(k) \right ) + \frac{p_m(p_m+1)}{2} \right) \nonumber \\ 
    & + \left( \sum_{m=0}^{M}p_m + M k\right) (I_0(2)-1),
	\end{align}
	\normalsize
	where we used that
	\begin{align}
	    \sum_{j=1}^k \frac{1}{(j!)^2} \le \sum_{j=1}^\infty \frac{1}{(j!)^2} =I_0(2)-1,
	\end{align}
	where $I_0(x)$ is the order-zero modified Bessel function of the first kind, and also that
\begin{equation}
    \sum_{j=1}^{k} \sum_{s=1}^j \left ( \frac{s!}{j!}\right)^2 = \sum_{j=1}^{k} \sum_{s=1}^j \left ( \frac{1}{(j-s)!}\right)^2 \le \sum_{j=1}^k \sum_{w=1}^j \frac{1}{(w!)^2} \le  k (I_0(2)-1).
\end{equation}
From numerical tests, we find the latter inequalities cannot be substantially improved. This completes the proof.
\end{proof}

     \begin{rmk}
     \label{rmk:solutionlinearsystem}
        Using the expressions found in the previous proof we can show that Eq.~\eqref{eq:vectorberry2} holds, i.e., that the ideal inversion of $L$ on the claimed form for $\v{c}$ gives the claimed output. Specifically, we choose $\v{\beta} = \v{c}$ via the conditions (refer to Eq.~\eqref{eq:cdefinition})
\begin{align}
  \v{\beta}^{(0,0)} &= \v{x}^0 \nonumber \\  
      \v{\beta}^{(m,p_m+1)}&=h\v{b} \mbox{ for all } m=0, \dots , M-1\nonumber \\
      \v{\beta}^{(m,r)} &=\v{0} \mbox{ otherwise}.
\end{align}
The idling terms will produce copies of data with this choice of input. More precisely, from Eq.~(\ref{eqn:idle-copying}) we have that
\begin{equation}
    \v{\gamma}^{(m,r)} = \v{\gamma}^{(m,0)},
\end{equation}
for $r =0, \dots p_m$ since $\v{\beta}^{(m,s)}=\v{0}$ for all $m$ and all $s=1, \dots ,p_m$.

The conditions $\v{\beta}^{(m,p_m+1)} =h\v{b}$ and $\v{\beta}^{(m,r)} = \v{0}$ on other indices, imply that we have (Eq.~\eqref{eq:Bm-1proof})
\begin{align}
    \v{B}^{(m-1)} &:=\sum_{r=1}^{p_{m-1}}T_k(Ah)  \v{\beta}^{(m-1,r)} +\sum_{j=1}^k\sum_{s=1}^j \frac{s!}{j!} (Ah)^{j-s} \v{\beta}^{(m-1,p_{m-1}+s)} \nonumber \\
    &= \sum_{j=1}^k \frac{1!}{j!} (Ah)^{j-1} \v{\beta}^{(m-1,p_{m-1}+1)}  = S_k(Ah) h\v{b}.
\end{align}
Therefore, Eq.~(\ref{eqn:m-iterate}) implies that
\begin{equation}
  \v{\gamma}^{(m,0)} =  T_k(Ah) \v{\gamma}^{(m-1,0)} +  S_k(Ah) h \v{b},
\end{equation}
for $m=1, \dots , M$ and the initial condition $\v{\gamma}^{(0,0)} = \v{x}^0$. This recursion relation coincides with the one for the discrete dynamics and therefore $\v{\gamma}^{(m,0)}= \v{x}^m$ for all $m=0, \dots , M$, as required. Moreover, idling implies $\v{\gamma}^{(m,r)}= \v{x}^m$ for all $m=0, \dots , M$ and all $r =0, \dots, p_m$. The remaining indices correspond to junk data. This shows that $L^{-1}\v{c}$ returns the claimed form of discrete solution.
     \end{rmk}

 We can now specialize the above theorem to idling only at the end (or not at all),  and consider both the stable $\alpha(A) <0$ case and the case $\alpha(A) \ge 0$ separately.
\begin{theorem}[Condition number bounds: stability cases]
\label{thm:conditionnumbernoidling}
Consider a complex matrix $A\in \mathbb{C}^{N \times N}$, a Taylor truncation order $k =k_{+,\times}$ associated to a discretization error $\epsilon_{\mbox{\tiny TD}}$, and the associated embedding of $A$ into a matrix $L$ defined by Eq.~(\ref{eq:berrymatrix}), for discretized dynamics from $t=0$ up to $t=T=Mh$, via $M+1$ uniform steps of size $h$. Consider the case of idling parameters given by $p_m =0$ for all $0 \le m<M$ and $p_M = p$. Let $I_0(x)$ be the modified Bessel function of order zero, and $
	    g(k) :=  \sum_{s=1}^k \left(s! \sum_{j=s}^k 1/j!\right)^2 \le ek$. Also, $c_+ = x_{\mathrm{max}}$ and $c_\times = 1$. Then the following upper bounds on the matrix $L$ hold:

\begin{itemize}
    \item \emph{($A$ is a stable matrix).} For the case where $A$ is a stable matrix, let $P$ and $Q$ be any two positive Hermitian matrices that satisfy the Lyapunov equation $PA + A^\dagger P=-Q$. Denote by $\mu_P(A)$ the log-norm of $A$ with respect to the elliptical norm induced by $P$ and denote by $\kappa_P$ the condition number of $P$. Then we have that
 \footnotesize
 \begin{align}
\! \! \! \! \!  \kappa_L \leq
	   & \left[ \left (1+ \frac{\epsilon_{\mbox{\tiny TD}}}{c_{+,\times}} \right)^2 \left( 1  + g(k) \right )  \kappa_P \left(p \frac{1- e^{2 M\mu_P(A) h+2 \mu_P(A) h}}{ 1- e^{2 \mu_P(A) h}}  +I_0(2) \frac{e^{2\mu_P(A)(M+2)h} + M +1 - e^{2\mu_P(A) h} (2+M)}{(1-e^{2 h \mu_P(A)})^2} \right) \right. \nonumber \\ & \left. + \frac{p(p+1)}{2}  + (p + M k) (I_0(2)-1) \right]^{1/2}  (\sqrt{k+1}+2),
	     \label{eq:conditionnumbernoidling}
	   \end{align}
	   \normalsize 
and therefore we obtain the asymptotic complexity $\kappa_L = O(k\sqrt{\kappa_P T +p^2})$.
 
     \item \emph{($A$ is not a stable matrix).} If $A$ is not stable, then consider any uniform bound $\|e^{At}\|\leq C_{\mathrm{max}}$ for all $t \in [0,T]$. Then we have that 
     \scriptsize
	\begin{equation}
	  \kappa_L \le  (\sqrt{k+1}+2)\sqrt{ C_{\mathrm{max}}^2\left (1+ \frac{\epsilon_{\mbox{\tiny TD}}}{c_{+,\times}} \right)^2 \left( 1  + g(k) \right ) (M+1)  [p +I_0(2)(M/2+1) ]   + \frac{p(p+1)}{2}  + (p + M k) (I_0(2)-1)},
	\end{equation}
	   \normalsize 
and therefore we obtain the asymptotic complexity $\kappa_L = O(k \sqrt{C^2_{\mathrm{max}}T^2+p^2})$.
\end{itemize}

\end{theorem}
Note that for all stable matrices we have that the condition number is $O(\sqrt{T} k)$, and so unless $k$ scales polynomially with $T$ (which happens only if $\lambda_{+,\times}(T)$ grows exponentially, i.e. the ODE solution norm decays or grows exponentially),  we have a quadratic improvement in the condition number dependence with $T$ compared to previously known results, which scale as $O(T \mathrm{poly}(k))$~\cite{berry2017quantum,krovi2022improved,berry2022quantum}. A related quadratic speedup~\cite{an2022theory} was obtained for the case of $A=-H^2$ for Hermitian $H$.

\begin{proof}[Proof of Theorem~\ref{thm:conditionnumbernoidling}]

From the proof of Theorem~\ref{thm:conditionnumbergeneral}, setting $p_m =0$ for $m=0,\dots, M-1$ and $p_M= p$, we have that
\small 
\begin{align}
	   \|\v{\gamma}\|^2
	   & \le \sum_{m=0}^{M} \left( \left (1+ \frac{\epsilon_{\mbox{\tiny TD}}}{c_{+,\times}} \right)^2 (p \delta_{m,M}+I_0(2)) \sum_{l=0}^m C(l h)^2 \left( 1  + g(k) \right ) \right ) + \frac{p(p+1)}{2}  + (p + M k) (I_0(2)-1) \nonumber \\
	   & = \left (1+ \frac{\epsilon_{\mbox{\tiny TD}}}{c_{+,\times}} \right)^2 \left( 1  + g(k) \right )\left(  p \sum_{l=0}^M C(l h)^2 + I_0(2)  \sum_{m=0}^M \sum_{l=0}^m C(l h)^2  \right ) + \frac{p(p+1)}{2}  + (p + M k) (I_0(2)-1).
	\end{align}
	\normalsize
	Now consider the different subcases for $A$.
	
	\bigskip
	
  \begin{itemize}
      \item \emph{($A$ is a stable matrix)} 
 We use $\|e^{At}\| \le \sqrt{\kappa_P} e^{\mu_P(A) t}$ for all $t \ge 0$ to obtain geometric sums in the previous inequality.
 This implies that
\begin{align}
  \sum_{l=0}^M C(l h)^2 &= \kappa_P \frac{1- e^{2 M\mu_P(A) h+2 \mu_P(A) h}}{ 1- e^{2 \mu_P(A) h}} \le M+1 =O(M)\nonumber \\
    \sum_{m=0}^{M} \sum_{l=0}^m C(l h)^2 &= \kappa_P \frac{e^{2\mu_P(A)(M+2)h} + M +1 - e^{2\mu_P(A) h} (2+M)}{(1-e^{2 h \mu_P(A)})^2 } \le (M+1)(M/2+1) = O(M^2),
  \end{align}
where in the limit $\mu_P(A) \rightarrow 0^-$, the sums saturate the upper bounds $M+1$ and $(M+1)(1+M/2)$, respectively.  This limit corresponds to the $\alpha(A) \ge 0$ bound upon replacing $\kappa_P$ with $C_{\mathrm{max}}^2$. For $\mu_P(A) < c<0$ for some constant $c$ independent of $N$ we instead get that the first expression scales in time as $O(1)$, whereas the second scales as $O(M)$.
\scriptsize 
\begin{align}
	\! \! \! \! \!    \|\v{\gamma}\|^2
	   & \le   \left (1+ \frac{\epsilon_{\mbox{\tiny TD}}}{c_{+,\times}} \right)^2 \left( 1  + g(k) \right )  \left(p \sum_{l=0}^M C(l h)^2 + I_0(2) \sum_{m=0}^{M} \sum_{l=0}^m C(l h)^2  \right ) + \frac{p(p+1)}{2}  + (p + M k) (I_0(2)-1) \nonumber \\
	   & =  \left (1+ \frac{\epsilon_{\mbox{\tiny TD}}}{c_{+,\times}} \right)^2 \left( 1  + g(k) \right ) \kappa_P  \left(p \frac{1- e^{2 M\mu_P(A) h+2 \mu_P(A) h}}{ 1- e^{2 \mu_P(A) h}} + I_0(2) \frac{e^{2\mu_P(A)(M+2)h} + M +1 - e^{2\mu_P(A) h} (2+M)}{(1-e^{2 h \mu_P(A)})^2 } \right ) \nonumber \\ & + \frac{p(p+1)}{2}  + (p + M k) (I_0(2)-1)
	\end{align}
	\normalsize
	and so
	\scriptsize 
\begin{align}
\kappa_L \leq & \left[   \left (1+ \frac{\epsilon_{\mbox{\tiny TD}}}{c_{+,\times}} \right)^2 \left( 1  + g(k) \right ) \kappa_P  \left(p \frac{1- e^{2 M\mu_P(A) h+2 \mu_P(A) h}}{ 1- e^{2 \mu_P(A) h}} + I_0(2) \frac{e^{2\mu_P(A)(M+2)h} + M +1 - e^{2\mu_P(A) h} (2+M)}{(1-e^{2 h \mu_P(A)})^2 } \right ) \right. \nonumber \\ & \left. + \frac{p(p+1)}{2}  + (p + M k) (I_0(2)-1) \right]^{1/2} (\sqrt{k+1}+2)
	\end{align}
	\normalsize
By inspection, we see that $\kappa_L = O(k\sqrt{\kappa_P T +p^2})$.

\item \emph{($A$ is not a stable matrix)} For this case we use that $C(t) \leq C_{\mathrm{max}}$ for all $t\in [0,T]$.  Then
	\small 
	\begin{align}
	   \|\v{\gamma}\|^2
	   & \le  C_{\mathrm{max}}^2\left (1+ \frac{\epsilon_{\mbox{\tiny TD}}}{c_{+,\times}} \right)^2 \left( 1  + g(k) \right ) (M+1)  [p +I_0(2)(M/2+1) ]   + \frac{p(p+1)}{2}  + (p + M k) (I_0(2)-1).
	\end{align}
	\normalsize
	This implies the following condition number upper bound,
	\footnotesize
	\begin{equation}
	  \kappa_L \le  (\sqrt{k+1}+2)\sqrt{ C_{\mathrm{max}}^2\left (1+ \frac{\epsilon_{\mbox{\tiny TD}}}{c_{+,\times}} \right)^2 \left( 1  + g(k) \right ) (M+1)  [p +I_0(2)(M/2+1) ]   + \frac{p(p+1)}{2}  + (p + M k) (I_0(2)-1)}.
	\end{equation}
	\normalsize
	For this we have $\kappa_L = O(k \sqrt{C^2_{\mathrm{max}}T^2+p^2})$. We note that the previous bound can also be obtained by taking $\kappa_P = C^2_{\mathrm{max}}$ and $\mu_P \rightarrow 0^-$ in the expression for stable $A$. 
 \end{itemize}
	   \end{proof}
	   Note that when outputting the history state the above result will be used setting $p=0$, whereas when outputting the solution state we will set $p=\Theta(\sqrt{T})$, as we shall see. In both cases, $\kappa_L = O(k \sqrt{\kappa_P T})$ for the case of $A$ being stable.
	   
To relate this result to prior works we make use of the following lemma.
	   \begin{lem}\label{lem:kappa-cases}
    Let $A \in \mathbb{C}^{N\times N}$ be a stable, complex matrix, and let $P>0$ be a Hermitian matrix that obeys the Lyapunov condition $PA +A^\dagger P <0 $. Let $\kappa_P$ denote the condition number of $P$ and $\mu_P(A)$ denote the log-norm of $A$ with respect to the ellipsoidal norm defined by $P$. Then the following special cases hold:
\begin{itemize}
    \item If $A$ is diagonalizable via $V^{-1}AV = D$, a diagonal matrix, then we may take $P=V^\dagger V$ and so $\kappa_P= \kappa_V^2$ (with $\kappa_V$ the condition number of $V$) and $\mu_P(A) = \alpha(A)$. In particular, if $A$ is a normal matrix then $\kappa_P = 1$ and $\mu_P(A) = \alpha(A)$.
        \item If $A$ has negative Euclidean log-norm $\mu(A)<0$, then we may take $P=I$ and so $\kappa_P = 1$ and $\mu_P(A) = \mu(A)$.
\end{itemize}
\end{lem}
 The proof of this lemma is as follows.
\begin{proof} 
If $A$ is diagonalizable and $A = V^{-1} D V $ for some diagonal $D$, then we may take $P = V^\dagger V >0$. In fact,
\begin{align}
    PA +A^\dagger P  &= V^\dagger V A + A^\dagger V^\dagger V \nonumber
    = V^\dagger ( D + D^\dagger ) V.
\end{align}
However, since $A$ cannot have any eigenvalues with nonnegative real parts, it follows that $V^\dagger (D + D^\dagger) V = -Q$ for some Hermitan operator $Q > 0$. We then have that $\kappa_P = \|V^\dagger V\| \|V^{-1} (V^{-1})^{\dagger}\| = \kappa_V^2$, since $\|X^\dagger X \| = \|X\|^2$ for any $X$ (in the particular case of $A$ normal, the operator is diagonalizable via a unitary $V$ for which $\kappa_V =1$, so $\kappa_P=1$). To show that $\mu_P(A) = \alpha(A)$ note that
\begin{equation}
    \max_{\|\v{x}\| \ne 0} \mathrm{Re} \,\frac{ \langle PA \v{x} ,\v{x}\rangle}{\langle P \v{x}, \v{x} \rangle} =    \max_{\|\v{x}\| \ne 0} \mathrm{Re} \,\frac{ \langle VA \v{x} ,V\v{x}\rangle}{\langle V \v{x}, V \v{x} \rangle} = \max_{\|\v{x}\| \ne 0} \mathrm{Re} \,\frac{ \langle VAV^{-1}  V\v{x} ,V\v{x}\rangle}{\langle V \v{x}, V \v{x} \rangle}=\max_{\|\v{x}'\| \ne 0} \mathrm{Re} \,\frac{ \langle D\v{x}', \v{x}'\rangle}{\langle \v{x}',  \v{x}' \rangle} = \alpha(A),
\end{equation}
where we replace $\v{x}' = V\v{x}$.

Suppose now that we choose $P=I$ and have that $PA + A^\dagger P= A+A^{\dagger}$ has strictly negative eigenvalues. This occurs if and only if $A$ has negative log-norm with respect to operator norm. In this case we have $\kappa_P = 1$ and $\mu_P(A) = \mu(A)$.
	\end{proof}

 The previous Lemma, together with Theorem~\ref{thm:conditionnumbernoidling} immediately imply the following corollary. 
 \begin{cor}\label{cor:specialcases}
Suppose $A\in \mathbb{C}^{N \times N}$ is a stable, complex matrix. Let $L$ be the associated embedding of the discretized dynamics induced by $A$ with Taylor truncation order $k$ for a time interval $[0,T]$ and idling parameter $p$, as given in Eq.~\eqref{eq:berrymatrix}. Then the following special cases hold:
\begin{itemize}
    \item If $A$ is diagonalizable via $V^{-1}AV = D$, a diagonal matrix, then $\kappa_L = O\left ( \kappa_V k\sqrt{T + p^2/\kappa_V^2} \right )$, where $\kappa_V$ is the condition number of $V$.
    \item If $A$ has a negative log-norm $\mu(A)<0$, then $\kappa_L = O(k \sqrt{T+p^2})$.
\end{itemize}
 \end{cor}
	 The first result in Corollary~\ref{cor:specialcases} gives a quadratic improvement in $T$ over~\cite{berry2017quantum}, which had $\kappa_L =O(\kappa_V k T )$, in the regime $\alpha(A)<0$.  The second result in Corollary~\ref{cor:specialcases} gives a similar improvement over~\cite{krovi2022improved}, which had a complexity $O(k^{3/2} T)$. In fact, for $A$ having a negative log-norm we improve both the dependence on $k$ and $T$ for the condition number\footnote{Note that an extra factor of $k$ must be included to account for the implicit matrix inversion required in Ref.~\cite{krovi2022improved} so the effective scaling is $O(T k^{3/2})$.} over Ref.~\cite{krovi2022improved}. We provide more detailed analysis for the case of negative log-norm in the next section.
  
	\subsection*{Special case: \texorpdfstring{$A$}{A} with negative log-norm}

We can analyse the behaviour of these bounds in some more detail for the case where $P=I$, which corresponds to the case of $A$ having negative log-norm. Firstly, the transition from the $O(\sqrt{M})$ scaling to the $O(M)$ scaling is determined by the square root of
\begin{equation}\label{eqn: double-sum}
    \frac{e^{2\mu(A)(M+2)h} + M +1 - e^{2\mu(A) h} (2+M)}{(1-e^{2 h \mu(A)})^2 }.
\end{equation}
In the regime $\mu<0$ the scaling is $O(M)$, but one can wonder how small the constant prefactor can be. We can in fact derive a lower bound for it. Note that for any matrix $A$ we have $-\|A\| \le \mu(A) \le \|A\|$, and since stability requires $\|A\|h \le 1$ this implies that $e^{-1} \le e^{\mu h} \le 1$. The expression in Eq.~\eqref{eqn: double-sum} is equal to the double sum $\sum_{m=0}^M \sum_{r=0}^m (e^{\mu h})^r$. Letting $x=e^{\mu h}$ this sum can be written as
\begin{equation}
    \sum_{m=0}^M \sum_{r=0}^m (e^{\mu h})^r = \sum_{s=0}^M c_s x^s,
\end{equation}
where $c_s >0$. A sum of strictly positive monomials $c_s x^s$ with $c_s>0$, is monotone increasing for $x >0$, and therefore on the interval $[e^{-1},1]$ the sum attains its minimum when $x=e^{\mu h} = e^{-1}$. Replacing in Eq.~\eqref{eqn: double-sum}, we get
\small
\begin{align}
      \frac{ \left(e^{4} (M+1)-e^{2} (M+2)+e^{-2M}\right)}{\left(e^2-1\right)^2}  = \frac{e^4-2e^2 + e^{-2M}}{(e^2-1)^2}+ \frac{e^2}{e^2-1} M.
\end{align}
\normalsize
The smallest value the double sum can take among all systems with negative log-norm is hence smaller than
\begin{equation}
    1 + \frac{e^2}{e^2-1}M.
\end{equation}
Overall we have that in this  best-case scenario
\begin{align}
   \kappa_L &  \leq  \sqrt{\left (1+ \frac{\epsilon_{\mbox{\tiny TD}}}{c_{+,\times}} \right)^2 \left( 1  + g(k) \right ) I_0(2) \left(  \frac{e^4-2e^2 + e^{-2M}}{(e^2-1)^2}+ \frac{e^2}{e^2-1} M \right) + k (I_0(2) -1) M}  \\ 
   & \approx \sqrt{\left (1+ \frac{\epsilon_{\mbox{\tiny TD}}}{c_{+,\times}} \right)^2 \left( 1  + g(k) \right ) I_0(2) \frac{e^2}{e^2-1} + k (I_0(2) -1) } \times \sqrt{M}.
\end{align}
For example, for small $\epsilon_{\mbox{\tiny TD}}$, $k=19$ and $c_{+,\times} = 1$, the pre-factor constant before $\sqrt{M}$ is approximately $9.8$.
This optimal case occurs when $\mu(A) = -\|A\|$, which in turn occurs if 
\begin{equation}
    \max_{\|\v{x}\| =1} \mathrm{Re}(\v{x}^\dagger A \v{x}) = -\max_{\|\v{x} \| =1} \| A \v{x} \|.
\end{equation}
We can show this optimal bound occurs if and only if $A = -\|A\| I$ in the case of negative log-norm. 
\begin{lem}
Under the conditions of Corollary~\ref{cor:specialcases}, assume that $A$ has a negative log-norm. The optimal linear scaling $ Me^2/(e^2-1)$ in $\sum_{m=0}^{M} \sum_{l=0}^m C(l h)^2$ occurs if and only if $A = -\|A\|I$.
\end{lem}
\begin{proof}The `if' direction is clear. Suppose conversely that $\mu(A) = -\|A\|$. Let $\lambda_i(X)$ denote the $i$'th largest eigenvalue of the Hermitian matrix $X$. Consider $\mu_i:=\lambda_{i} (Re(A)) := \lambda_{i}((A+A^\dagger)/2)$, and let $\v{v}_i$ be a corresponding unit eigenvector for this eigenvalue, with the eigenvalues sorted in decreasing order. By definition of log-norm, $\mu(A) = \mu_1$. Furthermore, $\mu_N \leq \dots \leq \mu_2 \leq \mu_1 <0$ and so 
\begin{equation}
    | \mu_N | \geq \dots \geq | \mu_1|. 
\end{equation}
Also,
\begin{equation}
    |\mu_i| = \left| \lambda_i\left( \frac{A+A^\dag}{2} \right)\right| \leq \frac{\| A\| + \|A^\dag\|}{2} = \|A\|
\end{equation}
Hence, we have the following chain of inequalities for all $i =1, \dots, N$
\begin{align}
\label{eq:bestconstantproof}
  \|A\| \geq | \mu_i| =  |\lambda_{i} (Re(A))| &= |\v{v}_i^\dagger Re(A) \v{v}_i| = |Re(\v{v}_i^\dagger A \v{v}_i)| \nonumber \\
    &\leq |\v{v}_i^\dagger A \v{v}_i| = \sqrt{\v{v}_i^\dagger (A^\dagger A) \v{v}_i} \le \|A\|.
\end{align}
These inequalities must be saturated. The equality $|Re(\v{v}_i^\dagger A \v{v}_i)| = |\v{v}_i^\dagger A \v{v}_i|$ implies $\v{v}_i^\dagger Im(A)\v{v}_i=0$ for all $i$. Since $(A+A^\dagger)/2$ is a Hermitian matrix, it follows that the vectors $\v{v}_i$ form a complete basis, and so we find that $(A - A^\dag)/2 =0$. The equality $|\mu_i| = \|A\|$ implies $Re(A)\v{v}_i = - \|A\| \v{v}_i$ for all $i$, since we are given that $\mu_i <0$. Since $\v{v}_i$ are a basis, this implies $(A+A^\dag)/2 = - \|A\| I$. Hence
\begin{equation}
    A = \frac{A+A^\dag}{2} + \frac{A-A^\dag}{2} = - \|A\| I,
\end{equation}
which completes the proof.
\end{proof}
This lemma implies that the bounds are near-optimal when $A$ has a relatively uniform spectrum with small imaginary components.

\section{Block-encoding of the linear system matrix}\label{sec:block-encoding}

For quantum linear solver algorithms it is typically assumed that the defining matrix has operator norm bounded above by $1$. As shown in the previous section, for a Taylor truncation order $k=k_{+,\times}$ we have that \mbox{$\| L\| \leq \sqrt{k+1} +2$}, and therefore we rescale the matrix to put it into a canonical form. We therefore define 
\begin{align}
    \tilde{L} &:= \frac{1}{\sqrt{k+1} +2}L\\
    \tilde{\v{c}} &:= \frac{1}{\sqrt{k+1} +2}\v{c}, 
\end{align}
and solve the rescaled problem $\tilde{L} \v{y} = \tilde{\v{c}}$, which has $\| \tilde{L} \| \leq 1$, as normally required. Note that the coherent solution $\ket{y}$ to this system coincides with the coherent solution of the original system. Moreover, the condition number $\kappa_{\tilde{L}}$ of the matrix $\tilde{L}$ equals $\kappa_L$ since the expression $\|L\|\|L^{-1}\|$ is invariant under rescaling.

For the linear-solver algorithm we require both a unitary $U_{\tilde{c}}$ that prepares the normalized state $\ket{\tilde{c}}$ as well as $U_{\tilde{L}}$, a block-encoding of $\tilde{L}$. For the former we note that $\ket{\tilde{c}}= \ket{c}$, since the rescaling is removed under the normalization of the state. It is readily seen that $U_{\tilde{c}}$ can be constructed via a single call to $U_0$ and a single call to $U_b$.

For $U_{\tilde{L}}$, we recall that an $(\omega_X, a_X, \epsilon_X)$ block-encoding of a matrix $X$ is defined as a unitary $U_X$ such that
\begin{equation}
   \|  (\bra{0^{a_X}} \otimes I) U_X (\ket{0^{a_X}} \otimes I) - X/\omega_X \| \leq \epsilon_X.
\end{equation}
Here we denote by $\ket{0^a}$ the state $\ket{0}^{\otimes a}$, for any integer $a$.
We assume oracle access to $A$ via an $(\omega, a, 0)$ block-encoding of $A$. From this we need to construct a block-encoding $U_{\tilde{L}}$ of the matrix $\tilde{L}$. The following result shows this can be readily done.
\begin{theorem} \label{thm:omegavalue} Let $\tilde{L} := \frac{1}{\sqrt{k+1} +2}L$, where $L$ is defined in Eq.~\eqref{eq:L}. Take $\omega h \geq 1$. Given a Taylor truncation scale $k=k_{+,\times}$
 we can construct an $(\tilde{\omega}, a+6,0)$-block-encoding $U_{\tilde{L}}$ of $\tilde{L}$ via a single call to an $(\omega, a, 0)$ block-encoding $U_{A}$ of the matrix $A$, with
 \begin{equation}
     \tilde{\omega} = \frac{ \omega h +1 + \sqrt{k+1}}{2 + \sqrt{k+1}} = O\left (1+\frac{ \omega h }{\sqrt{k}} \right),
 \end{equation}
 being the rescaling factor of the block-encoding for $\tilde{L}$.
\end{theorem}
\begin{proof}
We construct a block-encoding for $L = L_1 + L_2 + L_3$, which in turn provides a block-encoding for the rescaled matrix $\tilde{L}$. We consider the problem separately for $L_1$, $L_2$ and $L_3$ and then combine them with linear combination of unitaries.

First, $L_1$: this is just the identity matrix, hence the block-encoding is trivial -- we have a $(1,0,0)$-block-encoding of $L_1$. 

Second, $L_2$:
\begin{align}
L_{2} := & - \sum_{m=0}^{M-1} \sum_{j=1}^{k}  \ketbra{m,j}{m,j-1} \otimes \frac{Ah}{j} - \theta(p) \sum_{m=M}^{ M+ \frac{p}{k+1}-1} \sum_{j=1}^{k} \ketbra{m,j}{m,j-1} \otimes I,
\end{align}
which we rewrite as
\begin{align}
 L_2 = \omega h L'_2 L''_2 L'''_2,
\end{align}
where
\begin{align}
L'_{2} := - \left[\sum_{m=0}^{M-1} \ketbra{m}{m} \otimes I \otimes \frac{A}{\omega} + \sum_{m=M}^{ M+ \frac{p}{k+1}} \ketbra{m}{m} \otimes I \otimes I\right]
\end{align}
\begin{align}
L_2'' = \left[- \sum_{m=0}^{M-1} \sum_{j=1}^{k}  \ketbra{m,j}{m,j} \otimes \frac{I}{j}  - \theta(p) \sum_{m=M}^{ M+ \frac{p}{k+1}-1} \sum_{j=1}^{k} \ketbra{m,j}{m,j} \otimes \frac{I}{ \omega h} \right] \\
    L'''_2 = I \otimes \Delta^\dag \otimes I, \quad \quad \Delta^\dag = \sum_{j=0}^{k-1} \ketbra{j+1}{j}.
\end{align}
$L_2'$ can be $(1,a,0)$-block-encoded by $U_{L'_2}$,  where
\begin{equation}
     U_{L'_2} = \sum_{m=0}^{M-1} \ketbra{m}{m} \otimes I \otimes U_{A} + \sum_{m=M}^{M+p/(k+1)} \ketbra{m}{m} \otimes I \otimes I
\end{equation}
   $L''_2$ can be $(1,1,0)$-block-encoded as
    
\begin{equation}
    U_{L''_2} \ket{0} \ket{mj} = \begin{cases}
      \left(-\frac{1}{j} \ket{0} + \sqrt{1-\frac{1}{j^2}} \ket{1}\right) \ket{m j} \quad \mbox{ if }  (m < M \mbox{ and } j> 0) , \\ 
     \left(-\frac{1}{\omega h} \ket{0} + \sqrt{1-\frac{1}{\omega^2 h^2}} \ket{1}\right) \ket{m j} \quad \mbox{ if }  (m \geq M \mbox{ and } j> 0), \\
      \ket{1} \ket{mj} \quad \mbox{ if }  j=0,
    \end{cases}
\end{equation}
 where we assumed that $\omega h \ge 1$.
This gives a $(1,1,0)$-block-encoding of $L_2'$.
Finally, $L'''_2$ can be block-encoded as
$U_{L_2'''} = I \otimes U_{\Delta^\dag} \otimes I$, where $U_{\Delta^\dag}$ acts on the space where $\Delta^\dag$ acts augmented by an extra qubit ancilla:
\begin{equation}
    U_{\Delta^\dag} = \Delta^\dag \otimes \ketbra{0}{0} + \Delta \otimes \ketbra{1}{1} + \ketbra{0}{0} \otimes \ketbra{0}{1}.   
\end{equation}
Using $\Delta^\dag \Delta = I - \ketbra{0}{0}$ and $\Delta \Delta^\dag  = I$, one can check that $U_{\Delta^\dag}$ is unitary. Hence, we have a $(1,1,0)$-block-encoding of $I \otimes \Delta^\dag \otimes I$. 

Overall, we set $U_{L_2}:=U_{L'_2}U_{L_2''} U_{L_2'''}$, where each one of the three block-encoding unitaries has implicit identity operators on the auxiliary qubits of the other block-encodings. This gives an $(\omega h, a+ 2,0)$-block-encoding of $L_2$ involving a single query to~$U_{A}$.

Third, $L_3$: we rewrite
\begin{equation}
    L_3 = - \sqrt{k+1} (\tilde{\Delta}^\dag \otimes I \otimes I) (I \otimes \ketbra{0}{0} \otimes I) (I \otimes U^\dag_{USP} \otimes I),
\end{equation}
where $\tilde{\Delta}^\dag = \sum_{m=0}^{M-1} \ketbra{m+1}{m}$,
\begin{equation}
    U_{USP} \ket{0} = \ket{+_k}
\end{equation}
and $\ket{+_k} = \frac{1}{\sqrt{k+1}} \sum_{j=0}^k \ket{j}$.
Here $U_{USP}$ is a uniform state preparation circuit~\cite{lee2021even} (if $k$ is a power of $2$, this is just the tensor product of $\log_2 k$ Hadamards). Furthermore, the projector onto $\ketbra{0}{0}$ can be decomposed as
\begin{equation}
    \ketbra{0}{0} = \frac{1}{2} I + \frac{i}{2} e^{i \pi Z/2},
\end{equation}
and so one can obtain a $(1,1,0)$-block-encoding of $I \otimes \ketbra{0}{0} \otimes I$ via Linear Combination of Unitaries (LCU)~\cite{lin2022lecture}. Finally, we can construct a $(1,1,0)$-block-encoding of $\tilde{\Delta}^\dag$ exactly in the same way as we constructed one for $\Delta^\dag$ above. Overall, we obtain a $(\sqrt{k+1},2,0)$-block-encoding of $L_3$. Finally, consider
\begin{equation}
\frac{1}{1+\omega h + \sqrt{1+k}}L =  \frac{1}{1+\omega h + \sqrt{1+k}} \left[ L_1 + \omega h \left( \frac{L_2}{ \omega h}\right) + \sqrt{1+k} \left( \frac{L_3}{\sqrt{1+k}} \right) \right].  
\end{equation}

This can be realized by the LCU 
\begin{equation}
 \frac{1}{1+\omega h + \sqrt{1+k}} (U_{L_1} + \omega h U_{L_2} + \sqrt{1+k} U_{L_3})    
\end{equation}
using $2$ extra qubits. We hence constructed a $(1+\omega h + \sqrt{1+k}, a+6,0)$-block-encoding of $U_L$ using a single call to $U_{A}$. This implies that
\begin{equation}
  ( \bra{0^{a+6}} \otimes I ) U_L (\ket{0^{a+6}}\otimes I) = \frac{1}{\omega h +1 + \sqrt{1+k}} L =\frac{2+\sqrt{k+1}}{\omega h+1 + \sqrt{1+k}}\tilde{L} =: \frac{1}{\tilde{\omega}} \tilde{L},
\end{equation}
 as claimed.
\end{proof} Therefore, we have all the ingredients needed for the linear-solver component. The optimal asymptotic query complexity of the quantum linear solver algorithm~\cite{costa2022optimal} invokes the oracles $U_{\tilde{L}}$ and $U_{\tilde{c}}$ for a target error $\epsilon_L$ a number of times
\begin{equation}
O\left (\tilde{\omega} \kappa_{\tilde{L}} \log (1/\epsilon_L) \right ) = O( (\kappa_L + \omega h\kappa_{L}/\sqrt{k} ) \log (1/\epsilon_L)).
\end{equation}
Since these oracles are in turn constructed from single calls to $U_A$, $U_0$ and $U_b$ the above expression also gives the query complexity to the defining oracles of the linear ODE problem.

We next turn to post-processing of the state that is output from the linear solver subroutine.

\section{Success probabilities for post-selection} 
\label{sec:successprobabilityanalysis}

Recall that we denote by $\ket{y}$ a quantum state proportional to the solution to the linear system problem, \mbox{$\v{y} = L^{-1} \v{c}$}. Furthermore, $\ket{x_H^{\epsilon_{\mbox{\tiny TD}}}}$ and $\ket{x_T}$ are proportional to
\begin{equation}
\v{x}_H^{\epsilon_{\mbox{\tiny TD}}} =  [\v{x}^0, 0, \dots, 0, \v{x}^1, 0, \dots, 0, \v{x}^{M}]
\end{equation}
and
\begin{equation}
\v{x}_T = [0, \dots, 0, \v{x}^{M}],
\end{equation}
respectively. We then have,

\begin{theorem}[Success probability of ODE-solver -- history state]\label{thm: successprobhistorystate}
	Let $h>0$ be such that  $\|A\| h \leq 1$. Let $\ket{y}$ be the normalized solution to the linear system $L \v{y} = \v{c}$ with $p_m =0$ for all $m$ and
  \begin{equation}
    \v{c} = [ \overbrace{\v{x}^0, \underbrace{\v{b}, 0, \dots, 0}_{k}}^{m=0}, \dots, \overbrace{0,\underbrace{\v{b}, 0, \dots, 0}_{k}}^{m=M-1},0].
\end{equation}

We define $\ket{x_H^{\epsilon_{\mbox{\tiny TD}}}}$ as the state obtained by successful post-selection on the first outcome of the projective measurement $\{\Pi_{H}, I-\Pi_{H} \}$ on the auxiliary register:
\begin{equation}
    \Pi_{H} = I \otimes \ketbra{0}{0} \otimes I, \quad  \ket{x_H^{\epsilon_{\mbox{\tiny TD}}}} := \frac{\Pi_H \ket{y}}{\| \Pi_H \ket{y}\|}.
\end{equation}
For the case of multiplicative errors $\|\v{x}(mh) - \v{x}^m\| \le \epsilon_{\mbox{\tiny TD}} \|\v{x}(mh)\|$, we have that the success probability is lower bounded as
		\begin{equation}
			\label{eq:successprobabilityODE}
			\mathrm{Pr}_H  \geq \max \left\{ \left( 1 + \left(1+\frac{ \lambda_\times(T)}{(1-\epsilon_{\mbox{\tiny TD}})\|A\|}\right)^2 (I_0(2)-1)  \right)^{-1} , \left(1+ \frac{ (I_0(2)-1)}{K}\right)^{-1} \right\} 
		\end{equation}
		where $K = (3-e)^2$ if $\v{b} \neq 0$ and $K=1$ for $\v{b}=0$ and $\lambda_\times(T) = \max\{ \max_{t\in [0,T]} \|\v{b}\|/\|\v{x}(t)\|, T\}$. For the case of additive errors $\|\v{x}(mh) - \v{x}^m\| \le \epsilon_{\mbox{\tiny TD}}$, we have $\mathrm{Pr}_H \ge \left(1+ \frac{ (I_0(2)-1)}{K}\right)^{-1} $.
  Therefore, $Pr_H =\Omega(1)$ in all cases.
\end{theorem}

\begin{proof}

The desired data is given by
\begin{equation}
    \v{x}_H^{\epsilon_{\mbox{\tiny TD}}} = [\v{y}^{(0,0)},0, \dots, 0,\v{y}^{(1,0)}, 0, \cdots, 0, \v{y}^{(M,0)}  ]
\end{equation}
which is obtained by post-selecting success (first outcome) in the projective measurement $\{\Pi_{H}, I-\Pi_{H} \}$ (as it follows from the expression for $\v{y}$, Eq.~\eqref{eq:vectorberry2}).
Moreover the success probability is given by $\mathrm{Pr}_H = \|\v{x}_H^{\epsilon_{\mbox{\tiny TD}}}\|^2/\|\v{y}\|^2$.

From Eq.~\eqref{eq:berrymatrix} with $p_m \equiv 0$,
\begin{align}
    \v{y}^{(m,j)} &= \frac{(Ah)^{j-1}}{j!} \v{y}^{(m,1)}, \nonumber \\
     \v{y}^{(m,1)} &= Ah \v{y}^{(m,0)}+ h \v{b}, \nonumber \\
	    \v{y}^{(m+1,0)} &= \sum_{j=1}^k \frac{(A h)^{j-1}}{j!} \v{y}^{(m,1)}.
\end{align}
	
 We first show that	$\mathrm{Pr}_H  \geq \left(1+ \frac{ (I_0(2)-1)}{K}\right)^{-1}$.

We compute
	\begin{equation}
		\| \v{y} \|^2 = \sum_{m=0}^M \| \v{y}^{(m,0)}\|^2 + \sum_{m=0}^{M-1} \sum_{j=1}^{k} \| \v{y}^{(m,j)}\|^2 := \| \v{x}_H^{\epsilon_{\mbox{\tiny TD}}}\|^2 + \sum_{m=0}^{M-1} \sum_{j=1}^{k} \| \v{y}^{(m,j)}\|^2.
	\end{equation}
We start with the simple case where $\v{b} = 0$. Then
\begin{align*}
   \sum_{m=0}^{M-1}  \sum_{j=1}^k \| \v{y}^{(m,j)}\|^2 &=  \sum_{m=0}^{M-1}  \sum_{j=1}^k \left\| \frac{(Ah)^{j-1}}{j!} \v{y}^{(m,0)}\right\|^2  \leq  \sum_{m=0}^{M-1}  \sum_{j=1}^k \frac{1}{j!^2} \| \v{y}^{(m,0)}\|^2 \\ & \leq (I_0(2) -1)  \sum_{m=0}^M   \| \v{y}^{(m,0)}\|^2 = (I_0(2) -1) \|\v{x}_H^{\epsilon_{\mbox{\tiny TD}}}\|^2.
\end{align*}
Hence, $\| \v{y} \| \leq I_0(2) \|\v{x}_H^{\epsilon_{\mbox{\tiny TD}}}\|^2$ and the result follows.

We next consider the case $\v{b} \neq 0$. Using $\|A\| h \leq 1$ we have that
	\begin{align}
	    \| \v{y}^{(m+1,0)} \| & \geq \|\v{y}^{(m,1)}\| - \left\|\sum_{j=2}^k \frac{(A h)^j}{j!} \v{y}^{(m,1)} \right\| \geq \|\v{y}^{(m,1)}\| - \sum_{j=2}^k \frac{1}{j!} \left\| \v{y}^{(m,1)} \right\|\\ & = \left(1- \sum_{j=2}^k \frac{1}{j!}\right) \| \v{y}^{(m,1)} \| \geq (3-e) \| \v{y}^{(m,1)} \|.
	\end{align}
	So, $\sum_{m=1}^{M} \| \v{y}^{(m,0)} \|^2 \geq (3-e)^2 \sum_{m=0}^{M-1} \| \v{y}^{(m,1)}\|^2$. From that and $\| \v{y}^{(m,j)}\| \leq \frac{1}{j!} \| \v{y}^{(m,1)}\|$, we get
	\begin{align}
	    \sum_{m=0}^{M-1} \sum_{j=1}^k \| \v{y}^{(m,j)}\|^2 & \leq \sum_{m=0}^{M-1} \sum_{j=1}^k \left(\frac{1}{j!} \right)^2 \| \v{y}^{(m,1)}\|^2  \leq (I_0(2)-1) \sum_{m=0}^{M-1}\| \v{y}^{(m,1)}\|^2 \\ & \leq \frac{I_0(2)-1}{(3-e)^2} \sum_{m=1}^{M} \| \v{y}^{(m,0)} \|^2 \leq \frac{I_0(2)-1}{(3-e)^2} \|\v{x}_H^{\epsilon_{\mbox{\tiny TD}}}\|^2.
	\end{align}
 We conclude that
	\begin{equation}
	    \| \v{y} \|^2 \leq \left(1 + \frac{I_0(2)-1}{(3-e)^2} \right)  \|\v{x}_H^{\epsilon_{\mbox{\tiny TD}}}\|^2,
	\end{equation}
and hence,
\begin{equation}
    \mathrm{Pr}_H \ge \left(1 + \frac{I_0(2)-1}{(3-e)^2} \right)^{-1},
\end{equation}
which holds for both additive and multiplicative error conditions. We next show that for the multiplicative error case $\mathrm{Pr}_H  \geq\left( 1 + \left(1+\frac{ \lambda_\times(T)}{(1-\epsilon_{\mbox{\tiny TD}})\|A\|}\right)^2 (I_0(2)-1)  \right)^{-1} $.

We choose a function $\zeta(t)$ such that
\begin{equation}
    \|\v{b}\| \le \zeta(mh) \|A\| \|\v{x}^m\|,
\end{equation}
for all $m=0,1,2, \dots , M$, and define a constant $\zeta_* $ to be an upper bound on $\zeta(mh)$ for all $m$, which we specify shortly. Then,
\begin{equation}
    \|\v{y}^{(m,1)}\| \le \|Ah \v{y}^{(m,0)}\|+ \|h \v{b}\| \le \|\v{y}^{(m,0)}\|+ \zeta(mh)\| \v{x}^m\| = (1+ \zeta(mh)) \|\v{x}^m\|.
\end{equation}
It follows that
\begin{align}
     \sum_{m=0}^{M-1} \sum_{j=1}^k \| \v{y}^{(m,j)}\|^2 &\le  \sum_{m=0}^{M-1} \sum_{j=1}^k \frac{1}{(j!)^2}\| \v{y}^{(m,1)}\|^2 \nonumber \\
     &\le  (I_0(2) -1) \sum_{m=0}^{M} \| \v{y}^{(m,1)}\|^2 \nonumber \\
      &\le  (I_0(2) -1)(1+\zeta_*)^2 \sum_{m=0}^{M} \| \v{x}^{m}\|^2 =  (I_0(2) -1)(1+\zeta_*)^2  \|\v{x}_H^{\epsilon_{\mbox{\tiny TD}}}\|^2.
\end{align}
This upper bounds the norm of the unwanted part of $\v{y}$. 
The probability of success is hence lower bounded as
\begin{equation}
    \mathrm{Pr}_H \ge \frac{1}{1 + (I_0(2)-1)(1+\zeta_*)^2}.
 \end{equation}
Now recall that $\lambda_\times(T)$ is required to obey
\begin{equation}
    \lambda_\times (T) \ge \frac{\|\v{b}\|}{\|\v{x}(t)\|},
\end{equation}
for all $t\in [0,T]$. We also have that $\|\v{x}(mh) - \v{x}^m\| \le \epsilon_{\mbox{\tiny TD}} \|\v{x} (mh)\| $ for all $m$ by (reverse) triangle inequality. This implies that $\|\v{x}^m\| \ge (1-\epsilon_{\mbox{\tiny TD}} ) \|\v{x}(mh)\|$ for all $m$. Therefore,
\begin{equation}
\frac{1}{(1-\epsilon_{\mbox{\tiny TD}})} \frac{\lambda_\times(T)}{\|A\|} \ge  \frac{1}{(1-\epsilon_{\mbox{\tiny TD}})} \frac{\|\v{b}\|}{\|A\| \|\v{x}(mh)\|} \ge \frac{1}{\|A\|} \frac{\|\v{b}\|}{\|\v{x}^m\|}  
\end{equation}
for all $m=0,1,\dots , M$. Therefore we can take
\begin{equation}
    \zeta_* = \frac{1}{(1-\epsilon_{\mbox{\tiny TD}})} \frac{\lambda_\times(T)}{\|A\|}.
\end{equation}
This implies,
\begin{equation}
    \mathrm{Pr}_H  \geq  \left( 1 + \left(1+\frac{ \lambda_\times(T)}{(1-\epsilon_{\mbox{\tiny TD}})\|A\|}\right)^2 (I_0(2)-1)  \right)^{-1},
\end{equation}
which completes the proof.
\end{proof} 

A comment about the above bounds. Recall that we have
 \begin{equation}
    \v{x}(t) = e^{At} \v{x}(0) + \int_0^t dt' e^{At'} \v{b},
 \end{equation}
 and so for non-zero $\v{b}$ we can have $\v{x}(t) =0$ whenever
 \begin{equation}
\v{x}(0) = - \int_0^t dt' e^{A(t'-t)} \v{b}.
 \end{equation}
 Therefore, in the case of multiplicative errors, $\lambda_\times(T)$ is unbounded from above, and depends on the relations between $\v{b}, \v{x}(0)$ and $A$. Despite this, it can be shown that we can still bound the success probability above $0.05$, which in fact holds for both additive and multiplicative errors.

 We now turn to the solution state problem:

\begin{theorem}[Success probability of ODE-solver -- solution state]\label{thm:successprobfinalstate}
	Let $h>0$ be such that  $\|A\| h \leq 1$. Let $\ket{y}$ be the normalized solution to the linear system $L \v{y} = \v{c}$ with $L$ as in Eq.~\eqref{eq:L}, $p_m =0$ for all $m <M$, $k=k_{+,\times}$ an integer that gives a multiplicative or additive error $\epsilon_{\mbox{\tiny TD}}$, $p_M = p$ a positive multiple of $k+1$ and \begin{equation}
    \v{c} = [ \overbrace{\v{x}^0, \underbrace{\v{b}, 0, \dots, 0}_{k}}^{m=0}, \dots, \underbrace{\overbrace{\underbrace{0, 0, \dots, 0}_{k+1}}^{m=M}, \dots,  \overbrace{\underbrace{0, 0, \dots, 0}_{k+1}}^{m=M+\frac{p}{k+1}}}_{1+p} ].
\end{equation}

 We define $\ket{x_T^{\epsilon_{\mbox{\tiny TD}}}}$ as the state obtained by successful post-selection on the first outcome of the projective measurement $\{\Pi_{T}, I-\Pi_{T} \}$ on the clock register:
\begin{equation}
\label{eq:pif}
   \Pi_T = \sum_{m=M}^{M+p/(k+1)} \ketbra{m}{m} \otimes I \otimes I, \quad  \ket{y_T^{\epsilon_{\mbox{\tiny TD}}}} := \frac{\Pi_T \ket{y}}{\| \Pi_T \ket{y}\|}  =  \ket{x_T^{\epsilon_{\mbox{\tiny TD}}}} \otimes 
   \frac{1}{\sqrt{p}} \sum_{m=M}^{M+\frac{p}{k+1}} \sum_{j=0}^k \ket{mj}.
\end{equation}
		The success probability is lower bounded as
		\begin{equation}
			\label{eq:successprobabilityODEfinal}
			\mathrm{Pr}_T 
	 \geq \frac{1 }{ (1-\frac{I_0(2)-1}{(p+1)K})  + (M+1) \frac{I_0(2)-1}{(p+1)K} \left(\frac{1+\epsilon_{\mbox{\tiny TD}}}{1-\epsilon_{\mbox{\tiny TD}}}\right)^2  \bar{g}_{+,\times}^2}.
		\end{equation}
		where
		\begin{equation}
		   \bar{g}_\times (T)^2 =  \frac{1}{M+1} \sum_{m=0}^M \frac{\| \v{x}(mh)\|^2}{\|\v{x}(Mh)\|^2},
		\end{equation}
		with $K=1$ in the homogeneous case, $K=(3-e)^2$ in the general case. Therefore, $Pr_T = \Omega(\frac{p}{T\bar{g}_{+,\times}^2})$ in all cases.
\end{theorem}
\begin{proof}
The equality $\ket{y_T^{\epsilon_{\mbox{\tiny TD}}}} =  \ket{x_T^{\epsilon_{\mbox{\tiny TD}}}} \otimes 
	\frac{1}{\sqrt{p}} \sum_{m=M}^{M+\frac{p}{k+1}} \sum_{j=0}^k \ket{mj}$ follows from the expression for $\v{y}$ (Eq.~\eqref{eq:vectorberry2}), together with the form of the projector in Eq.~\eqref{eq:pif}.

The rest of the proof is similar to the previous theorem. First,
\begin{align}
\| \v{y} \|^2 = \sum_{m=M}^{M+p/(k+1)} \sum_{j=0}^k \| \v{y}^{(m,j)}\|^2 + \sum_{m=0}^{M-1} \sum_{j=1}^k \| \v{y}^{(m,j)}\|^2 = (p+1) \| \v{y}^{(M,0)}\|^2 + \sum_{m=0}^{M-1} \sum_{j=1}^k \| \v{y}^{(m,j)}\|^2.
\end{align}
For the unwanted part (the second piece), using the same bound as in the previous theorem we get
\begin{align}
\sum_{m=0}^{M-1} \sum_{j=1}^k \| \v{y}^{(m,j)}\|^2 \leq \frac{I_0(2)-1}{K} \sum_{m=1}^{M} \| \v{y}^{(m,0)}\|^2 
\end{align}
with $K=(3-e)^2$ for the inhomogeneous case and $K=1$ if $\v{b} =0$.

For the multiplicative case $(\times)$, the success probability is then lower bounded as in the general case as
\begin{align}
    \mathrm{Pr}_T =\frac{ (p+1)\| \v{y}^{(M,0)}\|^2}{\|\v{y}\|^2} & = \frac{ (p+1)\| \v{y}^{(M,0)}\|^2}{  (p+1)\| \v{y}^{(M,0)}\|^2 + \frac{I_0(2)-1}{K} \sum_{m=1}^M \| \v{y}^{(m,0)}\|^2} \\ &= \frac{1 }{ (1-\frac{I_0(2)-1}{K(p+1)})  + \frac{I_0(2)-1}{K (p+1)} \sum_{m=0}^M \frac{\| \v{y}^{(m,0)}\|^2}{\| \v{y}^{(M,0)}\|^2}} \\ 
    & \geq \frac{1 }{ (1-\frac{I_0(2)-1}{K(p+1)})  + \frac{I_0(2)-1}{K(p+1)} \left(\frac{1+\epsilon_{\mbox{\tiny TD}}}{1-\epsilon_{\mbox{\tiny TD}}}\right)^2 \sum_{m=0}^M \frac{\| \v{x}(mh)\|^2}{\| \v{x}(Mh)\|^2}} \\ & = \frac{1 }{ (1-\frac{I_0(2)-1}{K(p+1)})  + (M+1) \frac{I_0(2)-1}{K(p+1)} \left(\frac{1+\epsilon_{\mbox{\tiny TD}}}{1-\epsilon_{\mbox{\tiny TD}}}\right)^2 \bar{g}_\times^2}.
\end{align}

Note that the inequality in the above sequence follows from the time-discretization lemma (Lemma~\ref{lem:timediscretization}):
\begin{align}
\frac{\| \v{y}^{(m,0)}\|^2}{\| \v{y}^{(M,0)}\|^2} & = \frac{\| \v{x}^{m}\|^2}{\| \v{x}^{M}\|^2} \leq \frac{(\| \v{x}(mh)\| + \| \v{x}^{m} - \v{x}(mh)\|)^2}{(\| \v{x}(Mh)\| - \| \v{x}^{M} - \v{x}(Mh)\|)^2} \leq \frac{(\| \v{x}(mh)\| + \epsilon_{\mbox{\tiny TD}} \| \v{x}(mh)\| )^2}{(\| \v{x}(Mh)\| - \epsilon_{\mbox{\tiny TD}} \| \v{x}(Mh)\|)^2} \\ & = \left(\frac{1+\epsilon_{\mbox{\tiny TD}}}{1-\epsilon_{\mbox{\tiny TD}}}\right)^2 \frac{\| \v{x}(mh)\|^2}{\| \v{x}(Mh)\|^2}.
\end{align}
The result then follows.
For additive errors Lemma~\ref{lem:timediscretizationadditive} ensures that $\|\v{x}(mh)- \v{x}^m\| \le \epsilon_{\mbox{\tiny TD}}$, so the only difference is that
\begin{equation}
    \frac{\|\v{x}^m\|^2}{\|\v{x}^M\|^2} \le \left [ \frac{\|\v{x}(mh)\|+ \epsilon_{\mbox{\tiny TD}}}{\|\v{x}(Mh)\| - \epsilon_{\mbox{\tiny TD}}} \right]^2.
\end{equation}
We can analogously define:
\begin{equation}
		   \bar{g}_+ ^2 = \frac{1}{M+1} \sum_{m=0}^M \left(\frac{1-\epsilon_{\mbox{\tiny TD}}}{1+\epsilon_{\mbox{\tiny TD}}}\right)^2 \left [ \frac{\|\v{x}(mh)\|+ \epsilon_{\mbox{\tiny TD}}}{\|\v{x}(Mh)\| - \epsilon_{\mbox{\tiny TD}}} \right]^2.
		\end{equation}
  and obtain the same result with the replacement $\bar{g}_+ \mapsto \bar{g}_\times $.
\end{proof}

Note that the success probability for the history state behaves quite differently from that of the solution state output. One way to understand this is in terms of the changes in norm of the solution. For the solution state output we have that if $\v{x}(t)$ gets very close to zero then we have a much lower success probability. But this does not affect the history state in the same way: it simply means that the amplitudes within the history state decay, but the norm of the history state does not  -- in particular we have that $\|\v{x}_{H} \| \ge \|\v{x}(0)\|$.

As we discuss below, we do not want the choice of $p$ that we make to affect the scaling of $\kappa$ with time. Therefore if $A$ is stable we have a complexity $\kappa = O(\sqrt{M})$, otherwise we have $\kappa = O(M)$. Since $\kappa = O(p)$ we can choose either $p=\sqrt{M}$ or $p=M$, respectively, without affecting the time complexity of $\kappa$.  

If we set $p=M^{1/2}$,
\begin{equation}
    	\mathrm{Pr}_T 
	 \geq \frac{1 }{ \left (1-\frac{I_0(2)-1}{M^{1/2} K} \right)  + \frac{(M+1)}{M^{1/2}} \frac{I_0(2)-1}{K} \left(\frac{1+\epsilon_{\mbox{\tiny TD}}}{1-\epsilon_{\mbox{\tiny TD}}}\right)^2 \bar{g}_{+,\times} ^2}.
\end{equation}
With amplitude amplification, we need $O(1/\sqrt{	\mathrm{Pr}_T })$ queries to the whole ODE linear solver protocol to obtain success probability $O(1)$. This number is $O(\bar{g}_{+,\times} M^{1/4})$. 

Instead, if we set $p=M$
\begin{align}
    \mathrm{Pr}_T \geq \frac{1 }{ \left (1-\frac{I_0(2)-1}{K(M+1)} \right )  + \frac{I_0(2)-1}{K} \left(\frac{1+\epsilon_{\mbox{\tiny TD}}}{1-\epsilon_{\mbox{\tiny TD}}}\right)^2 \bar{g}_{+,\times} ^2}.
\end{align}
With amplitude amplification, we need $O(\bar{g}_{+,\times} )$ queries to the whole ODE linear solver protocol to obtain success probability $O(1)$.  The same considerations apply if we run the algorithm with additive errors per time-step. The same relations apply with the replacement $\bar{g}_\times \mapsto \bar{g}_+$.

\section{Error propagation}
\label{sec:errorpropagation}

The quantum linear solver applied to $L$ returns a mixed quantum state $\rho$ such that
\begin{equation}
\| \rho - \ketbra{y}{y} \|_1 \leq \epsilon_L,
\end{equation}
  where $\epsilon_L$ is the linear solver error, $\ket{y} = \v{y}/ \| \v{y}\|$ and $\v{y}=L^{-1} \v{c}$. The solution of the linear system is then projectively measured. Upon success the target (history or final) state is returned. The success probability in the ideal case has been analyzed in Sec.~\ref{sec:successprobabilityanalysis}. 
  
  In the non-ideal case, however, both the success probability and the post-measurement state are perturbed due to the finite precision of the QLSA. This is quantified by the following continuity statement with respect to the $\ell_1$ norm.
\begin{lem}[Continuity of success probability and post-selected state]
	\label{lemma:successpostmeasurementstateperturbation}
Let $\Pi$ be a Hermitian projector, $\rho \geq 0 $, $\tilde{\rho} \geq 0$ with $\| \tilde{\rho} - \rho \| _1 \leq \epsilon$ and $\tr{\Pi \rho} >\epsilon$. Then
\begin{equation}
	\left\| \frac{\Pi \rho \Pi}{\tr{\rho \Pi} } -  \frac{\Pi \tilde{\rho} \Pi}{\tr{\tilde{\rho} \Pi} } \right\|_1 \leq \frac{2\epsilon}{\tr{ \rho \Pi}-\epsilon}  
\end{equation}
and $\tr{\tilde{\rho} \Pi} \geq \tr{\rho \Pi} - \epsilon >0$.
\end{lem}
	\begin{proof}
		First note that
\begin{equation}
	\| \Pi (\rho - \tilde{\rho}) \Pi \|_1 \leq \| \Pi\|^2 \| \rho - \tilde{\rho} \|_1 \leq \epsilon,
\end{equation}
and
\begin{equation}
	| \tr{\rho \Pi} - \tr{ \tilde{\rho} \Pi} |  = |\tr{(\rho- \tilde{\rho}) \Pi}| \leq \| \rho - \tilde{\rho} \|_1 \| \Pi \| \leq \epsilon.
\end{equation}
		It follows that
		\begin{equation}
			\tr{\tilde{\rho} \Pi} \geq \tr{\rho \Pi} - \epsilon. 
		\end{equation}
	Note that $\| \Pi \rho \Pi \| _1 = \tr{\Pi \rho}$. Then,
		\begin{align}
			\left\| \frac{\Pi \rho \Pi}{\tr{\Pi \rho}} - \frac{\Pi \tilde{\rho} \Pi}{\tr{\Pi \tilde{\rho}}} \right\|_1 & \leq \left\|  \frac{\Pi \rho \Pi}{\tr{\Pi \rho}} - \frac{\Pi \rho \Pi}{\tr{P\tilde{\rho}}} \right\|_1 + \left\| \frac{\Pi \rho \Pi}{\tr{\Pi\tilde{\rho}}} - \frac{\Pi \tilde{\rho} \Pi}{\tr{\Pi \tilde{\rho}}} \right\|_1 \\ & =  \| \Pi \rho \Pi \|_1 \frac{| \tr{\Pi \rho} - \tr{\Pi \tilde{\rho}}|}{\tr{ \Pi\rho} \tr{ \Pi\tilde{\rho}} } +  \frac{\left\|\Pi (\rho- \tilde{\rho}) \Pi\right\|_1}{\tr{\Pi\tilde{\rho}}}
			\\ & \leq \frac{\epsilon}{\tr{ \Pi\tilde{\rho}} } +  \frac{\epsilon}{\tr{\Pi\tilde{\rho}}} \leq \frac{2 \epsilon}{\tr{\Pi \rho}-\epsilon}.
		\end{align} 
	\end{proof}
	
Applied to our setting this implies the following result:
\begin{lem}
\label{lem:totalerror}
 Assume we run a QLSA algorithm returning $\rho$ with $\| \rho - \ketbra{y}{y} \|_1 \leq \epsilon_L$, with $\ket{y}$ the normalized solution of $L \v{y} = \v{c}$. Moreover, assume that $\epsilon_L$ is chosen sufficiently small so that $Pr_{H,T} > \epsilon_L$. Then denote by 
\begin{equation}
     \sigma_{H} := \frac{\Pi_{H} \rho \Pi_{H}}{\tr{\Pi_{H} \rho} }, \quad \ketbra{x_{H}^{\epsilon_{\mbox{\tiny TD}}}}{x_{H}^{\epsilon_{\mbox{\tiny TD}}}} := \frac{\Pi_{H} \ketbra{y}{y} \Pi_{H}}{\tr{\Pi_{H} \ketbra{y}{y} }}.
 \end{equation}
 
 \begin{equation}
     \sigma_{T} := \mathrm{tr}_{CA}\frac{\Pi_{T} \rho \Pi_{T}}{\tr{\Pi_{T} \rho }}, \quad \ketbra{x_{T}^{\epsilon_{\mbox{\tiny TD}}}}{x_{T}^{\epsilon_{\mbox{\tiny TD}}}} := \mathrm{tr}_{CA} \frac{\Pi_{T} \ketbra{y}{y} \Pi_{T}}{\tr{\Pi_{T} \ketbra{y}{y} }}= \ketbra{x^M}{x^M},
 \end{equation}
 where the trace is over the clock $\{ \ket{m}\}$ and ancillary $\{ \ket{j} \}$ registers. Then,
\begin{equation}
\label{eq:junkerror}
	\| \sigma_{H,T} - \ketbra{x_{H,T}^{\epsilon_{\mbox{\tiny TD}}}}{x_{H,T}^{\epsilon_{\mbox{\tiny TD}}}} \|_1 \leq \frac{2 \epsilon_L}{\mathrm{Pr}_{H,T} - \epsilon_L},
\end{equation}
where $\mathrm{Pr}_{H,T} = \tr{\Pi_{H,T} \ketbra{y}{y} \Pi_{H,T}}$.
\end{lem}
\begin{proof}
The statement for the history state follows from Lemma~\ref{lemma:successpostmeasurementstateperturbation} with $\tilde{\rho} = \ketbra{y}{y}$. The statement for the solution state instead follows using the contractivity of the trace-norm under partial trace, followed by an application of Lemma~\ref{lemma:successpostmeasurementstateperturbation} with $\tilde{\rho} = \ketbra{y}{y}$.
\end{proof}
Now we analyze the total error in terms of the $1$-norm distance between the output state (after QLSA, measurement and discarding ancillas) and the ideal target.

\begin{theorem}[Total error bound]
\label{lem:totalerrorbound}
 Let $L$ be the matrix encoding of the discretized dynamics  within the linear system $L \ket{y} = \ket{c}$, where $L$ is defined in Eq.~\eqref{eq:L},  and $\v{c}$ in Eq.~\eqref{eq:cdefinition}. We take the truncation $k=k_{+,\times}$ to be an integer that gives a multiplicative or additive error $\epsilon_{\mbox{\tiny TD}}$ in the time-discretization. Assume we run a QLSA algorithm returning a normalized quantum state $\rho$ with $\| \rho - \ketbra{y}{y} \|_1 \leq \epsilon_L$. Denote by $P_{H,T}$ the success probabilities defined in Theorems~\ref{thm: successprobhistorystate} and \ref{thm:successprobfinalstate}, respectively, and $\ket{x_{H,T}}$ the ideal history and solution state defined as in Eq.~\eqref{eq:finalstate}, \eqref{eq:historystate}. Finally, assume that $\epsilon_L$ is chosen sufficiently small so that $Pr_{H,T} > \epsilon_L$. Then we have that
    \begin{equation}
    	\| \sigma_{H,T} - \ketbra{x_{H,T}}{x_{H,T}} \|_1 \leq \frac{2 \epsilon_L}{\mathrm{Pr}_{H,T} - \epsilon_L}+ 4 \epsilon_{\mbox{\tiny TD}} f^{\times,+}_{H,T},
\end{equation}
where $f^\times_{H} := f^\times_{T} := 1$, $f^+_H :=\frac{1}{ x_{\mathrm{rms}} }$, $f^+_T :=\frac{1}{\|\v{x}(T)\|},$ and where we define 
\begin{equation}
  x_{\mathrm{rms}} := \sqrt{\sum_{m=0}^M \|\v{x}(mh)\|^2/M}. 
\end{equation}
Hence, for the multiplicative $(\times)$ case, if $\epsilon_L = O(\epsilon Pr_{H,T})$ and $\epsilon_{\mbox{\tiny TD}} = O(\epsilon)$ then 
\begin{equation}
\| \sigma_{H,T} - \ketbra{x_{H,T}}{x_{H,T}} \|_1 = O(\epsilon).
\end{equation} 
For the additive $(+)$ case, to obtain the latter error bound, we instead require $\epsilon_L = O(\epsilon Pr_{H,T})$, and $\epsilon_{\mbox{\tiny TD}} = O(\epsilon \|\v{x}(T)\|)$ in the case of the solution state output and $\epsilon_{\mbox{\tiny TD}} = O(\epsilon  x_{\mathrm{rms}} )$ in the case of the history state output. 
\end{theorem}
\begin{proof} 
To show the claimed error bound, we first bound the vector norm error between the discrete and continuous dynamics. Recall that from Corollary~\ref{cor:unitvectorerrors}, with the notation $\ket{x_{T}^{\epsilon_{\mbox{\tiny TD}}}}:= \ket{x^M}$
\begin{equation}
    \| \ket{x^{\epsilon_{\mbox{\tiny TD}}}_{H,T}} - \ket{x_{H,T}} \| \leq 2 \epsilon_{\mbox{\tiny TD}} f^{+,\times}_{H,T},
\end{equation} 
where we introduce the factors $f^{+,\times}_{H,T}$ so as to handle together the four cases additive/multiplicative, history/final state. We have $f^{\times}_{H,T}=1$ for the case of multiplicative errors, and in the case of additive errors these factors are given by
\begin{align}
    f^+_{H} &:= \frac{1}{\sqrt{ \sum_{m=0}^M \|\v{x}(mh))\|^2/M}} = \frac{1}{ x_{\mathrm{rms}} }\nonumber \\
    f^+_T &:= \frac{1}{\|\v{x}(T)\|} = \frac{1}{\|\v{x}_T\|}.
\end{align}	
Now we have that
  \begin{equation}
    	\| \sigma_{H,T} - \ketbra{x_{H,T}}{x_{H,T}} \|_1 \leq 	\| \sigma_{H,T} - \ketbra{x_{H,T}^{\epsilon_{\mbox{\tiny TD}}}}{x_{H,T}^{\epsilon_{\mbox{\tiny TD}}}} \|_1  + 	\| \ketbra{x_{H,T}^{\epsilon_{\mbox{\tiny TD}}}}{x_{H,T}^{\epsilon_{\mbox{\tiny TD}}}} - \ketbra{x_{H,T}}{x_{H,T}} \|_1. 
  \end{equation}
  We can bound the first term using Eq.~\eqref{eq:junkerror}. 
  
  For the second term, we first note that 
  \begin{equation}
    \| \ketbra{x_{H,T}^{\epsilon_{\mbox{\tiny TD}}}}{x_{H,T}^{\epsilon_{\mbox{\tiny TD}}}} - \ketbra{x_{H,T}}{x_{H,T}} \|_1 = 2 \sqrt{1- \left|\braket{x_{H,T}^{\epsilon_{\mbox{\tiny TD}}}|x_{H,T}}\right|^2}. 
  \end{equation}
  Now, 
    \begin{equation}
      \| \ket{x_{H,T}^{\epsilon_{\mbox{\tiny TD}}}} - \ket{x_{H,T}} \|^2 = 2 ( 1- \mathrm{Re} \braket{x_{H,T}^{\epsilon_{\mbox{\tiny TD}}}|x_{H,T}}) \geq 2 \left( 1- \left|\braket{x_{H,T}^{\epsilon_{\mbox{\tiny TD}}}|x_{H,T}}\right|\right)
  \end{equation}
  Hence,
  \begin{align}
       \left( 1- \left|\braket{x_{H,T}^{\epsilon_{\mbox{\tiny TD}}}|x_{H,T}}\right|^2\right) &=   \left( 1- \left|\braket{x_{H,T}^{\epsilon_{\mbox{\tiny TD}}}|x_{H,T}}\right|\right)  \left( 1+ \left|\braket{x_{H,T}^{\epsilon_{\mbox{\tiny TD}}}|x_{H,T}}\right|\right) \\
     & \le \frac{1}{2} \| \ket{x_{H,T}^{\epsilon_{\mbox{\tiny TD}}}} - \ket{x_{H,T}} \|^2 \left( 1+ \left|\braket{x_{H,T}^{\epsilon_{\mbox{\tiny TD}}}|x_{H,T}}\right|\right) \\
     &\leq  \| \ket{x_{H,T}^{\epsilon_{\mbox{\tiny TD}}}} - \ket{x_{H,T}} \|^2.
  \end{align}
    We conclude
    \begin{equation}
          \| \ketbra{x_{H,T}^{\epsilon_{\mbox{\tiny TD}}}}{x_{H,T}^{\epsilon_{\mbox{\tiny TD}}}} - \ketbra{x_{H,T}}{x_{H,T}} \|_1 \leq 2 \sqrt{   \| \ket{x_{H,T}^{\epsilon_{\mbox{\tiny TD}}}} - \ket{x_{H,T}} \|^2 }  = 2 \| \ket{x_{H,T}^{\epsilon_{\mbox{\tiny TD}}}} - \ket{x_{H,T}} \| \leq 4 \epsilon_{\mbox{\tiny TD}} f^{+,\times}_{H,T}.
  \end{equation}
 Putting the two bounds together we get the claimed result. Moreover, if we set $\epsilon_L = \epsilon Pr_{H,T}$ then the first fractional term becomes
 \begin{equation}
     \frac{2 \epsilon_L}{\mathrm{Pr}_{H,T} - \epsilon_L} = 2\epsilon = O(\epsilon).
 \end{equation}
 and therefore choosing $\epsilon_L = O(\epsilon Pr_{H,T})$ and $\epsilon_{\mbox{\tiny TD}} = O(\epsilon)$ gives the claimed $O(\epsilon)$ scaling.
\end{proof}

\section{Fast-Forwarding, No Fast-Forwarding, and relation to linear solvers}
\label{sec:fastforwarding}

We now discuss our results in relation to asymptotic complexities and the relation between our ODE solver and both the problem of Hamiltonian simulation and quantum linear solvers.

\subsection{Fast-Forwarding of stable systems}

The ODE solver presented in this work relies upon a quantum linear solver algorithm (QLSA). The best query complexity is $O(\kappa \log (1/\epsilon))$, which we use to establish the claimed fast-forwarding result. For the explicit application to a concrete system we will however use the recent explicit query counts upper bounds given in~\cite{jennings2023efficient}.

We now establish the asymptotic fast-forwarding of stable systems, claimed in the main text. 
\begin{theorem}[Fast-forwarding of stable systems]
    Let $A\in \mathbb{C}^{N\times N}$ be an $N\times N$ matrix, and $\v{b} \in \mathbb{C}^N$. Choose a timestep $h$ such that $\|A\|h \le 1$. Consider the following linear ODE system, 
\begin{equation}
\label{eq:ODE3}
    \dot{\v{x}}(t) = A\v{x}(t) + \v{b}, \quad \v{x}(0) = \v{x}^0,
\end{equation} 
where $t\in [0,T]$, $T=Mh$ and the $\v{x}^0 \in \mathbb{C}^N$ specifying the initial conditions.
The four parameters  $\lambda_{+,\times}(T)$ and $\bar{g}_{+,\times} (T)$ are defined as follows:
    \begin{equation}
        \lambda_+(T) = \max \left\{ \|\v{b}\| T^2,  x_{\max}  T  \right\}, \quad \bar{g}_+(T)^2 = \frac{1}{M+1}\sum_{m=0}^M \left(\frac{1-\epsilon'}{1+\epsilon'}\right)^2 \frac{(\| \v{x}(mh) \|+\epsilon')^2}{  (\|\v{x}(Mh)\|-\epsilon')^2},
    \end{equation}
      \begin{equation}
        \lambda_\times(T) = \max \left\{ T^2 \frac{\|\v{b}\|}{ x_{\min} },  T  \right\}, \quad \bar{g}_\times(T)^2 = \frac{1}{M+1}\sum_{m=0}^M \frac{\| \v{x}(mh) \|^2}{  \|\v{x}(Mh)\|^2}.
    \end{equation}
    
The terms $\kappa_P$ and $\mu_P(A)$ are stability parameters satisfying 
    \begin{equation}
        \|e^{At}\| \leq \sqrt{\kappa_P} e^{\mu_P(A)t} \quad \mbox{ for all } t \in [0,T].
    \end{equation}
    The matrix $A$ is said to be stable precisely when such a $\mu_P(A)$ exists satisfying the previous condition with $T\rightarrow \infty$, and where $\mu_P(A) <0$. If $A$ is not stable, then we assume a constant upper bound $\|e^{At}\|\le C_{\max}$ for all $t \in [0,T]$.
    
Assume we have oracle access to an $(\omega, a, 0)$ block-encoding $U_A$  of the matrix $A$, a unitary $U_b$ that prepares the normalized state $\ket{b}$, and a unitary $U_0$ that prepares the normalized state $\ket{x^0}$. Then a density operator $\rho_\epsilon$ that is $\epsilon$-close in $1$-norm to either the history state $\ket{x_H}$ or to the solution state $\ket{x_T}$ can be output, with constant success probability, with query complexity to the oracles $U_A$, $U_0$ and $U_b$ as summarized by the following table of complexities:
\scriptsize
\begin{center}
\begin{tabular}{|l|c|c|} 
  \hline 
  \textbf{Case} & \textbf{Solution State}& \textbf{History State} \\ 
  \hline
  Stable $A$, $(\times)$ & $O\left ( \omega \sqrt{\kappa_P} \bar{g}_\times (T) T^{3/4} \log[\lambda_\times(T)/\epsilon]\log(T\bar{g}_\times/\epsilon) \right )$ & $O\left (  \omega\sqrt{\kappa_P} \sqrt{T} \log [\lambda_\times(T)/\epsilon]\log(1/\epsilon) \right )$ \\ 
  Stable $A$, $(+)$& $O\left ( \omega \sqrt{\kappa_P} \bar{g}_+ (T) T^{3/4} \log[\lambda_+(T)/(\|\v{x}(T)\|\epsilon)]\log(T\bar{g}_+/\epsilon) \right )$ &$O\left (  \omega\sqrt{\kappa_P} \sqrt{T} \log [\lambda_+(T)/ (x_{\mathrm{rms}}\epsilon)]\log(1/\epsilon) \right )$ \\ 
  Not stable $A$, $(\times)$ & $ O \left (\omega\bar{g}_\times (T)C_{\mathrm{max}} T \log [\lambda_\times(T)/\epsilon]\log(\bar{g}_\times /\epsilon) \right )$  & $O \left ( \omega C_{\mathrm{max}} T \log [\lambda_\times(T)/\epsilon]\log(1/\epsilon) \right)$\\ 
  Not stable $A$, $(+)$ & $ O \left ( \omega \bar{g}_+ (T)C_{\mathrm{max}} T \log [\lambda_+(T)/(\|\v{x}(T)\|\epsilon)]\log(\bar{g}_+/\epsilon) \right )$  &  $O \left (  \omega C_{\mathrm{max}} T \log [\lambda_+ (T) /(x_{\mathrm{rms}} \epsilon)]\log(1/\epsilon) \right)$ \\ 
  \hline
\end{tabular}
\end{center}
\normalsize
In the above, $(\times)$ denotes the choice of multiplicative error scheme, and $(+)$ denotes the choice of additive error scheme.
\end{theorem}

\begin{proof} 
We proceed as in the main text, and produce a solution by first embedding the dynamical problem into a linear system of equations $L \v{y} = \v{c}$ and then solving this with a quantum linear solver. Recall that to solve $L \v{y} = \v{c}$ we first rescale by a global constant to $\tilde{L} \v{y}=\tilde{\v{c}}$ so that $\|\tilde{L}\| \le 1$. This puts the problem in a canonical form for a quantum linear solver algorithm to solve $\tilde{L}\ket{y} = \ket{\tilde{c}}$. This was detailed in Appendix~\ref{sec:block-encoding}, and it was shown that we may obtain a state preparation unitary $U_{\tilde{c}}$ that prepares $\ket{\tilde{c}}$ via a single call to $U_b$ and a single call to $U_0$. We also showed that we may obtain an $(\tilde{\omega},a+6,0)$ unitary block-encoding $U_{\tilde{L}}$ for $\tilde{L}$ with a scaling factor $\tilde{\omega} = O(1+\omega h/\sqrt{k})$, with a single query to the $(\omega, a, 0)$ block-encoding $U_A$ for $A$. Given access to $U_{\tilde{L}}$ and $U_{\tilde{c}}$, this linear system can therefore be solved to error $\epsilon_L$ by an asymptotically optimal quantum linear solver such as in~\cite{costa2022optimal} with a query complexity
\begin{equation}
O\left (\tilde{\omega} \kappa_{\tilde{L}} \log (1/\epsilon_L) \right ) = O\left ( \left (1+\frac{\omega h}{\sqrt{k}} \right) \kappa_{L} \log (1/\epsilon_L) \right),
\end{equation}
to the oracles $U_{A}$, $U_{b}$ and $U_0$. This outputs the coherent encoding solution $\ket{y}$ with bounded probability. We continue the proof, by splitting up the case when $A$ is stable and when $A$ is not stable.

\begin{itemize}
    \item \emph{($A$ is a stable matrix)}. From Theorem~\ref{thm:conditionnumbernoidling}, we have that for both the multiplicative error assumption and the additive error assumption that  $\kappa_L = O(k_{+,\times} \sqrt{\kappa_P T+p^2})$, and we choose $p=\Theta(\sqrt{T})$ so that $\kappa_L = O(k_{+,\times} \sqrt{\kappa_P T})$. From Lemma~\ref{lem:sufficientk} we find that to obtain a discretization error $\epsilon_{\mbox{\tiny TD}}$ we take $k_{+,\times} = \Theta(\log(\lambda_{+,\times}(T)/\epsilon_{\mbox{\tiny TD}}))$, where $+, \times$ refer to the choice of setting additive or multiplicative errors per step. 
    This implies that,
    \begin{align}
    \mathcal{Q}_{QLSA}&= O\left (\sqrt{\kappa_P}  \sqrt{T}\left (\log[\lambda_{+,\times}(T)/\epsilon_{\mbox{\tiny TD}}] +\omega h\log^{1/2}[\lambda_{+,\times}(T)/\epsilon_{\mbox{\tiny TD}}] \right )\log(1/\epsilon_L) \right) \nonumber \\
    &= O\left (\sqrt{\kappa_P}  \omega \sqrt{T }\log[\lambda_{+,\times}(T)/\epsilon_{\mbox{\tiny TD}}] \log(1/\epsilon_L) \right),
    \end{align} 
    suffices to output the required approximation $\tilde{\rho}$ to $\ket{y}$ with probability of success bounded from below by a constant. We next perform the 2--outcome projective measurement $\{\Pi_{H,T}, I -\Pi_{H,T}\}$ on $\tilde{\rho}$, where $H$ and $F$ denote whether we want to realize the history state or solution state respectively. This gives rise to the normalized quantum state $\sigma_{H,T}$ that approximates the exact, continuous dynamics solution $\ket{x_{H,T}}$. We assume that $\epsilon_L < Pr_{H,T}$, and so from Theorem~\ref{lem:totalerrorbound}, we have that
 \begin{equation}
    	\| \sigma_{H,T} - \ketbra{x_{H,T}}{x_{H,T}} \|_1 \leq \frac{2 \epsilon_L}{\mathrm{Pr}_{H,T} - \epsilon_L}+ 4 \epsilon_{\mbox{\tiny TD}}  f^{+,\times}_{H,T}.
\end{equation}
To make the trace-distance $O(\epsilon)$, it suffices to take $\epsilon_L= O(\epsilon \mathrm{Pr}_{H,T})$ and $\epsilon_{\mbox{\tiny TD}} = O(\epsilon/f^{+,\times}_{H,T})$ from Theorem~\ref{lem:totalerrorbound}. This also implies, from Lemma~\ref{lemma:successpostmeasurementstateperturbation}, that the success probability on the approximate state $\tilde{\rho}$ is $\Theta (\mathrm{Pr}_{H,T})$. Therefore, to obtain an $\epsilon$--approximation to the exact solution with probability $O(1)$, it suffices to use amplitude amplification with $\Omega(1/\sqrt{\mathrm{Pr}_{H,T}})$ queries to the circuit. From Theorem~\ref{thm:successprobfinalstate} we have that $\mathrm{Pr}_T =  \Omega(p/(T\bar{g}_{+,\times}^2))= \Omega( 1/(\sqrt{T}\bar{g}_{+,\times}^2))$ in the stable case where we choose $p=\Theta(\sqrt{T})$. Also, from Theorem~\ref{thm: successprobhistorystate} we have that $\mathrm{Pr}_H = \Omega(1)$ in all cases. Thus, in the case of a stable system we have $O(T^{1/4} \bar{g}_{+,\times})$ and $O(1)$ calls to the circuit within amplitude amplification for the solution state/history state cases respectively. This implies that for the case where $A$ is a stable matrix and we have a target error $\epsilon$ for the exact dynamics, we have an overall query complexity 
\begin{equation}
 O \left (\bar{g}_{+,\times}(T)\sqrt{\kappa_P}  \omega (T)^{3/4}\log[\lambda_{+,\times}(T)f^{+,\times}_T/\epsilon]\log(\sqrt{T}\bar{g}_{+,\times}/\epsilon) \right),
\end{equation}
in the case of outputting the solution state $\ket{x_T}= \ket{x(T)}$, and
\begin{equation}
 O \left (  \sqrt{\kappa_P} \omega\sqrt{T}\log[\lambda_{+,\times}(T)f^{+,\times}_H/\epsilon]\log(1/\epsilon) \right),
\end{equation}
in the case of outputting the history state $\ket{x_H}$. This establishes the claimed fast-forwarding for stable systems. 
 
\item \emph{($A$ is not a stable matrix)}. From Theorem~\ref{thm:conditionnumbernoidling}, we have that $\kappa_L = O(k_{+,\times} \sqrt{C^2_{\mathrm{max}}T^2+p^2})$, where $C_{\mathrm{max}}$ is an upper bound on $\|e^{At}\|$ over the interval $t \in [0,T]$. 
In the case of outputting the solution state we can choose the idling parameter to be $p=\Theta(T)$ which implies firstly that $\kappa_L = O( C_{\mathrm{max}}k_{+,\times} T)$, and secondly that $\mathrm{Pr}_T = \Omega(p/(T\bar{g}_{+,\times}^2)) = \Omega (\bar{g}_{+,\times}^2)$. The choice of Taylor truncation order $k_{+,\times}$ is as in the case of $A$ being stable, and we again choose $\epsilon_L = O(\epsilon Pr_{H,T})$ and $\epsilon_{\mbox{\tiny TD}} = O(\epsilon/f^{+,\times}_{H,T})$. This implies
  \begin{equation}
    \mathcal{Q}_{QLSA}= O\left (  \omega C_{\mathrm{max}}T\log[\lambda_{+,\times}f^{+,\times}_{H,T}(T)/\epsilon]\log(\bar{g}_{+,\times}/\epsilon) \right),
    \end{equation} 
queries for outputting with bounded probability an $\epsilon_L$ approximation to $\ket{y}$. For outputting the solution state, we then require $O(\bar{g}_{+,\times})$ rounds within amplitude amplification to obtain, with bounded probability, the state $\sigma_{T}$ within $O(\epsilon)$ trace-norm distance from the exact solution $\ket{x_{T}}$, while for the history state it suffices to repeat $O(1)$ times to obtain an $O(\epsilon)$ approximation to $\ket{x_H}$. Therefore, we have an overall query complexity 
\begin{equation}
 O \left (  \omega \bar{g}_{+,\times} C_{\mathrm{max}} T\log[\lambda_{+,\times}f^{+,\times}_{T}(T)/\epsilon]\log(\sqrt{T}\bar{g}_{+,\times}/\epsilon) \right),
\end{equation}
in the case of outputting the solution state $\ket{x_T}= \ket{x(T)}$, and
\begin{equation}
 O \left ( \omega C_{\mathrm{max}}T \log[\lambda_{+,\times}(T)f^{+,\times}_{H}/\epsilon]\log(1/\epsilon) \right),
\end{equation}
in the case of outputting the history state $\ket{x_H}$.
\end{itemize}
This completes the proof.
\end{proof}

\begin{table*}
\resizebox{\textwidth}{!}{
      \begin{tabular}{ | p{0.66cm} | p{8cm}  | p{7.1cm} | p{2.39cm} |}
    \hline 
    \textbf{Ref.} & \textbf{Solution state} & \textbf{History state} & \textbf{Assumptions} \\ \hline
    \cite{berry2017quantum} & $\tilde{O}( \kappa_V g T \log [\frac{1}{\epsilon} \max\{\frac{\|\v{b}\|}{ x_{\min} },T\}] \log\frac{g T }{\epsilon})$ & $\tilde{O}( \kappa_V T \log [\frac{1}{\epsilon} \max\{\frac{\|\v{b}\|}{ x_{\min} },T\}] \log\frac{1}{\epsilon})$ &  $\alpha(A) \le 0$  \\ 
     \cite{krovi2022improved} & $\tilde{O}( \kappa_V g  T \log^{3/2} [\frac{1}{\epsilon} \max\{\frac{\|\v{b}\|}{ x_{\min} },T\}] \log\frac{g T}{\epsilon})$ & $\tilde{O}( \kappa_V T \log^{3/2} [\frac{1}{\epsilon} \max\{\frac{\|\v{b}\|}{ x_{\min} },T\}] \log\frac{1}{\epsilon})$ & $A$ diag. via~$V$ \\
    Here &  $O( \kappa_V \v{\bar{g}_\times} T \log [\frac{1}{\epsilon} \max\{\frac{\|\v{b}\|}{ x_{\min} },T\}]\log\frac{\bar{g}_\times T}{\epsilon})$ &  $O( \kappa_V  T \log [\frac{1}{\epsilon} \max\{\frac{\|\v{b}\|}{ x_{\min} },T\}] \log\frac{1}{\epsilon})$ &  \small $\kappa_V \! \! =\! \|V\| \|V^{-1}\|$ \normalsize  \\ 
    \hline \hline 
     
     \cite{berry2017quantum} & beyond scope & beyond scope &  $A$ stable, \\
     
     \cite{krovi2022improved} & $\tilde{O}( C_{\mathrm{max}} g  T \log^{3/2} [\frac{1}{\epsilon} \max\{\frac{\|\v{b}\|}{ x_{\min} },T\}] \log\frac{g T}{\epsilon})$ & $\tilde{O}( C_{\mathrm{max}} T \log^{3/2} [\frac{1}{\epsilon} \max\{\frac{\|\v{b}\|}{ x_{\min} },T\}]\log\frac{1}{\epsilon})$ &  i.e., $\alpha(A) <0$
     \\
    Here & $O(\v{\sqrt{\kappa_P} \bar{g}_\times   T^{3/4} \log [\frac{1}{\epsilon} \max\{\frac{\|\v{b}\|}{ x_{\min} },T\}] \log\frac{\bar{g}_\times T} {\epsilon}})$ & $O(\v{\sqrt{\kappa_P}  \sqrt{T} \log [\frac{1}{\epsilon} \max\{\frac{\|\v{b}\|}{ x_{\min} },T\}] \log\frac{1 }{\epsilon}})$ & 
    \\ 
       & $O(\v{\sqrt{\kappa_P} \bar{g}_+   T^{3/4} \log \frac{\max \left\{ \|\v{b}\|,  x_{\max}    \right\} T }{\epsilon \| \v{x}(T)\|} \log \frac{\bar{g}_+ T}{\epsilon}})$ & $O(\v{\sqrt{\kappa_P}  \sqrt{T} \log \frac{ \max \left\{ \|\v{b}\|,  x_{\max}   \right\} T }{\epsilon \, x_{{\rm rms}}} \log \frac{1}{\epsilon}})$ &
       \\ \hline 
       \cite{an2022theory} &  $\frac{\| \v{x}^0\| + \| \v{b}\|/|\alpha(A)|}{\| \v{x}(T)\|} \frac{\sqrt{T}}{|\alpha(A)|} \log \frac{1}{|\alpha(A)|\epsilon} \log \frac{T \log(1/(|\alpha(A)|\epsilon))}{|\alpha(A)|\epsilon} $ & - &  $A<0$
       \\
       \hline \hline 
      \cite{berry2017quantum} & beyond scope & beyond scope &  $A$ general \\
\cite{krovi2022improved} & $\tilde{O}( C_{\mathrm{max}} g  T \log^{3/2} [\frac{1}{\epsilon} \max\{\frac{\|\v{b}\|}{ x_{\min} },T\}] \log\frac{g T}{\epsilon})$ & $\tilde{O}( C_{\mathrm{max}} T \log^{3/2} [\frac{1}{\epsilon} \max\{\frac{\|\v{b}\|}{ x_{\min} },T\}]\log\frac{1}{\epsilon})$ & 
     \\ 
Here &  $O  (C_{\mathrm{max}} \v{\bar{g}_{\times}} T\v{\log}[\frac{1}{\epsilon}\max\{\frac{\|\v{b}\|}{ x_{\min} },T\}]\log \frac{\bar{g}_{\times}T}{\epsilon} )$ & $O (C_{\mathrm{max}}T \log[\frac{1}{\epsilon}\max\{ \frac{\| \v{b}\|}{ x_{\min} }, T\} ]\log\frac{1}{\epsilon} )$ & \\
&  $\v{O (C_{\mathrm{max}} \bar{g}_{+}  T\log[\frac{\max\{\| \v{b}\|,  x_{\max}   \}T}{\| \v{x}(T)\|\epsilon }]\log \frac{\bar{g}_{+} T}{\epsilon} )}$ & $\v{O (C_{\mathrm{max}}T \log[\frac{\max\{ \| \v{b}\|,  x_{\max}  \}T}{\epsilon \, x_{{\rm rms}}}   ]\log\frac{1}{\epsilon} )}$   & \\
    \hline 
    \end{tabular}  
    }
   \caption{\textbf{Asymptotic query complexities for linear ODEs with optimal quantum linear solvers for general~$\v{b}$}. Comparison of asymptotic complexities of the  algorithm in Ref.~\cite{berry2017quantum}, its new analysis (this work) and the recent improved algorithm in Ref.~\cite{krovi2022improved}. A simplified table (Table~\ref{table:comparisonssimplified}) was presented in the main text.  Note that the parameter $\omega$ enters linearly in all algorithmic complexities in the table above.} \label{table:comparisons}
\end{table*}

The table of complexities \ref{table:comparisons} is obtained by expanding the definitions of $\lambda_+(T)$ and $\lambda_\times (T)$. With this we can use Lemma~\ref{lem:kappa-cases} for the special cases for $\kappa_P$, to obtain the following corollary. 
\begin{cor}
Let $A\in \mathbb{C}^{N \times N}$ be a stable, complex matrix, and $\v{b} \in \mathbb{C}^N$. Consider the following linear ODE system, 
\begin{equation}
    \dot{\v{x}}(t) = A\v{x}(t) + \v{b}, \quad \v{x}(0) = \v{x}^0,
\end{equation} 
where $t\in [0,T]$, and the $\v{x}^0 \in \mathbb{C}^N$ specifying the initial conditions and the parameters $\bar{g}_{+,\times} (T)$ and $\lambda_{+,\times}(T)$ defined as in the previous theorem. Assume we have oracle access to an $(\omega, a, 0)$ block-encoding $U_A$  of the matrix $A$, a unitary $U_b$ that prepares the normalized state $\ket{b}$, and a unitary $U_0$ that prepares the normalized state $\ket{x^0}$. Then the following special cases hold:
\begin{itemize}
    \item If $A$ is diagonalizable as $V^{-1}AV = D$ for some diagonal matrix $D$, then we may output an $\epsilon$--approximation to the final state $\ket{x_T}= \ket{x(T)}$ using
\begin{equation}
        O\left ( \omega \kappa_V \bar{g}_{+,\times} (T) T^{3/4} \log[\lambda_{+,\times}(T)f^{+,\times}_{H,T}/\epsilon]\log(\sqrt{T}\bar{g}_{+,\times}/\epsilon) \right ),
    \end{equation}
    queries to the oracles. Moreover, we may output an $\epsilon$--approximation to the history state $\ket{x_H}$ using
     \begin{equation}
        O\left (  \omega \kappa_V \sqrt{T} \log [\lambda_{+,\times}(T)f^{+,\times}_{H,T}/\epsilon]\log(1/\epsilon) \right ),
    \end{equation}
    queries to the oracles.
    \item If $A$ has negative log-norm, then we may output an $\epsilon$--approximation to the final state $\ket{x_T}= \ket{x(T)}$ using
\begin{equation}
        O\left ( \omega \bar{g}_{+,\times} (T) T^{3/4} \log[\lambda_{+,\times}(T)f^{+,\times}_{H,T}/\epsilon]\log(\sqrt{T}\bar{g}_{+,\times}/\epsilon) \right ),
    \end{equation}
    queries to the oracles. Moreover, we may output an $\epsilon$--approximation to the history state $\ket{x_H}$ using
     \begin{equation}
        O\left (  \omega \sqrt{T} \log [\lambda_{+,\times}(T)/\epsilon]\log(1/\epsilon) \right ),
    \end{equation}
    queries to the oracles.
\end{itemize}
\end{cor}
These complexities improve on prior results, as detailed and summarized in Table~\ref{table:comparisons}.

	\subsection{The No Fast-Forwarding theorem and quantum linear solvers}\label{sec:No-FF}
	
	It should be noted that we may recast any linear ODE into another that is stable with the solutions of each simply related. More precisely, given $\dot{\v{x}} = A \v{x} +\v{b}$ we may choose $z \in \mathbb{R}$ such that $\mu(A-zI)= \mu(A) -z <0$. In particular $z = \|A\| +\delta$ suffices for any $\delta >0$. Then defining $A' = A-zI$ we have that
	\begin{align} 
	    \v{x}(t) &= e^{zt} e^{A' t} \v{x}(0)+ \int_0^t \!\! ds\,\, e^{A's} e^{zs}\v{b}\nonumber \\
	    &=e^{zt} \left ( e^{A' t} \v{x}(0)+ \int_0^t \!\! ds\,\, e^{A's} e^{z(s-t)}\v{b} \right) .
	\end{align}
	In the homogeneous case the solution to the original ODE can be obtained from $\dot{\v{x}} = A' \v{x}$ by rescaling by a scalar, while in the inhomogeneous case $\v{b} \ne \v{0}$ this generates an effective time-dependent driving term, and to account for this we can realize the integrated scaling $e^{-zt}$ term if in the algorithm we increment $\v{b}$ by a factor $e^{-zh}$ at each time-step. In particular, we can consider $A = -iH$, for Hermitian $H$, and transform into an ODE with negative log-norm. Given negative log-norm systems have $O(\sqrt{T})$ complexity, this would seem to conflict with No Fast-Forwarding, which requires that simulation of general quantum dynamics has complexity $\Omega(T)$. However, the error constraint on the $\dot{x} = A'\v{x}$ system is exponential in $T$ since we require $\epsilon'_{\mbox{\tiny TD}} = e^{-zT} \epsilon_{\mbox{\tiny TD}}$, where $\epsilon'_{\mbox{\tiny TD}}$ is the error in the system with negative log-norm. This restores the $\Omega(T)$ complexity.

Finally, we note that since stable systems have $\|e^{At}\|$ upper bounded by an exponentially decreasing function in $t$ we can extend the results of~\cite{an2022theory} in a straightforward manner to obtain a query complexity lower bound for the dynamics in terms of the associated complexity of inverting the stable matrix $A$. More, precisely, as observed in~\cite{an2022theory} the asymptotic dynamics of the inhomogeneous system $\dot{\v{x}} = A \v{x} - \v{b}$  tends to an equilibrium for which $A \v{x} - \v{b} = \v{0}$. This implies that $\v{x}(t) \approx  A^{-1} \v{b}$ for $t$ large and therefore constitutes a linear solver algorithm. This leads to the following complexity lower bound for simulating the dynamics of stable systems.
\begin{theorem}[Extended from~\cite{an2022theory}, Proposition 19.]
    Consider $\dot{\v{x}} = A \v{x} - \v{b}$, and assume that $A$ is stable. Any quantum algorithm outputting the normalized solution state $\ket{x_T}$ up to error $\epsilon$ must have worst-case complexity $\Omega (\min (\log^\gamma (1/\epsilon), T^\gamma))$ assuming the worst-case query complexity for a linear-solver for the equation $A\v{x} = \v{b}$ has complexity $\Omega (\log^\gamma (1/\epsilon))$. 
\end{theorem}
The proof of this is identical to the one presented in~\cite{an2022theory} with the Euclidean log-norm replaced with the more general log-norm that arises from the Lyapunov analysis.

 We also note that in the context of Hamiltonian simulation, various fast-forwarding results are known, see for example Ref.~\cite{gu2021fast} and references therein. As we mentioned before, stability analysis cannot help fast-forwarding Hamiltonian simulation, as this is only marginally stable. However, conversely, one could analyse for what class of stable dynamical systems the stability-induced fast-forwarding can be combined with other fast-forwarding techniques inspired by Hamiltonian simulation to give rise to even faster simulations. We leave this as an interesting line for future investigation.

\section{Cost of quantum linear system solver}
\label{sec:costQLSA}

All the methods analyzed in this work rely on a quantum linear solver algorithm as a subroutine. The best query count upper bounds are currently those presented in Ref.~\cite{jennings2023efficient}, which read as follows
		
		\begin{theorem}[Theorem 1 \cite{jennings2023efficient}]
		\label{thm:querycountsQLSA}
		Consider a system of linear equations $L\v{y}=\v{c}$, with $L$ an $N \times N$ matrix, and assume $L$ and $\v{c}$ have been scaled so that the singular values of $L$ lie in $[1/\kappa, 1]$.  Denote by $\ket{c}$ the normalized state that is proportional to $\sum_i c_i\ket{i}$, and $\ket{L^{-1}c}$ the normalized state proportional to $A^{-1}\ket{c}$.  Assume access to
		\begin{enumerate}
\item A  unitary $U_L$ that encodes the matrix $L/\omega_L$ in its top-left block, for some constant $\omega_L \ge 1$, using a number of qubits equal to $a$,
\item An oracle $U_c$ preparing $\ket{c}$.
\end{enumerate}
Then there is a quantum algorithm that outputs a quantum state $\epsilon$-close in $1$--norm to $\ket{A^{-1}c}$, using an expected $ Q^*_{QLSA}$ calls to $U_{L}$ or $U_{L}^\dag$ and $4Q^*_{QLSA}$  calls to $U_c$ or $U_c^\dag$, where
		\footnotesize
		\begin{align*}
		 Q^*_{QLSA} &\le  \left\{ \frac{ 581\omega_L e}{250}  \sqrt{\kappa^2 +1} \left(  \left[\frac{ 133}{125} + \frac{4}{25\kappa^{1/3}}\right]  \pi \log (2\kappa +3) +1 \right) \right. \\ & \left. + \frac{ 117}{50} \log ^2(2 \kappa +3) \left(\log \left(\frac{ 451 \log ^2(2 \kappa +3)}{\epsilon_L }\right)+1\right)+ \omega_L \kappa \log \frac{32}{\epsilon_L}\right\},
		\end{align*}
		\normalsize
		 for any $\epsilon_L \leq 0.2$ and $\kappa \geq \sqrt{12}$. The success probability is lower bounded by $0.39 - 0.201 \epsilon_L$. If the block-encoding of the $N\times N$ dimensional matrix $L$ requires $a$-auxiliary qubits, then the full QLSA requires $a+ 7 + \lceil \log_2 N\rceil $ logical qubits.
	\end{theorem} 

	The expected query complexity including the failure probability is
	\begin{equation}
		\mathcal{Q}_{QLSA} = \mathcal{Q}^*_{QLSA}/(0.39-0.204 \epsilon_L).
	\end{equation}

  \section{Explicit query counts for ODE-solver}. 
We now present the proof of Theorem~\ref{thm:detailedODEcounts}, stated in the main text, and repeated here for convenience.
\begin{theorem}[Explicit query counts for ODE-solver]
Consider any $N$--dimensional ODE system of the form
\begin{equation}
\dot{\v{x}} = A \v{x} + \v{b},\quad \v{x}(0) = \v{x}^0,
\end{equation}
 specified by a complex matrix $A\in \mathbb{C}^{N\times N}$, vectors $\v{b}, \v{x}^0 \in \mathbb{C}^N$ and $t\in [0,T]$. Choose a timestep $h$ such that $\|A\|h \le 1$ and $Mh=T$ for some integer $M$. Let $(\kappa_P, \mu_P(A),  x_{\max} ,  x_{\min} , x_{\mathrm{rms}} )$ be the associated ODE parameters, as defined in Table~\ref{table:parameters}. 

We assume that we have oracle access to a unitary $U_A$, which is an $(\omega, a, 0)$ block-encoding of $A$, a unitary $U_0$ that prepares the normalized state $\ket{x^0}$ and, if $\v{b} \neq \v{0}$, a unitary $U_b$ that prepares the normalized state $\ket{b}$. Then a quantum state $\epsilon$-close in $1$-norm to the solution state $\ket{x_T}=\ket{x(T)}$, or to the history state $\ket{x_H}$ in Eq.~\eqref{eq:historystate}, 
can be outputted with bounded probability of success, using an expected number
$\mathcal{Q}_{H,T}$ queries to $U_A$ and $4\mathcal{Q}_{H,T}$ queries to $U_0$. In the case where we have $\v{b} \neq \v{0}$ then $4\mathcal{Q}_{H,T}$ queries to $U_b$ are also made. The query count parameter $\mathcal{Q}_{H,T}$ is given as follows:  
\begin{enumerate}[label={(\arabic*)}]
 \item Set the time discretization error $\epsilon_{\mbox{\tiny TD}}=\epsilon/(8f^{+,\times}_{H,T})$, where $f^{\times}_{H,T} = 1$, $f^+_H = \frac{1}{ x_{\mathrm{rms}} }$, $f^+_T =\frac{1}{\|\v{x}(T)\|}$. 
 \item Compute the Taylor truncation order: 
        \begin{equation}
        \label{eq:k}
            k_{+,\times} = \left \lceil \frac{3\log(s_{+,\times})/2+1}{\log[1+ \log(s_{+,\times})/2]}-1 \right \rceil,
        \end{equation} 
        with
        \small
        \begin{align}
    s_\times \geq & \frac{M e^3}{\epsilon_{\mbox{\tiny TD}}} \left(1+ T e^2 \frac{\| \v{b}\|}{ x_{\min} }\right) \nonumber \\
    s_+ \geq & \frac{Me^3  x_{\max} }{\epsilon_{\mbox{\tiny TD}}} \left( 1 +  T e^2 \frac{\|\v{b}\|}{ x_{\max} } \right),
    \end{align}
    \normalsize
\item In the case of the history state output, set $p =0$. In the case of the solution state output and $A$ being stable, compute $M=A/h$ and set 
\small
\begin{align*}
     p =\lceil \sqrt{M}/(k_{+,\times}+1)\rceil (k_{+,\times}+1)
\end{align*}
\normalsize and otherwise
\begin{align}
    p=\lceil M/(k_{+,\times}+1)\rceil (k_{+,\times}+1),
\end{align}
 in the case when $A$ is not stable.
 \item Compute the scaling parameter $\omega_{\tilde{L}} = \frac{1+\sqrt{k_{+,\times}+1}+\omega h}{\sqrt{k_{+,\times}+1} +2}$.
 \item Compute the upper bound on the condition number of~$L$:
            \footnotesize
        \begin{align}
         \nonumber 
            \! \! \! \!   \! \! \! \!   \! \! \!
 \kappa_L \! = \! & \left[   (1+\epsilon_{\mbox{\tiny TD}})^2 \left( 1  + g(k_{+,\times}) \right ) \kappa_P  \left(p \frac{1- e^{2 T\mu_P(A) +2 h\mu_P(A) }}{ 1- e^{2h \mu_P(A) }} \right. \right. \\ \nonumber & \left. \left. + I_0(2) \frac{e^{2\mu_P(A)(T+2h)} + M +1 - e^{2h\mu_P(A) } (2+M)}{(1-e^{2  h\mu_P(A)})^2 } \right) \right. \\ & \left. + \frac{p(p+1)}{2} \! + (p + M k_{+,\times}) (I_0(2)-1) \right]^{\frac{1}{2}} \! \!\! (\sqrt{k_{+,\times}+1}+2)   \label{eq:kappaLbound}
        \end{align}  
        \normalsize
        where $g(k_{+,\times}) = \sum_{s=1}^{k_{+,\times}} \left(s! \sum_{j=s}^{k_{+,\times}} 1/j!\right)^2 \leq ek_{+,\times}$, $I_0(2) \approx 2.2796$. 
          \item Fix $K=(3-e)^2$ for the inhomogeneous case ($\v{b} \neq \v{0})$ and $K=1$ for the homogeneous case $(\v{b} = \v{0})$. For outputting the history state compute the success probability
\begin{equation}
\label{eq:prH}
  \! \! \!  \!\mathrm{Pr}^{+,\times}_H \ge \frac{K}{K-1+I_0(2)} \geq \begin{cases}
    29/500 \quad \v{b}\neq \v{0} \\
     219/500 \quad \v{b} = \v{0},
    \end{cases}
\end{equation}
and for outputting the solution state compute
\begin{align}
  \! \! \!  \! \mathrm{Pr}^{+,\times}_T 
	 \geq \frac{1 }{ (1-\frac{I_0(2)-1}{(p+1)K})  + (M+1) \frac{I_0(2)-1}{(p+1)K} \left(\frac{1+\epsilon_{\mbox{\tiny TD}}}{1-\epsilon_{\mbox{\tiny TD}}}\right)^2 \bar{g}_{+,\times}^2}.
  \end{align}
 \item Compute the error parameter $\epsilon_L = \epsilon\frac{\mathrm{Pr}^{+,\times}_{H,T} }{4+\epsilon}$.
  \item Compute
  \begin{align}
      \mathcal{Q}_{QLSA} = \mathcal{Q}_{QLSA}(\omega_{\tilde{L}}, \kappa_L, \epsilon_L),
  \end{align} the best available query counts for QLSA given the parameters $\omega_{\tilde{L}}$, $\kappa_L$, $\epsilon_L$. In this work we use the closed formula for $\mathcal{Q}_{QLSA}$ from Ref.~\cite{jennings2023efficient},
 $\mathcal{Q}_{QLSA} = \mathcal{Q}^*/{(0.39-0.204 \epsilon_L)}$, where $\mathcal{Q}^*(\omega_{\tilde{L}}, \kappa_L, \epsilon_L) $ equals 
		\scriptsize
		\begin{align*}
		   \left\{ \frac{ 581\omega_L e}{250}  \sqrt{\kappa_L^2 +1} \left(  \left[\frac{ 133}{125} + \frac{4}{25\kappa_L^{1/3}}\right]  \pi \log (2\kappa_L +3) +1 \right)  + \frac{ 117}{50} \log ^2(2 \kappa_L +3) \left(\log \left(\frac{ 451 \log ^2(2 \kappa_L +3)}{\epsilon_L }\right)+1\right)+ \omega_L \kappa_L \log \frac{32}{\epsilon_L}\right\},
		\end{align*}
		\normalsize
  \item Compute $\mathcal{A}(\mathrm{Pr}^{+,\times}_{H,T})$, the query counts for amplitude amplification from success probability $\mathrm{Pr}_{H,T}$ (see, e.g., Ref.~\cite{yoder2014fixed}), or $\mathcal{A}(\mathrm{Pr}^{+,\times}_{H,T}))= 1/\mathrm{Pr}^{+,\times}_{H,T})$ if we sample till success.
  \item Finally, the query count parameter is given as $\mathcal{Q}_{H,T} = \min_{+,\times} (\mathcal{A}\mathcal{Q}_{QLSA} )$, with the minimization taken over the choice of either $+$ or $\times$ for the discretization error scheme.
\end{enumerate}
 The total logical qubit count for the algorithm is $a + 1 3 + \lceil \log_2 [((M+1)(k_{+,\times}+1) +p )N] \rceil$.
\end{theorem}
The proof of this is as follows.
\begin{proof}
    To obtain an explicit query count for the ODE system we make use of the QLSA query count given in Theorem~\ref{thm:querycountsQLSA}. We now provide the analysis for each of the steps in the theorem.
    
 Steps~$(1)$ and $(7)$: The error parameters set in step~$(1)$ and step~$ (7)$ follow from Theorem~\ref{lem:totalerrorbound}. Explicitly, the output $\sigma_{H,T}$ of the quantum algorithm satisfies 
  \begin{align}
  \| \sigma_{H,T} - \ketbra{x_{H,T}}{x_{H,T}} \|_1 \leq    \frac{2 \epsilon_L}{\mathrm{Pr}_{H,T} - \epsilon_L}+ 4 \epsilon_{\mbox{\tiny TD}} f^{+,\times}_{H,T}.
  \end{align}
  In particular, requiring each term on the right-hand side to be equal to $\epsilon/2$ implies that
  \begin{align}
      \epsilon_{\mbox{\tiny TD}}  = \epsilon f^{+,\times}_{H,T}/ 8, \quad \epsilon_L = \frac{\mathrm{Pr}^{\times, +}_{H,T}\epsilon}{4 + \epsilon},
  \end{align}
  and the trace-distance to the ideal state is bounded by $\epsilon$.
  
 Step~$(2)$: The sufficient Taylor truncation order so as to achieve $\epsilon_{\mbox{\tiny TD}}$ is given by Lemma~\ref{lem:sufficientk}.

 Steps~$(3)$ and $(9)$: The idling parameter $p$ is set to zero for the history state output case. For the solution state output, and when $A$ is a stable matrix we see from Theorem~\ref{thm:conditionnumbernoidling} that setting $p$ to scale as $\sqrt{T}$ does not change the $O(\sqrt{T})$ complexity of $\kappa_L$. Similarly, setting $p$ to scale linearly in $T$ does not change the complexity of $\kappa_L$ in the case when $A$ is not stable. We choose the idling parameter such that the complexity of the condition number does not change. Requiring it to be an integer-valued multiple of $k_{+,\times}+1$ means that it suffices to choose $p$ as stated in step~$(3)$, for the two cases of $A$ being stable or not. From Theorem~\ref{thm:successprobfinalstate} we have that $Pr^{+,\times}_T = \Omega (p/(T\bar{g}^2_{+,\times})$ In particular, if in the case of stable $A$ we set $p=M^{1/2}$,
\begin{equation}
    	\mathrm{Pr}^{+,\times}_T 
	 \geq \frac{1 }{ \left (1-\frac{I_0(2)-1}{M^{1/2} K} \right)  + \frac{(M+1)}{M^{1/2}} \frac{I_0(2)-1}{K} \left(\frac{1+\epsilon_{\mbox{\tiny TD}}}{1-\epsilon_{\mbox{\tiny TD}}}\right)^2 \bar{g}_{+,\times} ^2}.
\end{equation}
In the case of $A$ not being stable, we set $p=M$ to obtain
\begin{align}
    \mathrm{Pr}^{+,\times}_T \geq \frac{1 }{ \left (1-\frac{I_0(2)-1}{K(M+1)} \right )  + \frac{I_0(2)-1}{K} \left(\frac{1+\epsilon_{\mbox{\tiny TD}}}{1-\epsilon_{\mbox{\tiny TD}}}\right)^2 \bar{g}_{+,\times} ^2}.
\end{align}
Irrespective of how $\mathrm{Pr}_T$ scales with time, with amplitude amplification, we need $O(1/\sqrt{\mathrm{Pr}^{+,\times}_T })$ queries to the whole ODE linear solver protocol to obtain success probability $O(1)$. We denote the exact number by $\mathcal{A}(\mathrm{Pr}^{+,\times}_T)$, alternatively, repeating until sucess requires an expected number $1/\mathrm{Pr}^{+,\times}_T$ of rounds.

Step~$(4)$: The stated block-encoding scaling parameter $\tilde{\omega}_L $ is obtain from Theorem~\ref{thm:omegavalue}.

Step~$(5)$: The stated upper bound on the condition number of the linear system is obtained from Theorem~\ref{thm:conditionnumbernoidling}.

Step~$(6)$: The stated lower bounds on the success probability for outputting either the history state or the solution state is obtained from Theorem~\ref{thm: successprobhistorystate} and Theorem~\ref{thm:successprobfinalstate}.

Steps~$(8)$ and $(10)$: The expected query count for the quantum linear-solver component is obtained from Theorem~\ref{thm:querycountsQLSA}. The inputs to this are the block-encoding scale factor $\omega_{\tilde{L}}$ from step $(4)$, the condition number upper bound $\kappa_L$ from step $(5)$ and the target error parameter $\epsilon_L$ for the linear solver, given in step $(7)$. This gives the query count to the block-encoding of the linear system matrix, and from Theorem~\ref{thm:omegavalue} this equals the number of queries to the block-encoding $U_A$ for the ODE matrix $A$. However, the post-processing of the output from the QLSA has a finite success of probability, and so using amplitude amplification or repeating until success requires $\mathcal{A}$ rounds of the QLSA before the desired output is obtained. We then minimize the product $\mathcal{A} \mathcal{Q}_{QLSA}$ over the two choices of discretization error conditions, namely the multiplicative case $\times$ and the additive case $+$. This gives a total expected query count to $U_A$ of $\mathcal{Q}_{H,T}$ as stated in step~$(10)$. As stated in Section~\ref{sec:block-encoding}, the block-encodings $U_0$ and $U_b$ are called once each in order to realize the QLSA vector $\ket{c}$, and so from Theorem~\ref{thm:querycountsQLSA} we call them $4$ times as often as $U_A$, which implies $4\mathcal{Q}_{H,T}$ times in total.

To see that the total logical qubit count is $a+12 + \lceil \log_2 [((M+1)(k_{+,\times}+1) +p )N] \rceil$ note that the QLSA uses a total of $\tilde{a}+  7 + \lceil \log_2 \tilde{N} \rceil$ logical qubits, where $\tilde{a}$ is the number of auxiliary qubits needed for the matrix block-encoding of $\tilde{L}$ and $\tilde{N}$ is the dimension of $\tilde{L}$. From Theorem~\ref{thm:omegavalue} we have that $\tilde{a} = a+6$. The dimension of $\tilde{L}$ is obtained from the dimension of $A$ times the total number of time-steps involved, including idling. This involves $(M+1) (k_{+,\times}+1) +p )$ time-steps, and so $\tilde{N} =((M+1)(k_{+,\times}+1) +p )N $. These parameters imply we require  $a+13 + \lceil \log_2 [((M+1)(k_{+,\times}+1) +p )N] \rceil$ logical qubits in total for the ODE-solver, where the choice of $k_{+,\times}$ is determined from the earlier minimization over $+,\times$.
 \end{proof}

\end{document}